\documentclass[11pt]{article}
\usepackage[margin = 1in]{geometry}

\usepackage[inline]{enumitem}
\usepackage{amsmath,amssymb,amsthm,hyperref,color}
\usepackage{blkarray}
\usepackage{multirow}
\hypersetup{colorlinks=true,citecolor=blue, linkcolor=blue,
urlcolor=blue}

\newtheorem{theorem}{Theorem}[section]
\newtheorem{lemma}[theorem]{Lemma}

\newtheorem{corollary}[theorem]{Corollary}
\newtheorem{remark}[theorem]{Remark}

\newtheorem{definition}[theorem]{Definition}

\newcommand{\eps}{\epsilon}

\newcommand{\ra}{\rightarrow}
\newcommand{\R}{\mathbb{R}}

\def\S{\mathcal{S}}
\def\G{\mathcal{G}}
\def\P{\mathbb{P}}
\def\E{\mathbb{E}}
\def\B{B}

\def\eps{\varepsilon}
\def\epsilon{\varepsilon}

\newcommand{\fptas}{\mathsf{FPTAS}}
\newcommand{\fpras}{\mathsf{FPRAS}}

\newcommand{\TreeD}{\mathbb{T}_{\Delta}}
\newcommand{\TreeDell}{\mathbb{T}_{\Delta,\ell}}

\def\a{\ensuremath{\mathbf{a}}}
\def\b{\ensuremath{\mathbf{b}}}

\def\p{\ensuremath{\mathbf{p}}}
\def\qb{\ensuremath{\mathbf{q}}}
\def\rb{\ensuremath{\mathbf{r}}}
\def\cb{\ensuremath{\mathbf{c}}}
\def\r{\ensuremath{\mathbf{r}}}
\def\s{\ensuremath{\mathbf{s}}}
\def\u{\ensuremath{\mathbf{u}}}
\def\w{\ensuremath{\mathbf{w}}}
\def\x{\ensuremath{\mathbf{x}}}
\def\y{\ensuremath{\mathbf{y}}}
\def\z{\ensuremath{\mathbf{z}}}
\def\A{\ensuremath{\mathbf{A}}}
\def\B{\ensuremath{\mathbf{B}}}
\def\E{\ensuremath{\mathbf{E}}}
\def\H{\ensuremath{\mathbf{H}}}
\def\I{\ensuremath{\mathbf{I}}}
\let\OldL\L
\def\L{\ensuremath{\mathbf{L}}}
\def\M{\ensuremath{\mathbf{M}}}
\def\P{\ensuremath{\mathbf{P}}}
\def\S{\ensuremath{\mathbf{S}}}
\def\W{\ensuremath{\mathbf{W}}}
\def\X{\ensuremath{\mathbf{x}}}
\def\Y{\ensuremath{\mathbf{y}}}

\def\T{\ensuremath{\intercal}}
\def\Sc{\ensuremath{\mathcal{S}}}
\def\Tc{\ensuremath{\mathcal{T}}}
\def\Det{\ensuremath{\mathrm{Det}}}

\def\Db{\ensuremath{\mathbf{D}}}
\def\Jb{\ensuremath{\mathbf{J}}}

\def\Pb{\ensuremath{\mathbf{P}}}
\def\Rb{\ensuremath{\mathbf{R}}}
\def\Sb{\ensuremath{\mathbf{S}}}
\def\Tb{\ensuremath{\mathbf{T}}}

\def\Vb{\ensuremath{\mathbf{V}}}
\def\Wb{\ensuremath{\mathbf{W}}}
\def\Zb{\ensuremath{\mathbf{Z}}}

\def\Yc{\ensuremath{\mathcal{Y}}}
\def\Qc{\ensuremath{\mathcal{Q}}}

\def\alphab{\ensuremath{\boldsymbol{\alpha}}}
\def\betab{\ensuremath{\boldsymbol{\beta}}}
\def\gammab{\ensuremath{\boldsymbol{\gamma}}}
\def\deltab{\ensuremath{\boldsymbol{\delta}}}

\def\Gc{\ensuremath{\mathcal{G}}}

\def\Tc{\ensuremath{\mathcal{T}}}
\def\Fc{\ensuremath{\mathcal{F}}}

\newcommand{\norm}[1]{\left\|#1\right\|}
\def\oneb{\ensuremath{\mathbf{1}}}
\def\ones{\ensuremath{\mathbf{1}}}
\def\zeros{\ensuremath{\mathbf{0}}}
\def\Lwt{\textsc{Lwt}}
\def\MLwt{\textsc{MaxLwt}}
\def\Det{\ensuremath{\mathrm{Det}}}
\def\diag{\ensuremath{\mathbf{diag}}}

\newcommand{\sgn}{\operatorname{sgn}}

\title{
Inapproximability for Antiferromagnetic Spin Systems in the Tree Non-Uniqueness Region\thanks{A preliminary version of this paper
appeared in 
{\em Proceedings of the 46th Annual ACM Symposium on Theory of Computing} (STOC), 823-831, 2014.
}}

\author{Andreas Galanis\thanks{University of Oxford,
  Wolfson Building, Parks Road, Oxford, OX1~3QD, UK.
  \texttt{andreas.galanis@cs.ox.ac.uk}.
The research leading to these results has received funding from the European Research Council under
the European Union's Seventh Framework Programme (FP7/2007-2013) ERC grant agreement no. 334828. The paper
reflects only the authors' views and not the views of the ERC or the European Commission. The European Union is not liable for any use that may be made of the information contained therein.
}
\and
 Daniel \v{S}tefankovi\v{c}\thanks{Department of Computer Science, University of Rochester,
Rochester, NY 14627.   \texttt{stefanko@cs.rochester.edu}.
Research supported in part by NSF grant CCF-1318374.}
 \and Eric Vigoda\thanks{School of Computer Science, Georgia
Institute of Technology, Atlanta GA 30332.
 \texttt{vigoda@cc.gatech.edu}.
Research supported in part by NSF grant CCF-1217458.
}
}

\begin{document} 

\maketitle

\let\SOldv\v
\def\v{\ensuremath{\mathbf{v}}}

\begin{abstract}
A remarkable connection has been established for antiferromagnetic 2-spin systems,
including the Ising and hard-core models, showing that the computational
complexity of approximating the partition function for graphs with maximum
degree $\Delta$ undergoes a phase transition that coincides with the
statistical physics uniqueness/non-uniqueness phase transition on
the infinite $\Delta$-regular tree.  Despite this clear picture for 2-spin systems,
there is little known for multi-spin systems.
We present the first analog of the above inapproximability results for multi-spin systems.

The main difficulty in previous inapproximability results was analyzing the behavior of the model on
random $\Delta$-regular bipartite graphs, which served as the gadget in the reduction.
To this end one needs to understand the moments of the partition function.
Our key contribution is connecting:
(i) induced matrix norms, (ii) maxima of the expectation of the partition function, and (iii)
attractive fixed points of the associated tree recursions (belief propagation).
The view through matrix norms allows a simple and generic analysis of the second moment for
any spin system on random $\Delta$-regular bipartite graphs.
This yields concentration results for
any spin system in which one can analyze the maxima of the first moment.
The connection to fixed points of the tree recursions enables an analysis of the
maxima of the first moment for specific models of interest.

For $k$-colorings we prove that for even $k$, in a tree non-uniqueness region
(which corresponds to $k<\Delta$) there is no FPRAS, unless NP=RP, to approximate the number
of colorings for triangle-free $\Delta$-regular graphs.
Our proof extends to the antiferromagnetic Potts model, and, in fact, to every antiferromagnetic
model under a mild condition.

\end{abstract}

%
%
%
%

\thispagestyle{empty}

\newpage

\setcounter{page}{1}

\section{Introduction}

\subsection{Background}

Spin systems are a general framework from statistical physics
that captures classical physics models, including the Ising and Potts models, and
models of particular combinatorial interest,
including $k$-colorings and
the hard-core lattice gas model defined on independent sets.  We define these
combinatorial models more precisely before presenting the context of our results.

The hard-core lattice gas model is an example of a 2-spin system.
For a graph $G=(V,E)$,
configurations of the model are the set $\Omega$ of independent sets of $G$.
The model is parameterized by an activity $\lambda>0$, and a configuration
$\sigma\in\Omega$ is assigned weight $w(\sigma) = \lambda^{|\sigma|}$.  The Gibbs
distribution is $\mu(\sigma) = w(\sigma)/Z$ where the normalizing factor is known as the
partition function and is defined as $Z=\sum_{\sigma \in\Omega} w(\sigma)$.
In the hard-core model the spins correspond to occupied/unoccupied.
Multi-spin systems are models with more than 2 spins, an example being
the $k$-colorings problem.  In the colorings problem,
for a graph $G=(V,E)$,
 configurations are the set $\Omega$ of assignments of a set of $k$ colors
to vertices so that neighboring vertices receive different colors.
The Gibbs distribution is the uniform distribution over $\Omega$, and in this case
the partition function $Z=|\Omega|$ is the number of $k$-colorings in $G$.

The hard-core model and colorings are examples
of antiferromagnetic systems -- neighboring vertices ``prefer''
to have different spins.
In contrast, in ferromagnetic systems
neighboring spins tend to align.  We defer the formal
definition of antiferromagnetic spin systems to
Section \ref{sec:generalresults} (see Definition~\ref{def:antiferromagnetic}), where we also discuss how our results
extend to general spin systems.

The focus of this paper is the computational complexity of computing
the partition function.  Exact computation of the partition function is typically \#P-complete,
even for very restricted classes of graphs \cite{Greenhill}.
Hence our focus is on the existence of a fully-polynomial
approximation scheme -- either a deterministic $\fptas$ or randomized $\fpras$ --
for estimating the partition function.
For any spin system, (approximate) sampling from the Gibbs distribution implies an $\fpras$
for estimating the partition function, and hence our hardness results also apply to
the associated sampling problem.

The computational complexity of approximating the partition function is now
well-understood for 2-spin systems, such as the Ising and hard-core models.
For all ferromagnetic 2-spin systems, there is an $\fpras$ for estimating the partition
function \cite{GJP}.   The picture is more intricate (and fascinating)
for antiferromagnetic 2-spin systems.
We will detail the picture after introducing the statistical physics
notion of a phase transition.

Let $\TreeDell$ denote the complete $\Delta$-regular tree of depth $\ell$ with root $r$.
The question of interest is whether or not we can fix a configuration on the leaves of $\TreeDell$
so that the root is influenced by this boundary configuration in the limit $\ell\rightarrow\infty$.
For the example of colorings, fix a coloring $\sigma_\ell$ of the leaves (such that
there is at least one coloring of the rest of the tree that is consistent with $\sigma_\ell$).
Look at a random coloring of the tree $\TreeDell$ conditioned
on the leaves having coloring $\sigma_\ell$.
For all sequences $(\sigma_\ell)$ of fixed leaf colorings, if in the limit $\ell\rightarrow\infty$,
 the marginal at the root is uniform over the $k$ colors, then we say uniqueness holds,
 and otherwise we say
non-uniqueness holds.
 (The terminology comes from statistical physics where the focus is on the set of
infinite-volume Gibbs measures, see \cite{Georgii}.)

For the hard-core model the critical activity
is $\lambda_c(\Delta) = (\Delta-1)^{\Delta-1}/(\Delta-2)^{\Delta}$ \cite{Kelly}.
\cite{Weitz} presented an $\fptas$ for estimating the partition function
in the tree uniqueness region (i.e., when $\lambda<\lambda_c(\Delta)$).
On the other side, \cite{Sly} (extended in \cite{SS,GGSVY,GSV:arxiv}) proved that, unless NP=RP,
it is NP-hard to obtain an $\fpras$ for $\Delta$-regular graphs in the tree non-uniqueness region (i.e., when $\lambda>\lambda_c(\Delta)$).
These results were extended to all 2-spin antiferromagnetic models by \cite{LLY} (see also \cite{SST})
and \cite{SS}.
For 2-spin antiferromagnetic models,
this establishes a beautiful picture connecting the computational complexity of approximating the partition function
 to statistical physics phase transitions in the infinite tree.

\subsection{Main Results}

The picture for multi-spin systems (systems with $q>2$ possible spins for vertices) is much less clear;
the above approaches for 2-spin systems do not extend to multi-spin models in a straightforward manner.
We aim to establish the analog of the
above inapproximability results for the colorings problem, namely, NP-hardness in the tree non-uniqueness region.
 Our techniques and results generalize
to a broad class of antiferromagnetic spin systems.

\subsubsection{Results for Colorings}

For the colorings problem, even understanding the uniqueness threshold  is challenging.
\cite{Jonasson} established
uniqueness when $k\geq \Delta+1$, and it is easy
to show non-uniqueness when $k\leq\Delta$ since a fixed coloring on the leaves can ``freeze'' the internal coloring.
For 2-spin systems uniqueness can be characterized by the existence of
multiple solutions of a certain system of equations \eqref{kkrtko}, called
tree recursions, see Section \ref{sec:treefirstmatrix} for additional explanation.
In statistical physics terminology the solutions to these equations
correspond to semi-translation invariant measures on the infinite tree $\TreeD$.
  For colorings the uniqueness threshold and the semi-translation invariant
  uniqueness threshold no longer coincide.
In particular, \cite{BW} established, for semi-translation
invariant measures, uniqueness when $k\geq\Delta$ and non-uniqueness when $k<\Delta$.

We prove, for even $k$, that it is NP-hard to approximate the number of colorings (in other words,
NP-hard to approximate the partition function)
when there is non-uniqueness of semi-translation invariant Gibbs measures on $\TreeD$, i.e., 
when $k<\Delta$.   Moreover, our result proves hardness for the class of triangle-free $\Delta$-regular graphs.
Hence, our result is particularly interesting in the region $k=\Omega(\Delta/\log{\Delta})$ since
a seminal result of \cite{Johansson} (see also \cite{MRbook}) shows that all
triangle-free graphs are colorable with $O(\Delta/\log \Delta)$ colors. His proof, which uses the nibble method and the Lov\'{a}sz Local Lemma, can be made algorithmic using the constructive proof of \cite{MoserT}.
For general graphs with maximum degree $\Delta$, the interesting region is
$k= \Delta - O(\sqrt{\Delta})$, since \cite{MR} showed,
for sufficiently large constant $\Delta$, a polynomial-time algorithm to determine
if a graph with maximum degree $\Delta$ is $k$-colorable when $k \geq \Delta - \sqrt{\Delta} + 3$. We note that most parts of the proof extend to the odd $k$ case as well, modulo the technical condition described in the end of Section~\ref{sec:generalresults}.

Here is the formal statement of our inapproximability result for colorings.

\begin{theorem}\label{thm:colorings}
For all even $k\geq4$, all $\Delta\geq 3$,
for the $k$-colorings problem, when $k<\Delta$,
unless \emph{NP=RP}, there is no $\fpras$ that approximates
the partition function for triangle-free $\Delta$-regular graphs.
Moreover, there exists $\eps=\eps(k,\Delta)$ such that, unless \emph{NP=RP},
one cannot approximate the partition function within
a factor $2^{\eps n}$ for triangle-free $\Delta$-regular graphs (where $n$ is the
number of  vertices).
\end{theorem}

\subsubsection{Results for Antiferromagnetic Potts}

Our result also extends to the antiferromagnetic Potts model.
In the $q$-state Potts model there is a parameter $B>0$ which
corresponds to the ``temperature'' and controls the strength of the
interactions along an edge.  For a graph $G=(V,E)$,
the set $\Omega$ of configurations
are assignments $\sigma$  where $\sigma:V\rightarrow [q]$.
Each configuration has a weight $w(\sigma) = B^{m(\sigma)}$
where $m(\sigma)$ is the number of monochromatic edges in $\sigma$.
The Gibbs distribution is $\mu(\sigma) = w(\sigma)/Z$ where
$Z=\sum_{\tau\in\Omega} w(\tau)$ is the partition function.
The case $B> 1$ is the ferromagnetic Potts model, and $B<1$ is the
antiferromagnetic Potts model.  Colorings corresponds to the $B= 0$ case,
and the Ising model is the $q=2$ case.

The uniqueness/non-uniqueness threshold for the infinite tree $\TreeD$ is not known for the antiferromagnetic
Potts model.    We prove that the uniqueness/non-uniqueness
threshold for semi-translation invariant Gibbs measures on  $\TreeD$
occurs at $B_c(\Delta) = \frac{\Delta-q}{\Delta}$.
We believe this threshold coincides with the uniqueness/non-uniqueness threshold,
unlike in the case of colorings.
We prove, for even $q$, that approximating the partition function is NP-hard in
the non-uniqueness region for semi-translation invariant measures.

\begin{theorem}\label{thm:Potts}
For all even $q\geq 4$, all $\Delta\geq 3$,
for the antiferromagnetic $q$-state Potts model, for all $B<\frac{\Delta-q}{\Delta}$,
unless \emph{NP=RP}, there is no $\fpras$ that approximates
the partition function for triangle-free $\Delta$-regular graphs.
Moreover, there exists $\eps=\eps(q,\Delta)$ such that, unless \emph{NP=RP},
one cannot approximate the partition function within
a factor $2^{\eps n}$ for triangle-free $\Delta$-regular graphs (where $n$ is the
number of  vertices).
\end{theorem}

\subsubsection{Results for General Antiferromagnetic Models}\label{sec:generalresults}

Our approach applies in much more generality
and yields inapproximability of the partition function for any
antiferromagnetic model when there is non-uniqueness of semi-translation
invariant measures on $\TreeD$ and mild additional conditions.  


We first need to define general antiferromagnetic models.
A general  $q$-spin system is
specified by a symmetric $q\times q$ interaction matrix $\B=(B_{ij})_{i,j\in[q]}$ with non-negative
entries, which specify the strength of the interaction between the spins.
For example, the interaction matrix for the Potts
model has off-diagonal entries equal to  $1$ and its diagonal entries  equal to $B$.
 For a finite undirected
graph $G=(V,E)$, a $q$-spin system is a probability distribution $\mu_G$ over the space $\Omega_G$ of all
\textit{configurations}, i.e., spin assignments $\sigma:V\rightarrow [q]$. The weight of a 
configuration $\sigma\in\Omega_G$ is  the product of neighboring spin interactions, that is,
\begin{equation*}
w_G(\sigma) = \prod_{(u,v)\in E}B_{\sigma(u)\sigma(v)}.
\end{equation*}
The Gibbs distribution $\mu_G$ is defined as $\mu_G(\sigma) = w_G(\sigma)/Z_G$ where
the partition function $Z_G$ is $Z_G=\sum_{\sigma\in\Omega_G} w_G(\sigma)$.  We drop the subscript $G$ when the graph
under consideration is clear.

We use the following definition of antiferromagnetic models 
in terms of the signature of the interaction matrix $\B$, i.e., the signs  of its eigenvalues. 
The interaction matrix $\B$ is assumed to be symmetric and have non-negative entries. These are standard assumptions since we are interested in undirected graphs and the Gibbs distribution should be a probability distribution. W.l.o.g., we will also assume that $\B$ is irreducible. Otherwise, 
by a suitable permutation of the spins, $\B$ can be put into block diagonal form (which coincides with the normal form of the reducible $\B$) where each of the blocks is either irreducible or zero. Effectively, this says that the original spin model can be studied  by considering the induced sub-models of each block which correspond to irreducible  symmetric matrices (where our results apply). For connected graphs $G$, the partition function for the original model is simply the sum of the partition functions of each sub-model.   

We are now ready to give the definition of antiferromagnetism we use.

\begin{definition}\label{def:antiferromagnetic}
Let $\B$ be the interaction matrix of a $q$-state spin system.
Since $\B$ is symmetric all of its eigenvalues are real.
Also note that it has non-negative entries and by irreducibility, the Perron-Frobenius theorem implies that one 
of the eigenvalues of $\B$ with the largest magnitude is positive and simple, i.e., the associated eigenspace is one-dimensional.
The model is called {\em antiferromagnetic}
if all the other eigenvalues are negative. Note that no eigenvalue is allowed to be zero and hence $\B$ is regular.
\end{definition}

The above definition generalizes antiferromagnetism for 2-spin systems (see \cite{GJP,LLY,SS}), and
captures colorings as well as the antiferromagnetic region for the Potts models.
Moreover, the above definition seems natural in that it implies that
neighboring vertices prefer to have different spin assignments (see Corollary \ref{col:antiferromagnetic} in Section~\ref{sec:antiferromagnetic}). 
Another nice feature of Definition~\ref{def:antiferromagnetic} is that it  does not depend on the presence of external fields. Specifically, for $\Delta$-regular graphs, any external field can be pushed into the interaction matrix $\B$ with a congruence transformation of the matrix $\B$. The resulting interaction matrix, by Sylvester's law of inertia, has the same number of positive, zero and negative eigenvalues and in particular remains antiferromagnetic.

We conclude this discussion by pointing out that some of our results for general models are more easily stated when $\B$ is further assumed to be aperiodic. We shall refer to such matrices $\B$ (irreducible and aperiodic) as \emph{ergodic}. Note that if $\B$ is periodic, its period must be two, since $\B$ is symmetric. Such a model is only interesting on bipartite graphs (otherwise the partition function is zero). Definition~\ref{def:antiferromagnetic} implies that the interaction matrix $\B$ of an antiferromagnetic model is ergodic  whenever $q\geq3$ (note that it is trivial to compute the partition function on periodic models with $q=2$).

We need several additional definitions concerning the moments of the partition function.
For antiferromagnetic models on a random $\Delta$-regular bipartite graph $G=(V,E)$
with bipartition $V=V_1\cup V_2$, the
goal is to understand the Gibbs distribution $\mu_G$ by looking at the
distribution of spin values in $V_1$ and $V_2$. Let $n=|V_1|=|V_2|$.
For a configuration $\sigma:V\rightarrow[q]$, we shall denote the set of
vertices assigned spin $i$ by $\sigma^{-1}(i)$. Denote by $\triangle_q$ the simplex
$\triangle_{q}=\{(x_1,x_2,\hdots,x_q)\in \mathbb{R}^q\,|\, \mbox{$\sum^q_{i=1}$}\,x_i=1\mbox{ and } x_i\geq 0\mbox{ for } i=1,\hdots,q\}$. For $\alphab,\betab\in \triangle_q$,
let
\[ \Sigma^{\alphab,\betab} = \left\{\sigma:V\rightarrow\{1,\hdots,q\}\,\big|\,  |\sigma^{-1}(i)\cap V_1| = \alpha_i n,\, |\sigma^{-1}(i)\cap V_2| = \beta_i n \mbox{ for } i=1,\hdots,q\right\},
\]
that is, configurations in $\Sigma^{\alphab,\betab}$ assign $\alpha_in$ and $\beta_i n$ vertices in $V_1$ and $V_2$ the spin value $i$,
respectively\footnote{Technically we need to define
$\Sigma^{\alphab,\betab} =
\left\{\sigma:V\rightarrow[q]\,\big|\,  |\sigma^{-1}(i)\cap
V_1| = \hat{\alpha}_i,\, |\sigma^{-1}(i)\cap V_2| = \hat{\beta}_i
\mbox{ for } i\in [q]\right\}$,
where $\{\hat{\alpha}_i\}$  are $\{\alpha_i n\}$ rounded in a
canonical fashion so that their sum is preserved (for example using ``cascade rounding") and in the same way
 $\{\hat{\beta}_i\}$  are $\{\beta_i n\}$ rounded.}. We will be interested in the total weight
$Z^{\alphab,\betab}_G$ of configurations in
$\Sigma^{\alphab,\betab}$, namely
\begin{equation*}
Z^{\alphab,\betab}_G=\mbox{$\sum_{\sigma\in
\Sigma^{\alphab,\betab}}$}\, w(\sigma).
\end{equation*}
We study $Z^{\alphab,\betab}_G$ by looking at the moments
$\E_\G[Z^{\alphab,\betab}_G]$ and
$\E_{\G}[(Z^{\alphab,\betab}_G)^2]$, where the expectation is over the  distribution of the random $\Delta$-regular bipartite
graph, from hereon denoted by $\Gc$.

For $\alphab,\betab\in\triangle_q$, denote the leading term of the first and second moments as:
\begin{eqnarray}
\label{def:Psi_1}
\Psi_1(\alphab,\betab)  & = & \Psi_1^{\B}(\alphab,\betab) := \lim_{n\rightarrow \infty}\frac{1}{n}\log\E_{\G}\big[Z^{\alphab,\betab}_G\big].
\\
\label{def:Psi_2}
\Psi_2(\alphab,\betab) & = & \Psi_2^{\B}(\alphab,\betab)  := \lim_{n\rightarrow \infty}\frac{1}{n}\log\E_{\G}\left[\left(Z^{\alphab,\betab}_G\right)^2\right].
\end{eqnarray}

We will refer to $\alphab,\betab$ that maximize $\Psi_1$ as {\em dominant phases}.  
Moreover, we say that a
dominant phase $(\alphab,\betab)$ is  {\em Hessian dominant} if the
Hessian of $\Psi_1$ at $(\alphab,\betab)$ is negative definite. (Note this is a sufficient condition
for $\alphab,\betab$ to be a local maximum.)
In the uniqueness region there is a unique dominant phase and it has $\alphab=\betab$.
In contrast, for 2-spin antiferromagnetic models and for colorings
in the semi-translation non-uniqueness region,
 the dominant phases have $\alphab\neq\betab$, and one expects this would hold for all
 antiferromagnetic models.  In our reduction we will need this additional condition that 
 the dominant phases are not symmetric (i.e., $\alphab\neq\betab$).

Our main technical result relates the second moment to the first moment, for any model on random bipartite regular
graphs.

\begin{theorem}\label{thm:second-moment}
For any spin system, for all $\Delta\geq 3$,
\[\max_{\alphab,\betab} \Psi_2(\alphab,\betab) = 2\max_{\alphab,\betab}\Psi_1(\alphab,\betab).\]
\end{theorem}

Crucially, Theorem~\ref{thm:second-moment} implies that
$\Psi_2(\alphab,\betab)=2\Psi_1(\alphab,\betab)$ for dominant phases, 
which is key for our arguments, since it will eventually allow us to 
find the asymptotic distribution of the random variables $Z_G^{\alphab,\betab}$ (as $n\rightarrow\infty$). We do this by applying the so-called small subgraph conditioning method. The asymptotic convergence is utilized to prove the properties of the gadget we use in the reduction. The gadget is a slight modification of a random $\Delta$-regular bipartite graph and its properties are described in Section~\ref{sec:slysunstuff}. The precise formulation of these properties does not matter at this stage, but rather that we can prove them when the dominant phases $(\alphab,\betab)$ satisfy the following conditions: (i) each dominant phase is Hessian dominant, (ii) the dominant phases are permutation symmetric, i.e., obtainable from one another by a suitable permutation of the set of spins (we clarify here  that the permutations must be automorphisms of the interaction matrix $\B$)\footnote{\label{foot:permutation}More precisely, the permutation symmetric property can be stated as follows: for any two dominant phases, say $(\alphab_1,\betab_1)$ and $(\alphab_2,\betab_2)$, there exists a $q\times q$ permutation matrix $\P$ such that $\B=\P\B\P^{\T}$ and $(\alphab_1,\betab_1)=(\P\alphab_2,\P\betab_2)$ or $(\alphab_1,\betab_1)=(\P\betab_2,\P\alphab_2)$. In other words, the dominant phases can be obtained from each other by interchanging $\alphab$ and $\betab$, by permuting the spins in a way that $\B$ is left invariant, or a combination of the previous two operations.},   
(iii) each dominant phase $(\alphab,\betab)$ has  $\alphab\neq \betab$. Condition (iii) implies that the model is in the  non-uniqueness region of $\TreeD$ and, further, that a typical configuration in the Gibbs distribution of the random graph is ``unbalanced" between the two sides, which allows to encode a CSP (in our case \textsc{Max-Cut}). Condition (i) ensures the asymptotic convergence of $Z_G^{\alphab,\betab}$. Condition (ii) ensures that  the asymptotic distribution of $Z_G^{\alphab,\betab}$ is identical for all the dominant phases.

We want to remark why the permutation symmetry condition arises
naturally. A generic multi-spin system in the semi-translational
non-uniqueness region will have exactly two maxima of $\Psi_1$ and
hardness (assuming $\mathrm{NP=RP}$) follows easily. Models coming from
statistical physics (for example, Potts model or Widom-Rowlinson
model) are not generic since they usually come with permutation
symmetries of the same type as condition (ii) in the previous
paragraph. (The symmetries make the hardness result more difficult to
state and prove.) 

We now state our general inapproximability result.

\begin{theorem}\label{thm:general-inapprox}
Let $q\geq 2, \Delta\geq 3$.  For an antiferromagnetic $q$-spin system whose interaction matrix $\B$ is ergodic,
if the dominant phases $(\alphab,\betab)$ of $\Psi_1$  
are permutation symmetric
and all of them are
Hessian dominant and satisfy $\alphab\neq \betab$, then, unless \emph{NP=RP}, there is no $\fpras$
for approximating the partition function for triangle free $\Delta$-regular graphs. Moreover, there exists $\eps=\eps(q,\Delta)$ such that, unless \emph{NP=RP},
one cannot approximate the partition function within
a factor $2^{\eps n}$ for triangle-free $\Delta$-regular graphs (where $n$ is the
number of  vertices).
\end{theorem}

We remark here that, whenever the hypotheses of Theorem~\ref{thm:general-inapprox} are satisfied, the spin system with interaction matrix $\B$ is in the tree non-uniqueness region of $\TreeD$, see Section~\ref{sec:treefirstmatrix} for more details. However, the reverse direction is not necessarily true, that is, an antiferromagnetic spin system in the tree non-uniqueness region of $\TreeD$ does not necessarily have multiple dominant phases, an example is the $k$-colorings model when $k=\Delta$ (see Theorem~\ref{thm:fase} below).

For illustrative purposes, we first note that the inapproximability results for antiferromagnetnic 2-spin systems in the tree non-uniqueness region \cite{Sly,SS,GSV:arxiv} follow as corollaries of Theorem~\ref{thm:general-inapprox}. In particular, for antiferromagnetic 2-spin systems it is well known that for any $\Delta\geq 3$, in the non-uniqueness region of $\TreeD$, the maximizers of $\Psi_1$ are exactly two pairs $(\alphab,\betab)$ and $(\betab,\alphab)$ with $\alphab\neq\betab$. Note that these two dominant phases satisfy trivially the permutation symmetric property. Moreover, it can also be verified that they are Hessian dominant and hence the hypotheses of Theorem \ref{thm:general-inapprox} are satisfied.

As a more indicative application of Theorem~\ref{thm:general-inapprox}, let us deduce Theorems~\ref{thm:colorings} and~\ref{thm:Potts}. To do this, we need the following theorem (proved in Section~\ref{sec:phase-diagram}) which describes the dominant phases for the colorings and antiferromagnetic Potts models.

\begin{theorem}\label{thm:fase}
Let $q\geq 3$, $0\leq B< 1$ and $\Delta\geq 3$. For the antiferromagnetic $q$-state Potts model with parameter $B$ on a random $\Delta$-regular bipartite graph (note that the $k$-colorings model corresponds to $B=0$ and $q=k$ in the following), it holds that
\begin{enumerate}
\item \label{itt:uniqueness} When $B\geq\frac{\Delta-q}{\Delta}$, there is a unique dominant phase $(\alphab,\betab)$ which satisfies $\alphab=\betab$.
\item \label{itt:dominant} For all even $q\geq 4$, for all $\Delta\geq 3$, when $0\leq B<\frac{\Delta-q}{\Delta}$, the dominant phases $(\alphab,\betab)$ are in one-to-one correspondence with subsets $T\subseteq[q]$ with $|T|=q/2$. Moreover, there exist $a(q,\Delta,B),b(q,\Delta,B)$ with $a\neq b$ such that for $T\subseteq[q]$ with $|T|=q/2$, the dominant phase $(\alphab,\betab)$ corresponding to $T$ satisfies
\begin{equation}\label{eq:fase}
\begin{aligned} 
\alpha_i&=a\mbox{ if }i\in T,\quad\alpha_i=b\mbox{ if }i\notin T,\\ 
\beta_i&=b\mbox{ if } i\in T,\quad\beta_i=a\mbox{ if } i\notin T.
\end{aligned}
\end{equation}
Moreover, the dominant phases are Hessian. 
\end{enumerate} 
\end{theorem}

\begin{proof}[Proof of Theorems~\ref{thm:colorings} and~\ref{thm:Potts}]
We verify  the hypotheses of Theorem~\ref{thm:general-inapprox}. Equation~\eqref{eq:fase} of Theorem~\ref{thm:fase} establishes that the dominant phases $(\alphab,\betab)$ are permutation symmetric and each of them satisfies $\alphab\neq \betab$. Thus, the hypotheses of Theorem~\ref{thm:general-inapprox} hold in the regime $q<\Delta$ and $0\leq B<\frac{\Delta-q}{\Delta}$. 
\end{proof}

Note that the restriction of even $k,q$ in Theorems~\ref{thm:colorings} and~\ref{thm:Potts}, respectively, is a technical one and  comes from the second part of Theorem~\ref{thm:fase}. For odd $q$,  we are unable to establish whether the dominant phases are supported on vectors with two or three different entries, see Section~\ref{sec:phase-diagram} for more details. Classifying the dominant phases for odd $q$ would also extend the inapproximability results of Theorems~\ref{thm:colorings} and~\ref{thm:Potts}. 

\subsection{Proof Approach}

The key gadget in the inapproximability results for 2-spin models is a random $\Delta$-regular
bipartite graph.  The rough idea for the hard-core model is that in the tree non-uniqueness region,
on a random $\Delta$-regular bipartite graph, an independent set from the Gibbs distribution is
``unbalanced'' with high probability
(the fraction of occupied vertices in the two parts of the bipartition differ by a constant).
To analyze random regular bipartite graphs,
the original inapproximability result
of \cite{Sly} relied on a second moment analysis of \cite{MWW},
which
Sly called a technical tour-de-force.  The optimization at the heart of that analysis was difficult
enough that his result only held for $\lambda$ close to the uniqueness threshold.

We present a new approach for the associated optimization problem which is at the heart
of the second moment analysis.  Our approach yields
a simple, short analysis that holds for {\em any model} on random $\Delta$-regular bipartite graphs.
The key idea is to define a new function $\Phi$, which is represented as an induced matrix norm,
and has the same critical points
as the first moment.  We can then use the fact that induced matrix norms are
multiplicative over tensor product to analyze the second moment.

\subsection{Paper Outline}
In Section~\ref{sec:derivations} we derive some basic expressions for the first and second moments.
Then in Section \ref{sec:secondmomentanalysis} 
we analyze the second moment using matrix norms and thereby prove Theorem \ref{thm:second-moment}. In Section~\ref{sec:treefirstmatrix}, we analyze  the maxima of the function $\Psi_1$.  There, we further prove a connection between local maxima of $\Psi_1$  and stable fixpoints of the so-called tree recursions which is used in later sections.

The reduction for the inapproximability results uses an intermediate problem, which we call the ``phase labeling problem". Our inapproximability results hinge on showing  that the phase labeling problem is hard to approximate. In Section~\ref{sec:reduction}, we give the main elements of this reduction for the colorings model to introduce the relevant concepts. The hardness of approximating the phase labeling problem for general antiferromagnetic models is proved in Section~\ref{sec:generalreduction}, where we also fill in the details which were omitted in the simplified exposition for the colorings model. 

We show how the phase labeling problem reduces to the approximation of the partition function in Section~\ref{sec:slysunstuff}, based on arguments in \cite{SS}. The reduction uses  gadgets whose existence and construction are based on a slight variation of the random $\Delta$-regular bipartite graph distribution. At this point, to establish the properties of the gadgets, we use the small subgraph conditioning method. The application of the method is fairly standard though technically intensive due to its use of precise asymptotics for  the first and second moments. The technical details of applying the method in our case are given in Appendix~\ref{sec:small-graph}, while the asymptotics for  the first and second moments are derived in Appendix~\ref{sec:momentasymptotics}. 

The proof of our general inapproximability result (Theorem~\ref{thm:general-inapprox}) is given in Section~\ref{sec:phaseproblem}. We saw in Section~\ref{sec:generalresults} how to  deduce the inapproximability results for the colorings and Potts models (Theorems~\ref{thm:colorings} and~\ref{thm:Potts}) from Theorem~\ref{thm:general-inapprox} using the classification of the dominant phases in Item~\ref{itt:dominant} of Theorem~\ref{thm:fase}. The proof of Item~\ref{itt:dominant} in Theorem~\ref{thm:fase} is given in Section~\ref{sec:phase-diagram}.

Finally, in Appendix~\ref{sec:semi-uniqueness}, we extend the argument of \cite{BW} to prove Item~\ref{itt:uniqueness} of Theorem~\ref{thm:fase}, that is, show uniqueness for semi-translation invariant Gibbs measures for the antiferromagnetic Potts model when $B\geq (\Delta-q)/\Delta$.

\section{Expressions for the first and second moments}
\label{sec:derivations}

In this section we derive the expressions for the first and second moments of $Z_G^{\alphab,\betab}$ and, in particular, the expressions for $\Psi_1$ and $\Psi_2$.

Let $\G_n(\Delta)$ be the probability distribution over bipartite graphs with $n+n$
vertices formed by taking the union of $\Delta$ random perfect matchings. We will use the simplified notation $\G_n:=\Gc_n(\Delta)$ or even $\G:=\Gc_n(\Delta)$ when $n$ is clear from context. Strictly speaking, this distribution is over bipartite multi-graphs.  
However, since our results hold asymptotically almost surely (a.a.s.) over
$\G_n$, as noted in \cite{MWW}, 
by contiguity arguments they also hold a.a.s. for the uniform distribution over bipartite
$\Delta$-regular graphs. For a complete account of contiguity, we refer the reader to \cite[Chapter 9]{JLR}.

Let $G\sim\G$. We will denote the two sides of the bipartition of $G$ as $V_1,V_2$. We first compute the first moment $\E_{\G}[Z^{\alphab,\betab}_G]$. For $\sigma\in
\Sigma^{\alphab,\betab}$ and a uniform matching between $V_1$ and
$V_2$, let $x_{ij}$ denote the number of edges matching vertices
in $\sigma^{-1}(i)\cap V_1$ and $\sigma^{-1}(j)\cap V_2$.  Under
the convention that $0^0\equiv 1$, we then have
\begin{multline}
\E_{\G}[Z^{\alphab,\betab}_G]
=\binom{n}{\alpha_{1}n,\hdots,\alpha_{q}n}\binom{n}{\beta_{1}n,\hdots,\beta_{q}n}
\\
\times \bigg(
\sum_{\X}
\frac{   \prod_{i}\binom{\alpha_{i}n}{x_{i1}n,\hdots,x_{iq}n}\prod_{j}\binom{\beta_{j}n}{x_{1j}n,\hdots,x_{qj}n}\prod_{i,j}B^{nx_{ij}}_{ij}  }
{   \binom{n}{x_{11}n,\hdots,x_{qq}n} }
\bigg)^\Delta,
\label{eq:firstmoment}
\end{multline}
where the sum ranges over $\X=(x_{11},\hdots,x_{qq})$ with $n\X\in
\mathbb{Z}^{q^2}$ satisfying the following constraints:
\begin{equation}
\label{eq:constraintfirst}
\begin{gathered}
\begin{aligned}
\mbox{$\sum_{j}$}\,x_{ij}&=\alpha_i& & \big(\forall i\in[q]\big),& \mbox{$\sum_{i}$}\,x_{ij}&=\beta_j& &\big(\forall j\in[q]\big),\\
\end{aligned}\\
x_{ij}\geq 0\ \ \big(\forall(i,j)\in[q]^2\big).
\end{gathered}
\end{equation}
The first line in \eqref{eq:firstmoment} accounts for the
cardinality of $\Sigma^{\alphab,\betab}$, while the second line is
$\E_\G[w_G(\sigma)]$ for an arbitrary
$\sigma\in\Sigma^{\alphab,\betab}$. Since the weight of a
configuration is multiplicative over the edges and the matchings
are independent, $\E_\G[w_G(\sigma)]$ is the $\Delta$-power of the
expected contribution of a single matching. The latter is
completely determined by $\X$ and is equal to
$\prod_{i,j}B^{x_{ij}}_{ij}$, scaled by the probability that the
matching induces the prescribed $\X$.

We next calculate the second moment of $Z^{\alphab,\betab}_G$. To
do this, for $(\sigma_1,\sigma_2)\in\Sigma^{\alphab,\betab}\times
\Sigma^{\alphab,\betab}$, we need to compute
$\E_\G[w_G(\sigma_1)w_G(\sigma_2)]$. Let $\gamma_{ik}=
|\sigma^{-1}_1(i)\cap \sigma_2^{-1}(k)\cap V_1|/n$,
$\delta_{jl}=|\sigma^{-1}_1(j)\cap \sigma_2^{-1}(l)\cap V_2|/n$. The
vectors $\gammab$ and $\deltab$ capture the overlap of
configurations in $V_1$ and $V_2$, respectively. For a uniform
matching between $V_1$ and $V_2$, let $y_{ikjl}$ denote the number
of edges matching vertices in $\sigma^{-1}_1(i)\cap
\sigma_2^{-1}(k)\cap V_1$ and
$\sigma^{-1}_1(j)\cap\sigma^{-1}_2(l)\cap V_2$ (scaled by $n$). Under the
convention $0^0\equiv 1$, we then have
\begin{multline}
\label{eq:secondmoment}
\E_{\G}[(Z^{\alphab,\betab}_G)^2]
= 
\sum_{\gammab,\deltab}
\binom{n}{\gamma_{11}n,\hdots,\gamma_{qq}n}
\binom{n}{\delta_{11}n,\hdots,\delta_{qq}n}
\\
 \times \bigg( \sum_{\Y}
\frac{
\prod_{i,k}
\binom{\gamma_{ik}n}{y_{ik11}n,\hdots,y_{ikqq}n}\prod_{j,l}
\binom{\delta_{jl}n}{y_{11jl} n,\hdots,y_{qqjl} n}
\prod_{ikjl}(B_{ij}B_{kl})^{ny_{ikjl}}
} 
{
\binom{n}{y_{1111}n,\hdots,y_{qqqq}n}
}
 \bigg)^\Delta, 
\end{multline}
where the sums range over
$\gammab=(\gamma_{11},\hdots,\gamma_{qq})$,
$\deltab=(\delta_{11},\hdots,\delta_{qq})$,
$\Y=(y_{1111},\hdots,y_{qqqq})$ with $n\gammab,n\deltab\in
\mathbb{Z}^{q^2}$ and $n\Y\in \mathbb{Z}^{q^4}$ satisfying
\begin{equation}
\label{eq:constraintsecond}
\begin{gathered}
\begin{aligned}
\mbox{$\sum_{k}$}\,\gamma_{ik}&=\alpha_i   & &\big(\forall i\in [q]\big),& \mbox{$\sum_{l}$}\,\delta_{jl}&=\beta_j& &\big(\forall j\in [q]\big),& \mbox{$\sum_{j,l}$}\,y_{ikjl}&=\gamma_{ik} & &\big(\forall(i,k)\in [q]^2\big),\\
\mbox{$\sum_{i}$}\,\gamma_{ik}&=\alpha_k   & &\big(\forall k\in [q]\big),& \mbox{$\sum_{j}$}\,\delta_{jl}&=\beta_l& &\big(\forall l\in [q]\big),& \mbox{$\sum_{i,k}$}\,y_{ikjl}&=\delta_{jl}& &\big(\forall(j,l)\in [q]^2\big),\\
\end{aligned}\\
\begin{aligned}
\gamma_{ik}&\geq0 & &\big(\forall(i,k)\in[q]^2\big),&
\delta_{jl}&\geq0 & &\big(\forall(j,l)\in [q]^2\big),&
y_{ikjl}&\geq 0 & &\big(\forall(i,k,j,l)\in[q]^4\big).
\end{aligned}
\end{gathered}
\end{equation}
The first line in \eqref{eq:secondmoment} accounts for the
cardinality of $\Sigma^{\alphab,\betab}\times
\Sigma^{\alphab,\betab}$, while the second line is
$\E_\G[w_G(\sigma_1)w_G(\sigma_2)]$ for
$(\sigma_1,\sigma_2)\in\Sigma^{\alphab,\betab}\times
\Sigma^{\alphab,\betab}$ with the prescribed $\gammab,\deltab$.
Since the weight of a configuration is multiplicative over the
edges and the matchings are independent,
$\E_\G[w_G(\sigma_1)w_G(\sigma_2)]$ is the $\Delta$-power of the
expected weight of a single matching. The latter is completely
determined by $\Y$ and is equal to
$\prod_{i,k,j,l}(B_{ij}B_{kl})^{y_{ikjl}}$, scaled by the probability
that the matching induces the prescribed $\Y$.

\begin{remark}\label{rem:pairedspin}
Note that \eqref{eq:secondmoment} shows that the second moment can be interpreted as the first moment of
a paired-spin model with interaction matrix $\B\otimes \B$. Indeed, we can interpret $B_{ij}B_{kl}$ as the activity
between the paired spins $(i,k)$ and $(j,l)$, thus giving the desired alignment.
\end{remark}

The sums in \eqref{eq:firstmoment} and  \eqref{eq:secondmoment}
are typically exponential in $n$. The most critical component of
our arguments is to find the quantitative structure of
configurations which determine the exponential order of the
moments. Formally, we study the limits of
$\frac{1}{n}\log\E_{\G}\big[Z^{\alphab,\betab}_G\big]$ and
$\frac{1}{n}\log\E_{\G}\big[(Z^{\alphab,\betab}_G)^2\big]$ as
$n\rightarrow \infty$.  Under the usual conventions that $\ln
0\equiv -\infty$ and $0\ln 0\equiv 0$, standard application of
Stirling's approximation yields the following:
\begin{align}
\Psi_1(\alphab,\betab):=\lim_{n\rightarrow \infty}\frac{1}{n}\log\E_{\G}\big[Z^{\alphab,\betab}_G\big]&=\max_{\X}\Upsilon_1(\alphab,\betab,\X),
\label{eq:limitfirst}\\
\mbox{ where } \ \ \
\Upsilon_1(\alphab,\betab,\X)&:=(\Delta-1)f_1(\alphab,\betab)+ \Delta g_1(\X)\notag\\
f_1(\alphab,\betab)&:=\mbox{$\sum_i$}\, \alpha_i\ln\alpha_i+\mbox{$\sum_j$}\, \beta_j\ln\beta_j\notag\\
g_1(\X)&:=\mbox{$\sum_{i,j}$}\,x_{ij}\ln
B_{ij}-\mbox{$\sum_{i,j}$}\, x_{ij}\ln x_{ij}.\notag
\end{align}
And for the second moment:
\begin{align}
\Psi_2(\alphab,\betab):=\lim_{n\rightarrow \infty} \frac{1}{n}\log\E_{\G}\big[(Z^{\alphab,\betab}_G)^2\big]&=\max_{\gammab,\deltab}\max_{\Y}\Upsilon_{2}(\gammab,\deltab,\Y), \label{eq:limitsecond}  \\
\mbox{ where } \ \ \
\Upsilon_{2}(\gammab,\deltab,\Y)&:=(\Delta-1)f_2(\gammab,\deltab)+\Delta g_2(\Y)\notag\\[0.15cm]
f_{2}(\gammab,\deltab)&:=\mbox{$\sum_{i,k}$}\,\gamma_{ik}\ln\gamma_{ik}+\mbox{$\sum_{j,l}$}\,\delta_{jl}\ln\delta_{jl}\notag\\[0.15cm]
g_2(\Y)&:=\mbox{$\sum_{i,k,j,l}$}\, y_{ikjl}\ln
(B_{ij}B_{kl})-\mbox{$\sum_{i,k,j,l}$}\, y_{ikjl}\ln y_{ikjl}\notag
\end{align}

The functions $\Upsilon_1$ and $\Upsilon_2$ are defined on the regions
\eqref{eq:constraintfirst} and \eqref{eq:constraintsecond},
respectively. We also relax the integrality constraints of the
vectors $\alphab,\betab,\X$ and $\gammab,\deltab,\Y$ which were
imposed by the expressions \eqref{eq:firstmoment} and
\eqref{eq:secondmoment}.  This does not affect our considerations
in the limit $n\rightarrow\infty$.  Moreover, note that the
function $\Upsilon_2$ depends on $\alphab,\betab$ due to the linear
constraints \eqref{eq:constraintsecond}. This dependence is
omitted above since we are going to study the second moment for
$\alphab,\betab$ fixed to some well chosen vectors.

The limits \eqref{eq:limitfirst} and \eqref{eq:limitsecond} can
be justified using standard  Laplace arguments (see for example \cite[Chapter 4]{deBru}).

\begin{remark}\label{rem:cruc}
The maximization in the first moment
depends only on the function $g_1(\X)$ which is strictly concave
in the convex region where it is defined. Hence, for any fixed $\alphab,\betab$, 
the global maximum of $\Upsilon_1(\alphab,\betab,\X)$ with respect to $\X$ is
achieved at a unique point. Similarly, for any fixed $\gammab,\deltab$, 
the maximum of $\Upsilon_2(\gammab,\deltab,\y)$ with respect to $\y$ is
achieved at a unique point. Crucially for the calculation of the asymptotics of the second moment in  Appendix~\ref{sec:momentasymptotics}, if $\alphab,\betab$ are global maximizers 
of $\Psi_1$, the global maximum of $\Upsilon_2(\gammab,\deltab,\y)$ with respect to $\gammab,\deltab,\y$ 
is also achieved at a unique point, see Lemma~\ref{lem:maxphi2} in Section~\ref{sec:optsecmoment}.
\end{remark}

A notational convention that we have adopted silently so far is
perhaps useful to allude: the indices $i,k$ ``point" to the set
$V_1$, while indices $j,l$ ``point" to the set $V_2$.


\section{Second Moment Analysis}\label{sec:secondmomentanalysis}

In this section we prove Theorem~\ref{thm:second-moment}.  We first present some basic definitions 
concerning matrix norms.  We then show that the maximum of the first moment function $\Psi_1$
can be reformulated in terms of matrix norms.  This then enables a short proof of Theorem~\ref{thm:second-moment}.


\subsection{Basic Definitions: Matrix Norms}

We will reformulate the maxima of the first and
second moments in terms of matrix norms.
We first recall the basic definitions regarding matrix norms.
The usual vector norms are denoted as:
\[
\|\x\|_p = \Big(\sum_{i=1}^n
x_i^p\Big)^{1/p}.
\]
We will use the subordinate matrix norm (also known as the induced matrix norm)
which will be denoted as $\|\cdot\|_{p\rightarrow q}$ and is defined as:
\begin{equation*}
\| \A\|_{p\rightarrow q} = \max_{\|\x\|_p = 1} \|\A\,\x\|_q.
\end{equation*}
Note that if $\A$ has non-negative entries then one
can restrict the maximization to $\x$ with non-negative entries. A well-known
example of an induced norm is the spectral norm~$\|\cdot\|_{2\rightarrow 2}$.

\subsection{Reformulating the First Moment in Terms of Matrix Norms}

A key component in the analysis of the second moment is the following function $\Phi$.
Let $p=\Delta/(\Delta-1)$.
For non-negative $\rb,\cb$, define $\Phi(\rb,\cb)$ by:
\begin{equation*}
 \exp\big(\Phi(\rb,\cb)/\Delta\big) =
 \frac{ \rb^\T \B \cb}{\|\rb\|_{p} \|\cb\|_p}.
\end{equation*}

We will show that the critical points of $\Phi$ and $\Psi_1$ match in the sense that there is a
one-to-one correspondence between them and their values are equal at the corresponding critical points.
The full statement is contained in Theorem \ref{new:zako1} in Section \ref{sec:scale-free}, but the
important element for the current discussion is captured in the following lemma:
\begin{lemma}\label{lem:keycomponent}
\label{lem:max-match}
\[
\max_{\alphab,\betab\in \triangle_q}   \Psi_1(\alphab,\betab)
 =
 \max_{\rb,\cb} \Phi(\rb,\cb).
\]
\end{lemma}
Therefore, to determine the dominant phases of $\Psi_1$ it suffices to study $\Phi$.
The maximum of $\Phi$ can be compactly expressed in terms of matrix norms as follows:
\begin{equation}
\label{max-Phi}
 \max_{\rb,\cb} \exp\big(\Phi(\rb,\cb)/\Delta\big) = \max_{\cb} \max_{\rb} \frac{ \rb^\T \B
\cb}{\|\rb\|_{p} \|\cb\|_p} = \max_{\cb} \frac{\|\B \cb\|_\Delta}{\|\cb\|_p} =
\|\B\|_{p\ra\Delta},
\end{equation}
where the second equality follows from matrix norm duality.

Hence, the dominant phases of $\Psi_1$ can be expressed in terms of matrix norms:
\begin{equation}
\label{eq:max-norms}
\max_{\alphab,\betab\in \triangle_q}\exp\big(\Psi_1(\alphab,\betab)/\Delta\big) = \|\B\|_{\frac{\Delta}{\Delta-1}\ra \Delta}.
\end{equation}

\subsection{Analyzing the Second Moment: Proof of Theorem \ref{thm:second-moment}}

To analyze the second moment function $\Psi_2$
we will reduce it to the first moment optimization in the following manner.
The key observation is that the associated optimization for the second moment is equivalent to
a first moment optimization of a ``paired-spin'' model which is specified by the tensor product of
the original interaction matrix with itself.  This property enables us to relate the maximum for the
second moment calculations with the maximum of the first moment calculations.

\begin{proof}[Proof of Theorem \ref{thm:second-moment}]

The second moment considers a pair of configurations, say $\sigma$ and $\sigma'$,
which are constrained to have a given phase $\alphab$ for $V_1$ and $\betab$ for $V_2$,
where $V=V_1\cup V_2$.   We capture
this constraint using a pair of vectors
$\gammab,\deltab$ corresponding to the overlap between $\sigma$ and $\sigma'$, in particular,
$\gamma_{ij}$ (and $\delta_{ij}$) is the number of vertices in $V_1$ (and $V_2$, respectively)
with spin $i$ in $\sigma$ and
spin $j$ in $\sigma'$.

Recall, $\Psi^{\B}_1$ indicates the dependence of the function $\Psi_1$ on the interaction matrix $\B$; to simplify the
notation we will drop the exponent if it is $\B$.
 We have (see Remark \ref{rem:pairedspin} in Section \ref{sec:derivations} for more details on this connection)
\begin{equation}\label{dupp}
\Psi_2(\alphab,\betab) = \max_{\gammab,\deltab} \Psi_1^{\B\otimes\B}(\gammab,\deltab),
\end{equation}
\vspace{-.05in}
where the optimization in~\eqref{dupp} is constrained to $\gammab$ and $\deltab$ such that
\begin{equation}\label{dupp2}
\mbox{$\sum_i$}\, \gamma_{ik} = \alpha_k, \quad \quad \mbox{$\sum_k$}\, \gamma_{ik} = \alpha_i, \quad\quad
\mbox{$\sum_j$}\, \delta_{j\ell} = \beta_\ell \quad\mbox{and}\quad \mbox{$\sum_\ell$}\, \delta_{j\ell} = \beta_j.
\end{equation}
Ignoring the four constraints in~\eqref{dupp2} can only increase the value of~\eqref{dupp} and hence
\begin{equation}\label{huko3}
\max_{\alphab,\betab} \exp(\Psi_2(\alphab,\betab)/\Delta) \leq
\max_{\gammab,\deltab} \exp\left(\Psi^{\B\otimes \B}_1(\gammab,\deltab)/\Delta\right) = \|\B\otimes
\B\|_{\frac{\Delta}{\Delta-1}\ra \Delta}.
\end{equation}
The key fact we now use is that for induced norms $\|\cdot\|_{p\rightarrow q}$ with $p\leq q$ it holds
(c.f., \cite[Proposition 10.1]{MR0493490}) that:
\begin{equation}\label{mm13}
\|\B\otimes \B\|_{p\rightarrow q} = \|\B\|_{p\rightarrow q}\, \|\B\|_{p\rightarrow q}.
\end{equation}
Therefore,
\begin{equation}
\max_{\alphab,\betab} \Psi_2(\alphab,\betab)    \leq
2\Delta\log\|\B\|_{\frac{\Delta}{\Delta-1}\ra \Delta}=2\max_{\alphab,\betab} \Psi_1(\alphab,\betab).
\label{eq:maxmax}
\end{equation}

To complete the proof of Theorem \ref{thm:second-moment} it just remains to prove the
reverse inequality, which follows from the
fact that $\E[X^2]\geq \E[X]^2$.
\end{proof}

\subsection{Optimal second moment configuration}\label{sec:optsecmoment}
We will need more detailed information about the $\gammab,\deltab$ which achieve equality in Theorem~\ref{thm:second-moment}
and equation~\eqref{dupp}. The following lemma is true whenever $\B$ is regular (and hence for antiferromagnetic models as well, cf. Definition~\ref{def:antiferromagnetic}).
Roughly, the lemma captures that the major contribution to the second moment comes from pairs of configurations which are uncorrelated. This is crucial to calculate the  asymptotics of the second moment in Appendix~\ref{sec:momentasymptotics}.

\begin{lemma}\label{lem:maxphi2}
Assume that $\B$ is regular. The $\gammab,\deltab$  for which the equality in
\begin{equation}\label{qqzzz}
\max_{\alphab,\betab} \max_{\gammab,\deltab\ \mbox{\rm satisfying}\ \eqref{dupp2}} \Psi_1^{\B\otimes \B}(\gammab,\deltab) =
\max_{\alphab,\betab} \Psi_1^{\B}(\alphab,\betab),
\end{equation}
is achieved satisfy (for all $i,j,k,l\in [q]$)
\begin{equation}\label{qqwww2}
\gamma_{ik} = \alpha_i \alpha_k\quad\mbox{and}\quad \delta_{jl} = \beta_j \beta_l.
\end{equation}
\end{lemma}

\begin{proof}
We will have to dig in to the proof of~\eqref{mm13} and use~\eqref{dupp2}. Bennett's proof of~\eqref{mm13} is the following (our particular values are $q'=\Delta$ and $p=\Delta/(\Delta-1)$):
\begin{eqnarray*}
\|(\B\otimes\B)\rb\|_{q'}  
& = &
 \bigg( \sum_{k} \sum_i \Big| \sum_j B_{ij} \sum_l B_{kl}
R_{jl}\Big|^{q'}\bigg)^{1/{q'}}
\\
&\leq &
 \|\B\|_{p\rightarrow q'} \left(\sum_k \bigg( \sum_j \Big| \sum_l B_{kl} R_{jl} \Big|^p\bigg)^{q'/p}\right)^{1/q'} 
\\ 
&\leq &
\|\B\|_{p\rightarrow q'} \left(\sum_j \bigg( \sum_k \Big| \sum_l B_{kl} R_{jl} \Big|^{q'}\bigg)^{p/q'}\right)^{1/p} 
\\ 
&\leq& 
\|\B\|_{p\rightarrow q'}^2 \bigg(\sum_{j,l} R_{jl}^p\bigg)^{1/p}.
\end{eqnarray*}
Note that in the last inequality one uses $\| \B\ \rb \|_{q'} \leq \|\B\|_{p\rightarrow q'} \|\rb\|_p$, applied to the vectors $\rb'_j:=(R_{j1},R_{j2},\dots,R_{jq})$, for $j=1,\dots,q$. Thus if $\rb$ is a maximizer of
\begin{equation}\label{jako9}
\max_{\rb} \frac{\|(\B\otimes\B) \rb\|_{q'}}{\|\rb\|_p},
\end{equation}
then the vectors $\rb'_j$ are maximizers of
\begin{equation}\label{rako9}
\max_{\rb'} \frac{\| \B\, \rb'\|_{q'}}{\|{\rb'}\|_p}.
\end{equation}
The same, by symmetry, applies to $\rb''_l := (R_{1l},R_{2l},\dots,R_{ql})$, for $l=1,\dots,q$.

The second inequality in Bennett's proof is Minkowski's inequality applied to vectors $\B\rb_1', \dots, \B\rb_q'$. The
equality is achieved only if $\B\rb'_1,\dots,\B\rb'_q$ generate space of dimension one, and since $\B$ is regular
 we have also that $\rb'_1,\dots,\rb'_q$ generate space of dimension one.
Hence, for a maximizer $\rb$ of~\eqref{jako9} we have $\rb=\rb'\otimes \rb''$, where
$\rb'$ and $\rb''$ are maximizers of~\eqref{rako9}. By Theorem~\ref{new:zako1} (equation~\eqref{jqwe}) we then
have
\begin{equation}\label{qqwe12}
\gamma_{ik} = \alpha'_i \alpha''_k
\end{equation}
for the corresponding maximizers of $\Psi_1^{\B\otimes\B}(\gammab,\deltab)$ and $\Psi_1^{\B}(\alphab,\betab)$.
Equation~\eqref{qqwe12} together with constraints
$$
\sum_{i} \gamma_{ik} = \alpha_k\quad\mbox{and}\quad \sum_{k} \gamma_{ik} = \alpha_i,
$$
from \eqref{dupp2} imply $\gamma_{ik}=\alpha_i \alpha_k$ (since $\alpha_k = \sum_{i} \gamma_{ik} = \sum_i \alpha'_i \alpha''_k = \alpha''_k$ and similarly $\alpha_i=\alpha'_i$). The proof of $\delta_{jl}=\beta_j\beta_l$ is analogous.
\end{proof}

\section{Tree recursions, first moment, and matrix norms}\label{sec:treefirstmatrix}
The second moment results of the previous section will be used
to establish that, with probability $1-o(1)$ over the choice of a random $\Delta$-regular bipartite graph, the Gibbs distribution has most of its mass on configurations whose spin frequencies on the two sides of the graph are (close to) dominant phases. To do this, it will be important to examine dominant phases, i.e., the maxima of $\Psi_1(\alphab,\betab)$ and, further, to characterize the local maxima. We will use this information to connect the functions $\Phi$ and $\Psi_1$ and  thus prove Lemma~\ref{lem:keycomponent} which was the critical component in the second moment analysis; in fact, Lemma~\ref{lem:keycomponent} is an immediate corollary of the upcoming Theorem~\ref{new:zako1}  which details further the  connection between $\Phi$ and $\Psi_1$.

For a spin system with interaction matrix $\B$, the following recursions are relevant for the analysis of the critical points of $\Psi_1$. 
\begin{equation}\label{kkrtko}
\hat{R}_i \propto \Big(\sum_{j=1}^q B_{ij}
C_j\Big)^{\Delta-1}\quad\mbox{and}\quad \hat{C}_j \propto
\Big(\sum_{i=1}^q B_{ij} R_j\Big)^{\Delta-1}.
\end{equation} 
We refer to \eqref{kkrtko} as \emph{tree recursions} since they emerge naturally in the analysis of spin systems on the infinite $\Delta$-regular tree $\TreeD$. More precisely, the fixpoints of the tree recursions correspond to semi-translation invariant Gibbs measures on $\TreeD$ (fixpoints of \eqref{kkrtko} are those
$R_i$'s and $C_j$'s such that
$ \hat{R}_i \propto R_i$ and $\hat{C}_j \propto C_j$,
for all $i,j\in [q]$). The fixpoints of the tree recursions correspond to
critical points of $\Psi_1$, as was first observed in \cite{MWW},
see Section~\ref{sec:recursions} for a derivation in our setting. 

We prove the following result which connects tree recursions,  the function $\Phi$ and the function~$\Psi_1$. Lemma \ref{lem:max-match} is a corollary of the following more general theorem.

\begin{theorem}\label{new:zako1}
There is a one-to-one correspondence between the fixpoints of the tree
recursions and the critical points of $\Phi$ (both considered for $R_i\geq 0,C_j\geq 0$
in the projective space, that is, up to scaling by a constant).

The following transformation $(\rb,\cb)\mapsto(\alphab,\betab)$ given by:
\begin{equation}\label{jqwe}
\alpha_i = \frac{R_i^{\Delta/(\Delta-1)}}{\sum_i R_i^{\Delta/(\Delta-1)}}
\quad\mbox{and}\quad
\beta_j = \frac{C_j^{\Delta/(\Delta-1)}}{\sum_j C_j^{\Delta/(\Delta-1)}}
\end{equation}
yields a one-to-one-to-one correspondence between the critical points of $\Phi$ and the critical points of $\Psi_1$ (in the region defined by
$\alpha_i\geq 0,\beta_j\geq 0$ and $\sum_i \alpha_i = 1, \sum_j\beta_j = 1$).

Moreover, for the corresponding critical points $(\rb,\cb)$
and $(\alphab,\betab)$ one has
\begin{equation}\label{pwwww1}
\Phi(\rb,\cb) = \Psi_1(\alphab,\betab).
\end{equation}

Finally, for spin systems whose interaction matrix $\B$ is ergodic, the local maxima of $\Phi$ and $\Psi_1$ happen at the critical points (that is, there are no local
maxima on the boundary).
\end{theorem}


To argue that the Gibbs distribution places most of its mass on configurations whose spin frequencies  are given by dominant phases, we need a more explicit handle on \emph{local maxima} of $\Psi_1$. The latter will also be crucial to analyze the global maxima of $\Psi_1$ for specific models of interest. 

We connect local maxima of $\Psi_1$ to attractive fixpoints of the associated tree recursions. Specifically, we call a fixpoint $x$ of a function $f$ a {\em Jacobian attractive
fixpoint} if the Jacobian of $f$ at $x$ has spectral radius less than~$1$.
We say that a
critical point $\alphab,\betab$ is a {\em Hessian local maximum} if the
Hessian of $\Psi_1$ at $\alphab,\betab$ is negative definite  (note this is a sufficient condition
for $\alphab,\betab$ to be a local maximum).

We prove the following theorem in Section \ref{sec:connection}.

\begin{theorem}\label{thm:connection}
Jacobian attractive fixpoints of  the tree recursions \eqref{kkrtko} (considered as a function $(R_1,\hdots,R_q,C_1,\hdots,C_q)\mapsto (\hat{R}_1,\hdots,\hat{R}_q,\hat{C}_1,\hdots,\hat{C}_q)$) 
correspond to Hessian local maxima of $\Psi_1$.
\end{theorem}

Theorem \ref{thm:connection} is important for analyzing the global maxima of $\Psi_1$ for
colorings and antiferromagnetic Potts model (see Section \ref{sec:phase-diagram}). Moreover, it will be used to apply the small subgraph conditioning method (see Section~\ref{sec:applicationsmallgraph}).

\subsection{Connection between $\Phi$ and $\Psi_1$}\label{sec:scale-free}
In this section, we prove Theorem~\ref{new:zako1}.
\subsubsection{Preliminaries on maximum-entropy distributions}

Let $\alphab$ and $\betab$ be non-negative vectors in ${\mathbb R}^q$ such
that
\begin{equation}\label{nwuu}
\sum_i\alpha_i = 1\quad\mbox{and}\quad\sum_j\beta_j = 1.
\end{equation}
For $\alphab$ and $\betab$ that satisfy~\eqref{nwuu} let
\begin{equation}\label{pop}
g(\alpha_1,\dots,\alpha_q,\beta_1,\dots,\beta_q) = \max
\sum_{i=1}^q\sum_{j=1}^q x_{ij} (\ln(B_{ij}) - \ln x_{ij}),
\end{equation}
where the maximum is taken over non-negative $x_{ij}$'s such that
\begin{equation}\label{linsu1}
\alpha_i = \sum_{j} x_{ij}\quad\mbox{and}\quad\beta_j=\sum_{i}
x_{ij}.
\end{equation}

\begin{lemma}\label{helma2}
The maximum of the right-hand-side of~\eqref{pop} is achieved at unique $x_{ij}$.
The $x_{ij}$ are given by
\begin{equation}\label{oqw}
x_{ij} = B_{ij} R_i C_j,
\end{equation}
where $\rb$ and $\cb$ satisfy
\begin{equation}\label{eex}
R_i\sum_{j=1}^q B_{ij} C_j = \alpha_i\quad\mbox{and}\quad
C_j\sum_{i=1}^q B_{ij} R_i = \beta_j,
\end{equation}
and
\begin{equation}\label{eextropo}
\begin{split}
\sum_{j=1}^q B_{ij} C_j = 0 \implies R_i=0;\\
\sum_{i=1}^q B_{ij} R_i = 0 \implies C_j=0.
\end{split}
\end{equation}
The value of $g$, in terms of $R_i$'s and $C_j$'s, is given by
\begin{equation}\label{bronco}
g(\alpha_1,\dots,\alpha_q,\beta_1,\dots,\beta_q) =
-\sum_{i=1}^q\sum_{j=1}^q B_{ij} R_i C_j \ln (R_iC_j).
\end{equation}
\end{lemma}

\begin{proof}
From strict concavity of $-x\ln x$ it follows that the
right-hand side of~\eqref{pop} has a unique critical point (if there
were two critical points then the segment between the points lies
in the linear space defined by~\eqref{linsu1}; the function
has a zero derivative on both ends of the segment; and the
second derivative of the function is negative on the segment;
a contradiction).

Using the method of Lagrange multipliers we obtain that the
critical points of the right-hand side of~\eqref{pop} are $x_{ij}$
given by~\eqref{oqw} where $R_i$'s and $C_j$'s are solutions of~\eqref{eex}.
We can make any solution of~\eqref{eex} satisfy~\eqref{eextropo}: if
$\sum_{j=1}^q B_{ij} C_j=0$ then set $R_i=0$ (and symmetrically,
if $\sum_{i=1}^q B_{ij} R_i=0$ then set $C_j=0$). We now argue that this
change does not violate~\eqref{eex}. Suppose that after the change
for some $k\in [q]$ we have
\begin{equation}\label{rruq2}
R_k \sum_{j=1}^q B_{kj} C_j \neq \alpha_k.
\end{equation}
Then $i=k$ (since only $R_i$ changed) and since
$\sum_{j=1}^q B_{ij} C_j=0$ we also have $\alpha_i=0$, a contradiction
(with~\eqref{rruq2}). Now suppose that after the change for some $j\in [q]$ we have
\begin{equation}\label{rruq}
C_j \sum_{k=1}^q B_{kj} R_k \neq \beta_j.
\end{equation}
Then $B_{ij}>0$ and $C_j>0$ (otherwise changing $R_i$ would not violate~\eqref{rruq}).
This then implies
$\sum_{j=1}^q B_{ij} C_j > B_{ij} C_j > 0$, a contradiction. Thus
the change does not violate~\eqref{eex}.

Equation~\eqref{bronco} is obtained by substituting~\eqref{eex} into~\eqref{pop}.
\end{proof}

\begin{remark}\label{hop134}
Scaling all the $R_i$'s up by the same factor while scaling all
the $C_j$'s down by the same factor preserves~\eqref{oqw}
and~\eqref{eex}. Modulo such scaling the $R_i$'s and $C_j$'s are
unique, since the $x_{ij}$'s are unique and~\eqref{oqw} determines
the $R_i$'s and $C_j$'s once one value (say $R_1$) is fixed (here
we use the fact that the matrix of the model is ergodic).
\end{remark}

\begin{remark}
Note that the condition~\eqref{nwuu} translates (using \eqref{eex}) into the following condition
on $R_i$'s and $C_j$'s
\begin{equation}\label{gokso}
\sum_{i=1}^q \sum_{j=1}^q B_{ij} R_i C_j = 1.
\end{equation}
\end{remark}

Our goal now is to see how the value of~\eqref{pop} changes when
we perturb $\alpha_i$'s and $\beta_j$'s. We are going to view them
as functions of a new variable $z$. All differentiation in this
section will be with respect to $z$. Note that to stay in the
subspace defined by~\eqref{nwuu} we should have, in particular,
\begin{equation}\label{nwuu2}
\sum_i\alpha'_i = \sum_i \alpha''_i = 0\quad\mbox{and}\quad\sum_j\beta'_j = \sum_j\beta''_j= 0.
\end{equation}
Differentiating~\eqref{eex} we obtain
\begin{equation}\label{eex2}
\sum_{j=1}^q B_{ij} (R_i C_j)' = \alpha'_i\quad\mbox{and}\quad
\sum_{i=1}^q B_{ij} (R_i C_j)' = \beta'_j.
\end{equation}
The following ratio of~\eqref{eex} and~\eqref{eex2} will be useful
later
\begin{equation}\label{eex22}
\frac{\alpha_i'}{\alpha_i} = \frac{R_i'}{R_i} + \frac{\sum_{j=1}^q
B_{ij} C'_j}{\sum_{j=1}^q B_{ij} C_j} \quad\mbox{and}\quad
\frac{\beta_j'}{\beta_j} = \frac{C_j'}{C_j} + \frac{\sum_{i=1}^q
B_{ij} R'_i}{\sum_{i=1}^q B_{ij} R_i}.
\end{equation}

The scaling freedom for $R_i$'s and $C_j$'s (discussed in Remark~\ref{hop134})
is equivalent to increasing all $R'_i/R_i$'s by
the same (additive) amount and decreasing all $C'_i/C_i$ by the same
(additive) amount. We are going to remove this freedom by
requiring
\begin{equation}\label{reqoq}
\sum_{i=1}^q \alpha_i \frac{R_i'}{R_i} = \sum_{j=1}^q \beta_j
\frac{C_j'}{C_j}.
\end{equation}
(Recall that we study the effect of perturbing~$g$ when we change
$\alpha_i$'s and $\beta_j$'s; equation~\eqref{reqoq} just fixes
the corresponding change in $R_i$'s and $C_j$'s.)

Now we compute the derivatives of $g$.
\begin{lemma}
We have
\begin{gather}
g' = - \sum_{i=1}^q (\ln R_i) \alpha'_i - \sum_{j=1}^q (\ln C_j)\beta'_j,\label{l11}\\
g'' = - \sum_{i=1}^q \frac{R'_i}{R_i}\alpha_i' -
\sum_{j=1}^q\frac{C'_j}{C_j}\beta_j'-\sum_{i=1}^q (\ln R_i)\alpha''_i
-\sum_{j=1}^q (\ln C_j)\beta''_j.\label{l22}
\end{gather}
\end{lemma}

\begin{proof}
Using $(f \ln f)' = (1+\ln f) f'$ and equations~\eqref{eex2}
and~\eqref{nwuu2} we obtain
$$
g' = - \sum_{i=1}^q\sum_{j=1}^q B_{ij} \big(1+\ln(R_i C_j)\big) (R_iC_j)' = \\
- \sum_{i=1}^q (\ln R_i)\alpha'_i - \sum_j (\ln C_j)\beta'_j.
$$
Differentiating~\eqref{l11} we
obtain~\eqref{l22}. \end{proof}

Note the expressions~\eqref{l11} and~\eqref{l22} are independent of the
choice of scaling of $R_i$'s and $C_j$'s (this follows from~\eqref{nwuu2}).
The particular tying of $R_i'/R_i$'s and $C'_j/C_j$'s to $\alpha'_i$ and
$\beta'_j$ (given by~\eqref{reqoq}) will be useful later.

\subsubsection{Critical points of $\Psi_1$ and the tree recursions}
\label{sec:recursions}

In this section we establish the connection between the critical points of
$\Psi_1$ and the fixpoints of the tree recursions.

\begin{lemma}\label{st1}
Let $\alphab,\betab$ be a critical point of $\Psi_1(\alphab,\betab)$ in the subspace
defined by~\eqref{nwuu}. Let $\rb,\cb$
be given by~\eqref{eex}. Then \begin{equation}\label{gen}
\alpha_i \propto
R_i^{\Delta/(\Delta-1)}\quad\mbox{and}\quad
\beta_j \propto
C_j^{\Delta/(\Delta-1)}.
\end{equation}
Consequently, $\rb,\cb$ satisfy the tree recursions
stated in Section~\ref{sec:treefirstmatrix}:
\begin{equation*}
\tag{\ref{kkrtko}}
R_i\propto\Big(\sum_{j=1}^q B_{ij}
C_j\Big)^{\Delta-1}\quad\mbox{and}\quad
C_j\propto\Big(\sum_{i=1}^q B_{ij}
R_i\Big)^{\Delta-1}.
\end{equation*}
\end{lemma}

\begin{proof}
At the critical points of $\Psi$ the first derivative of $\Psi$
has to vanish for all $\alpha'_i$'s and $\beta'_j$'s from the
subspace defined by \eqref{nwuu2}, that is,
\begin{align}
\nonumber \Psi' & = (\Delta-1) \Big(\sum_{i=1}^q
(1+\ln\alpha_i)\alpha_i' + \sum_{j=1}^q
(1+\ln\beta_j)\beta_j'\Big) - \Delta\Big(\sum_{i=1}^q (\ln
R_i)\alpha_i' + \sum_{j=1}^q (\ln C_j)\beta_j'\Big)
\\
\label{kkko} & = \sum_{i=1}^q \big((\Delta-1)(1+\ln\alpha_i)-\Delta\ln
R_i\big)\alpha_i' + \sum_{j=1}^q \big((\Delta-1)(1+\ln\beta_j)-\Delta \ln
C_j\big)\beta_j' = 0,
\end{align}
where the $R_i$'s and $C_j$'s are given by \eqref{eex}.
Inspecting~\eqref{kkko} we see that
$(\Delta-1)(1+\ln\alpha_i)-\Delta\ln R_i$ have the same value. Indeed, if
two of them, say with indices $i_1,i_2$, had different values then
we could increase $\alpha_{i_1}$ and decrease $\alpha_{i_2}$ by
the same infinitesimal amount and violate~\eqref{kkko}.
Similarly, $(\Delta-1)(1+\ln\beta_j)-\Delta C_j$ have the same
value and hence we have~\eqref{gen}. Plugging~\eqref{gen}
into~\eqref{eex} one obtains~\eqref{kkrtko}.
\end{proof}

\begin{lemma}\label{st2}
Let $(\rb,\cb)$ be a solution of the tree recursions~\eqref{kkrtko}.
Let $(\alphab,\betab)$ be given by~\eqref{jqwe}. Then~$(\alphab,\betab)$ is a
critical point of $\Psi_1(\alphab,\betab)$ in the subspace defined by~\eqref{nwuu}.
\end{lemma}

\begin{proof}
Let
\begin{equation*}
Z_R:=(\Delta-1)(1+\ln\alpha_i)-\Delta\ln R_i = (\Delta-1) \Big(1 - \ln \sum_{i=1}^q R_i^{(\Delta+1)/\Delta}\Big),
\end{equation*}
where the second equality follows from~\eqref{jqwe}. Note that $Z_R$ is independent of the choice of $i$. Similarly let
\begin{equation*}
Z_C:=(\Delta-1)(1+\ln\beta_j)-\Delta\ln C_j = (\Delta-1) \Big(1 - \ln \sum_{j=1}^q C_j^{(\Delta+1)/\Delta}\Big).
\end{equation*}
For perturbations of $\alphab,\betab$ in the subspace given by~\eqref{nwuu} we have
\begin{equation*}
\begin{split}
\Psi_1'(\alphab,\betab) =
\sum_{i=1}^q \big((\Delta-1)(1+\ln\alpha_i)-\Delta\ln R_i\big)\alpha_i'
+
\sum_{j=1}^q \big((\Delta-1)(1+\ln\beta_j)-\Delta\ln C_j\big)\beta_j'\\
= Z_R \sum_{i=1}^q \alpha_i' + Z_C \sum_{j=1}^q \beta_j' = 0,
\end{split}
\end{equation*}
and hence $(\alphab,\betab)$ is a critical point.
\end{proof}

\subsubsection{Value of $\Psi_1$ at the critical points}

\begin{lemma}\label{st3}
Let $(\alphab,\betab)$ be critical point of $\Psi_1(\alphab,\betab)$. Let $(\rb,\cb)$
be given by~\eqref{eex}. Then
\begin{equation*}\tag{\ref{pwwww1}}
\Phi(\rb,\cb) = \Psi_1(\alphab,\betab).
\end{equation*}
Moreover, $(\rb,\cb)$ is a critical point of $\Phi(\rb,\cb)$.
\end{lemma}

\begin{proof}
We have (see equation~\eqref{bronco})
\begin{equation}\label{ooxe}
\Psi_1(\alphab,\betab) = (\Delta-1)\Big(\sum_{i=1}^q \alpha_i\ln\alpha_i +
\sum_{j=1}^q \beta_j\ln\beta_j\Big) - \Delta
\sum_{i=1}^q\sum_{j=1}^q B_{ij} R_i C_j \ln (R_iC_j).
\end{equation}
At the critical points we have (see equation~\eqref{gen})
\begin{equation}\label{gennn}
\alpha_i = \frac{R_i^{\Delta/(\Delta-1)}}{\sum_{i=1}^q
R_i^{\Delta/(\Delta-1)}} \quad\mbox{and}\quad \beta_j =
\frac{C_j^{\Delta/(\Delta-1)}}{\sum_{j=1}^q
C_j^{\Delta/(\Delta-1)}}.
\end{equation}
Plugging~\eqref{eex} into~\eqref{ooxe} we obtain
\begin{equation}\label{ooxe2Q}
\begin{split}
\Psi_1(\alphab,\betab) = (\Delta-1)\Big(\sum_{i=1}^q \alpha_i\ln\alpha_i + \sum_{j=1}^q \beta_j\ln\beta_j\Big) - \Delta \Big(\sum_{i=1}^q\alpha_i\ln R_i +  \sum_{j=1}^q\beta_j\ln C_j\Big)=\\
\sum_{i=1}^q\alpha_i \ln\frac{\alpha_i^{\Delta-1}}{R_i^{\Delta}} +
\sum_{j=1}^q \beta_j \ln\frac{\beta_j^{\Delta-1}}{C_j^{\Delta}}=
-(\Delta-1)\Big[\ln\Big(\sum_{i=1}^q
R_i^{\Delta/(\Delta-1)}\Big)+\ln\Big(\sum_{j=1}^q
C_j^{\Delta/(\Delta-1)}\Big)\Big],
\end{split}
\end{equation}
where in the last equality we used~\eqref{gennn} and the fact that
$\alpha_i$'s and $\beta_j$'s sum to $1$. Recall that
\begin{equation}\label{cobb}
\sum_{i=1}^q \sum_{j=1}^q B_{ij} R_i C_j = \sum_{i=1}^q \alpha_i = 1,
\end{equation}
and hence the following is obtained by adding zero to the
right-hand side of~\eqref{ooxe2Q}
\begin{eqnarray*}
\Psi_1(\alphab,\betab) 
&=& 
\Delta\ln\Big(\sum_{i=1}^q \sum_{j=1}^q B_{ij} R_i
C_j\Big) - (\Delta-1)\Big[\ln\Big(\sum_{i=1}^q
R_i^{\Delta/(\Delta-1)}\Big)+\ln\Big(\sum_{j=1}^q
C_j^{\Delta/(\Delta-1)}\Big)\Big]
\\ & =& 
 \Phi(\rb,\cb).
\end{eqnarray*}
Now we argue that $(\rb,\cb)$ is a critical point of
$\Phi(\rb,\cb)$. We have
\begin{equation}\label{gch1}
\frac{\partial}{\partial R_i}\Phi(\rb,\cb) = \Delta \frac{\sum_{j=1}^q
B_{ij} C_j}{\sum_{i=1}^q\sum_{j=1}^q B_{ij}R_iC_j} -
(\Delta-1)\frac{\frac{\Delta}{\Delta-1}R_i^{1/(\Delta-1)}}{\sum_{i=1}^q
R_i^{\Delta/(\Delta-1)}}= \Delta \frac{\alpha_i}{R_i} - \Delta\frac{\alpha_i}{R_i} =0.
\end{equation}
where we used~\eqref{gennn},~\eqref{eex}, and~\eqref{nwuu}. The same argument yields
\begin{equation}\label{gch1cvik}
\frac{\partial}{\partial C_j}\Phi(\rb,\cb) = \Delta \frac{\sum_{i=1}^q
B_{ij} R_i}{\sum_{i=1}^q\sum_{j=1}^q B_{ij}R_iC_j} -
(\Delta-1)\frac{\frac{\Delta}{\Delta-1}C_j^{1/(\Delta-1)}}{\sum_{j=1}^q
C_j^{\Delta/(\Delta-1)}}=0.
\end{equation}
and hence $\rb,\cb$ is a critical point of $\Phi$.
\end{proof}

\begin{lemma}\label{st4}
Let $(\rb,\cb)$ be a critical point of $\Phi(\rb,\cb)$.
Let $\alphab,\betab$ be given
by~\eqref{jqwe}. Then $\alphab,\betab$ is a critical point of $\Psi_1(\alphab,\betab)$
in the subspace defined by~\eqref{nwuu}.
\end{lemma}

\begin{proof}
At a critical point of $\Phi$ we have that~\eqref{gch1} is zero for $i\in [q]$.
Note that the denominators do not depend on $i$ and hence we have
$$
R_i^{1/(\Delta-1)}\propto \sum_{j=1}^q B_{ij} C_j.
$$
Similarly, from~\eqref{gch1cvik} we obtain
$$
C_j^{1/(\Delta-1)}\propto \sum_{i=1}^q B_{ij} R_i.
$$
Hence $(\rb,\cb)$  satisfy the tree recursions. Now we use
Lemma~\ref{st2} to conclude that $(\alphab,\betab)$ is a critical point of
$\Psi_1(\alphab,\betab)$ in the subspace
defined by~\eqref{nwuu}.
\end{proof}

\subsubsection{Local maxima of $\Psi_1$ are in the interior}\label{sergodic}

In this section we show that for models with ergodic (irreducible and aperiodic) interaction matrix $\B$ the
maximum of $\Phi(\rb,\cb)$ is achieved in the interior.
A symmetric matrix is irreducible if the graph whose edges correspond to non-zero edges
of $\B$ is connected. A symmetric matrix is aperiodic if the graph whose edges correspond to
non-zero edges of $\B$ has an odd cycle.

\begin{lemma}\label{st5}
Assume that $\B$ is ergodic.
Let $(\rb,\cb)\neq 0$ be a local maximum of $\Phi$ in the
region $\rb,\cb\geq 0$. Then $R_i>0$ for all $i\in [q]$ and $C_j>0$ for
all $j\in [q]$.
\end{lemma}

\begin{proof}
Suppose not, that is, we have a maximum that has a zero on some coordinate of $\rb$
or $\cb$. From the ergodicity of $\B$ we have that there exist $i,j\in [q]$ such
that i) $R_i=0$, $C_j>0$, and $B_{ij}>0$ or ii) $R_i>0$, $C_j=0$, and $B_{ij}>0$.
(Suppose not. Let $Z_R\subseteq [q]$ be the set of $i$ such that $R_i=0$. Similarly let $Z_C\subseteq [q]$ be the set of $j$
such that $C_j=0$. If neither i) nor ii) happens then non-zero $B_{ij}$ are possibly between $i\in Z_R$ and $j\in Z_C$ and
$i\in [q]\setminus Z_R$ and $j\in [q]\setminus Z_C$. Thus in $\B^2$ the non-zero $(B^2)_{ij}$
are possibly between $i,j\in Z_R$ and $i,j\in [q]\setminus Z_R$. Thus $\B$ is not ergodic.)
W.l.o.g. assume that it is the case i) (the case ii) is handled analogously).

The derivative of $\Phi$ w.r.t. $R_i$ is (we are using $R_i=0$)
\begin{equation*}
\frac{\partial}{\partial R_i} \Phi(\rb,\cb)= \Delta \frac{\sum_{j=1}^q
B_{ij} C_j}{\sum_{i=1}^q \sum_{j=1}^q B_{ij}R_i C_j} > \Delta \frac{
B_{ij} C_j}{\sum_{i=1}^q \sum_{j=1}^q B_{ij}R_i C_j} > 0,
\end{equation*}
and hence we are not at a maximum, a contradiction.
\end{proof}

\begin{lemma}\label{st6}
Assume that $\B$ is ergodic. Let $\alphab,\betab\geq 0$ be a local maximum of
$\Psi_1(\alphab,\betab)$ in the subspace defined by~\eqref{nwuu}.
Then $\alpha_i>0$ for all $i\in [q]$ and $\beta_j>0$ for all $j\in [q]$.
\end{lemma}

\begin{proof}
It will be useful to view $\Psi_1$ as a function of $(\rb,\cb)$. Because of Lemma~\ref{helma2} we have $(\rb,\cb)$ satisfying~\eqref{cobb} and~\eqref{eextropo} (and any such $(\rb,\cb)$ yields $(\alphab,\betab)$
satisfying~\eqref{nwuu}). We have (from~\eqref{ooxe2Q})
\begin{eqnarray*}\label{ooxe2}
\Psi_1(\alphab,\betab) &=& 
\sum_{i=1}^q \sum_{j=1}^q B_{ij} R_i C_j
\left(
(\Delta-1)\bigg(\ln\Big(\sum_{j=1}^q B_{ij} C_j\Big) + \ln\Big(\sum_{i=1}^q B_{ij} R_i\Big)\bigg) -
\ln R_i - \ln C_j \right) 
\\
&=: & \hat{\Psi}_1(\rb,\cb).
\end{eqnarray*}
If $\rb$ has a zero coordinate then, by ergodicity of $\B$ there exists $k,\ell\in [q]$ such that
i) $R_k=0$, $C_\ell>0$, and $B_{k\ell}>0$ or ii) $R_k>0$, $C_\ell=0$, and $B_{k\ell}>0$ (see the argument in
the proof of Lemma~\ref{st5}). W.l.o.g. it is the case i).

Note that we have
\begin{equation}\label{iii121}
\frac{\partial}{\partial R_k}\sum_{i=1}^q\sum_{j=1}^q B_{ij} R_i C_j = \sum_{j=1}^q B_{kj} C_j \geq B_{k\ell} C_\ell > 0.
\end{equation}
We have
\begin{align}
\frac{\partial}{\partial R_k} \hat{\Psi}_1 &=
\sum_{j=1}^q B_{k j} C_j
\bigg(
(\Delta-1)\ln\Big(\sum_{i=1}^q B_{ij} R_i\Big) - \ln C_j\bigg) \nonumber\\
&+
\bigg((\Delta-1)\ln\Big(\sum_{j=1}^q B_{kj} C_j\Big) - \ln R_k\bigg)
\Big(\sum_{j=1}^q B_{k j} C_j\Big)  + (\Delta-2)\sum_{j=1}^q  B_{kj} C_j.\label{ziako1}
\end{align}
The first sum in~\eqref{ziako1} is finite since if $C_j>0$ then $\sum_{i=1}^q B_{ij} R_i>0$
(using~\eqref{eextropo}); if $C_j=0$ then the contribution of the
term to the sum is zero (we are using the usual convention $0\ln 0 = 0$). The second term in~\eqref{ziako1} has value $+\infty$ since
$\ln R_k = -\infty$ and~\eqref{iii121}. Finally, the last term in~\eqref{ziako1} is finite and hence
we have $\frac{\partial}{\partial R_k} \hat{\Psi}_1 = + \infty$.

Recall that $C_\ell>0$ and hence (using~\eqref{eextropo}):
\begin{equation}\label{iii121B}
\frac{\partial}{\partial C_\ell}\sum_{i=1}^q\sum_{j=1}^q B_{ij} R_i C_j = \sum_{i=1}^q B_{i\ell} R_i > 0.
\end{equation}
Finally, we argue that $\frac{\partial}{\partial C_\ell} \hat{\Psi}_1$ is finite. We have (analogously to~\eqref{ziako1})
\begin{align}
\frac{\partial}{\partial C_\ell} \hat{\Psi}_1 &=
\sum_{i=1}^q B_{i \ell} C_i
\bigg(
(\Delta-1)\ln\Big(\sum_{j=1}^q B_{ij} C_j\Big) - \ln R_i\bigg) \nonumber\\
&+
\bigg((\Delta-1)\ln\Big(\sum_{i=1}^q B_{i\ell} R_i\Big) - \ln C_\ell\bigg)
\Big(\sum_{i=1}^q B_{i \ell} R_i\Big)  + (\Delta-2)\sum_{i=1}^q  B_{i\ell} R_i.\label{zrialo1}
\end{align}
The first and third term in~\eqref{zrialo1} are finite by the same argument as for~\eqref{ziako1}.
In the second term we use~\eqref{iii121B} and $C_\ell>0$.

Now we increase $R_k$ by an infinitesimal amount and change $C_\ell$ to maintain~\eqref{gokso}
(and hence~\eqref{nwuu}). (This is possible because both $C_\ell$ and $R_k$ change the value
of~\eqref{gokso}, see equations~\eqref{iii121} and~\eqref{iii121B}.) This change will increase $\hat{\Psi}_1$
and hence $\Psi_1$ contradicting the local maximality of $\alphab,\betab$.
\end{proof}

\subsubsection{Proof of Theorem \ref{new:zako1}}

\begin{proof}[Proof of Theorem \ref{new:zako1}]
Lemmas~\ref{st1} and~\ref{st2} give the connection between the
critical points of $\Psi_1$ and the fixpoints of the tree
recursions. Lemmas~\ref{st3} and~\ref{st4} give connection between
the critical points of $\Psi_1$ and $\Phi$ and show that the
values agree on the corresponding critical points. Finally,
Lemmas~\ref{st5} and~\ref{st6} show that the maxima happen in the
interior (that is, for $R_i>0, C_j>0$ in the case of $\Phi$ and
for $\alpha_i>0,\beta_j>0$ in the case of $\Psi_1$).
\end{proof}

\subsection{Connecting Local Maxima and Stability of Tree Recursions}
\label{sec:connection}

In this section we prove Theorem~\ref{thm:connection}.

\subsubsection{Maximum entropy configurations on random $\Delta$-regular bipartite graphs}\label{sdrg}
We analyze the critical points by looking at the
second derivative. Using $(f \ln f)'' = (f')^2/f +(1+\ln f) f''$ we have
\begin{eqnarray}
\nonumber
\lefteqn{ \Psi_1''(\alphab,\betab) }
\\
\nonumber
&=&
(\Delta-1)\sum_{i=1}^q \Big((\alpha'_i)^2/\alpha_i + (1+\ln\alpha_i)\alpha_i'' \Big)
-\Delta\sum_{i=1}^q\Big(\alpha_i'\frac{R_i'}{R_i}+(\ln R_i)\alpha_i''\Big)
\\
\nonumber
&& \ + \  (\Delta-1)\sum_{j=1}^q \Big((\beta'_j)^2/\beta_j + (1+\ln\beta_j)\beta_j'' \Big)
-\Delta\sum_{j=1}^q\Big(\beta_j'\frac{C_j'}{C_j}+(\ln C_j)\beta_j''\Big)
 \\
\label{aqwq}
&=&
 (\Delta-1)\sum_{i=1}^q (\alpha'_i)^2/\alpha_i -\Delta\sum_{i=1}^q \alpha_i'\frac{R_i'}{R_i}
+ \sum_{i=1}^q \alpha_i''\Big( (\Delta-1)(1+\ln\alpha_i) - \Delta \ln R_i \Big) 
\\
\nonumber
&&
\ + \  (\Delta-1)\sum_{j=1}^q (\beta'_j)^2/\beta_j -\Delta\sum_{j=1}^q \beta_j'\frac{C_j'}{C_j}
+\sum_{j=1}^q \beta_j''\Big( (\Delta-1)(1+\ln\beta_j) - \Delta \ln C_j \Big) 
\\
\nonumber
&   = &
(\Delta-1)\sum_{i=1}^q (\alpha'_i)^2/\alpha_i -\Delta\sum_{i=1}^q \alpha_i'\frac{R_i'}{R_i}
+ (\Delta-1)\sum_{j=1}^q (\beta'_j)^2/\beta_j -\Delta\sum_{j=1}^q \beta_j'\frac{C_j'}{C_j},
\end{eqnarray}
where the last equality follows from~\eqref{kkko} (replacing $\alpha_i'$ by $\alpha_i''$
and $\beta_j'$ by $\beta_j''$; note that they are both from the same subspace~\eqref{nwuu2}).

Plugging~\eqref{eex22} into~\eqref{aqwq} we obtain
\begin{equation}\label{aqwq2}
\begin{split}
\Psi_1''(\alphab,\betab) = \sum_{i=1}^q \alpha'_i\left((\Delta-1)\frac{\sum_{j=1}^q
B_{ij} C'_j}{\sum_{j=1}^q B_{ij} C_j} - \frac{R_i'}{R_i} \right) +
\sum_{j=1}^q \beta'_j\left((\Delta-1)\frac{\sum_{i=1}^q B_{ij}
R'_i}{\sum_{i=1}^q B_{ij} R_i} -\frac{C_j'}{C_j}\right).
\end{split}
\end{equation}

We are going to use the second partial derivative test (which
gives a sufficient condition) to establish maxima of $\Psi_1$. We
will use the following terminology for local maxima established
using this method.

\begin{definition}
A critical point $x$ of a function $f:{\cal M}\ra\R$ is called
{\bf Hessian local maximum} if the Hessian of $f$ at $x$ is
negative definite.
\end{definition}

Let $\L$ be the (matrix of) linear map $(r_1,\dots,r_q,c_1,\dots,c_q)\mapsto
(\hat{r}_1,\dots,\hat{r}_q,\hat{c}_1,\dots,\hat{c}_q)$ given by
\begin{equation}\label{lmp}
\hat{r}_i =
\sum_{j}\frac{B_{ij}R_iC_j}{\sqrt{\alpha_i\beta_j}} c_j
\quad\mbox{and}\quad \hat{c}_j =
\sum_{i}\frac{B_{ij}R_iC_j}{\sqrt{\alpha_i\beta_j}} r_i.
\end{equation}
In the following, we denote by $\I$ the identity matrix of dimension $2q\times 2q$.

\begin{lemma}\label{ztt1}
A critical point $(\alphab,\betab)$ is a Hessian local maximum of $\Psi_1(\alphab,\betab)$
in the subspace defined by~\eqref{nwuu2} if and only if
$\w^{\T} (\I + \L)( (\Delta-1) \L - \I) \w<0$ for all $\w=(r_1,\dots,r_q,c_1,\dots,c_q)^{\T}$
such that
\begin{equation}\label{pakoprob}
\sum_{i=1}^q \sqrt{\alpha_i} r_i = 0
\quad\mbox{and}\quad
\sum_{j=1}^q \sqrt{\beta_j} c_j = 0.
\end{equation}
\end{lemma}

\begin{proof}
To check whether we are at a Hessian local maximum of $\Psi(\alphab,\betab)$ we
have to have~\eqref{aqwq2} negative for non-zero $\alpha'_i$'s and
$\beta'_j$'s from the subspace defined by \eqref{nwuu2}
and~\eqref{reqoq}.

Let $r_i = \sqrt{\alpha_i} R_i' / R_i$ and $c_j = \sqrt{\beta_j}
C_j'/C_j$. Using~\eqref{eex22} we have
\begin{align*}
\Psi'' & =  \sum_{i} \alpha_i\bigg(\frac{R_i'}{R_i} +
\frac{\sum_{j} B_{ij} C'_j}{\sum_{j} B_{ij} C_j}\bigg)
\bigg((\Delta-1)\frac{\sum_{j} B_{ij} C'_j}{\sum_{j}
B_{ij} C_j} - \frac{R_i'}{R_i} \bigg)
\\  &\qquad  + \
\sum_{j} \beta_j\bigg(\frac{C_j'}{C_j} + \frac{\sum_{i}
B_{ij} R'_i}{\sum_{i} B_{ij}
R_i}\bigg)\bigg((\Delta-1)\frac{\sum_{i} B_{ij}
R'_i}{\sum_{i} B_{ij} R_i} -\frac{C_j'}{C_j}\bigg)
\\ & =
\sum_{i}\bigg(r_i + \sum_{j}
\frac{B_{ij}R_iC_j}{\sqrt{\alpha_i\beta_j}} c_j\bigg)\bigg(
\sum_{j} (\Delta-1)\frac{B_{ij}R_iC_j}{\sqrt{\alpha_i\beta_j}}
c_j - r_i\bigg)
\\ &\qquad+ \
\sum_{j} \bigg(c_j + \sum_{i}
\frac{B_{ij}R_iC_j}{\sqrt{\alpha_i\beta_j}} r_i\bigg)\bigg(
\sum_{i} (\Delta-1)\frac{B_{ij}R_iC_j}{\sqrt{\alpha_i\beta_j}}
r_i - c_j\bigg).
\end{align*}
Let $\w=(r_1,\dots,r_q,c_1,\dots,c_q)^{\T}$. In terms of $\L$ and $\w$ we
have
\begin{equation}\label{liq}
\Phi'' = \w^\T (\I + \L)( (\Delta-1) \L - \I) \w.
\end{equation}
We have to examine when~\eqref{liq} is in the subspace defined by \eqref{nwuu2}
and~\eqref{reqoq}, which in terms of $r_i$'s and $c_j$'s become
\begin{align}
\sum_{i}\alpha'_i = \sum_{j}\beta_j' &= \sum_{i}
\sqrt{\alpha_i} r_i + \sum_{j} \sqrt{\beta_j} c_j = 0,\label{pako}\\
\sum_{i}\alpha_i \frac{R_i'}{R_i} - \sum_{j}\beta_j
\frac{C_j'}{C_j} &= \sum_{i} \sqrt{\alpha_i} r_i - \sum_{j}
\sqrt{\beta_j} c_j = 0. \label{spako}
\end{align}
We give more detail on the derivation of~\eqref{pako} below. We
have
\begin{align*}
\sum_{i} \alpha'_i &= \sum_{i} \alpha_i
\frac{\alpha'_i}{\alpha_i}=
\sum_{i} \alpha_i \bigg(\frac{R_i'}{R_i} + \frac{\sum_{j}
B_{ij} C'_j}{\sum_{j} B_{ij} C_j}\bigg)=\sum_{i} r_i\sqrt{\alpha_i} + \sum_{i} \sum_{j} B_{ij}
R_i C'_j\\
&=\sum_{i} r_i\sqrt{\alpha_i} + \sum_{j}
\frac{c_j}{\sqrt{\beta_j}} \sum_{i} B_{ij} R_i C_j=\sum_{i} r_i\sqrt{\alpha_i} + \sum_{j} c_j\sqrt{\beta_j},
\end{align*}
the derivation for $\sum_{j}\beta_j'$ is analogous.
\end{proof}

\subsubsection{Attractive fixpoints of tree recursions}


The variables $R_i$, $C_j$, $\alpha_i$, $\beta_j$ in this section
refer to a priori different quantities as the variables in
Section~\ref{sdrg}. We feel that this conflict is justified since
we will establish that they coincide.

For convenience we repeat the tree recursions as stated in the introduction:
\begin{equation*}
\tag{\ref{kkrtko}}
\hat{R}_i \propto \bigg(\sum_{j=1}^q B_{ij}
C_j\bigg)^{\Delta-1}\quad\mbox{and}\quad \hat{C}_j \propto
\bigg(\sum_{i=1}^q B_{ij} R_j\bigg)^{\Delta-1}.
\end{equation*}
We are interested in the {\bf fixpoints} of the tree recursions,
that is, $R_i$'s and $C_j$'s such that
$$
\hat{R}_i \propto R_i\quad\mbox{and}\quad \hat{C}_j \propto C_j
$$
for all $i,j\in [q]$. Note that the fixpoints  correspond to
the critical points of $\Psi_1$ (using Theorem~\ref{new:zako1})).

Next we examine the stability of fixpoints. For a continuously
differentiable map a sufficient condition for a fixpoint to be
attractive is if the spectral radius of the derivative is less
than one at the fixpoint. We will use the following terminology
for fixpoints whose attractiveness is established using this
method.

\begin{definition}
A fixpoint $x$ of a function $f:{\cal M}\ra{\cal M}$ is called 
{\bf Jacobian attractive fixpoint} if the Jacobian of $f$ at $x$
has spectral radius less than $1$.
\end{definition}

\begin{lemma}\label{ztt2}
Let $(\rb,\cb)$ be a fixpoint of the tree recursions.
Let $\alpha_i = \sum_{j=1}^q B_{ij} R_i C_j$ and
$\beta_j = \sum_{i=1}^q B_{ij} R_i C_j$ and let $\L$ be the (matrix of the) map
defined by~\eqref{lmp}. We have that $(\rb,\cb)$ is Jacobian
attractive if and only if $(\Delta-1)\L$ has spectral radius
less than $1$ in the subspace of $\w=(r_1,\dots,r_q,c_1,\dots,c_q)$ that satisfy
\begin{equation*}\tag{\ref{pakoprob}}
\sum_{i=1}^q\sqrt{\alpha_i} r_i = 0\quad\mbox{and}\quad
\sum_{j=1}^q\sqrt{\beta_j} c_j = 0.
\end{equation*}
\end{lemma}

\begin{proof}
W.l.o.g. we can assume that $(\rb,\cb)$ is scaled so that
\begin{equation}\label{wuzzzo}
\sum_{i=1}^q\sum_{j=1}^q B_{ij} R_i C_j = 1.
\end{equation}
Note that the scaling does not affect the value of $\L$ nor does it affect the
constraint~\eqref{pakoprob}.

When we perturb the $R_i$'s and $C_j$'s and apply
one step of the tree recursion we obtain
\begin{equation}\label{qa}
\frac{\hat{R}'_i}{\hat{R}_i} = (\Delta-1)\frac{\sum_{j=1}^q
B_{ij}C_j \frac{C'_j}{C_j}}{\sum_{j=1}^q B_{ij} C_j}
\quad\mbox{and}\quad \frac{\hat{C}'_j}{\hat{C}_j} =
(\Delta-1)\frac{\sum_{i=1}^q B_{ij}R_i
\frac{R'_i}{R_i}}{\sum_{i=1}^q B_{ij} R_i}.
\end{equation}
 We can rewrite~\eqref{qa} as follows
\begin{equation}\label{qa2}
\frac{\hat{R}'_i}{\hat{R}_i} = (\Delta-1)\frac{\sum_{j=1}^q
B_{ij}R_iC_j \frac{C'_j}{C_j}}{\alpha_i} \quad\mbox{and}\quad
\frac{\hat{C}'_j}{\hat{C}_j} = (\Delta-1)\frac{\sum_{i=1}^q
B_{ij}R_iC_j \frac{R'_i}{R_i}}{\beta_j}.
\end{equation}
The perturbation that scales all $R_i$'s by the same factor does
not change the messages (since they are in the projective space)
and hence we need to exclude it when studying local stability
of~\eqref{qa}. Similarly scaling all $C_j$'s by the same factor
does not change the messages. We need to locate an invariant subspace
of~\eqref{qa2} whose complement corresponds to the scaling.
We obtain the following subspace (it corresponds to preserving~\eqref{wuzzzo}):
\begin{equation}\label{norm2}
\sum_{i=1}^q \alpha_i\frac{R'_i}{R_i} = 0
\quad\mbox{and}\quad
\sum_{j=1}^q \beta_j\frac{C'_j}{C_j} = 0.
\end{equation}
Now we check that~\eqref{norm2} is invariant under the map~\eqref{qa2}, indeed,
\begin{equation}\label{normpol}
\sum_{i=1}^q \alpha_i \frac{\hat{R}'_i}{\hat{R}_i} = (\Delta-1)
\sum_{i=1}^q \sum_{j=1}^q B_{ij} R_i C_j \frac{C_j'}{C_j} =
(\Delta-1)\sum_{j=1}^q \beta_j\frac{C'_j}{C_j}=0;
\end{equation}
the argument for $\sum_{j=1}^q \beta_j \frac{\hat{C}'_j}{\hat{C}_j}=0$
is analogous.

A fixpoint $(R_1,\dots,R_q,C_1,\dots,C_q)$ is Jacobian attractive
if the linear transformation
$$
\left(\frac{R_1'}{R_1},\dots,\frac{R_q'}{R_q},\frac{C_1'}{C_1},\dots,\frac{C_q'}{C_q}\right)\mapsto
\left(\frac{\hat{R}_1'}{\hat{R}_1},\dots,\frac{\hat{R}_q'}{\hat{R}_q},\frac{\hat{C}_1'}{\hat{C}_1},\dots,\frac{\hat{C}_q'}{\hat{C}_q}\right)
$$
given by \eqref{qa} has spectral radius less than $1$ in the
subspace defined by~\eqref{norm2}.

Let $r_i = \sqrt{\alpha_i} R_i' / R_i$, $c_j = \sqrt{\beta_j} C_j'/C_j$, $\hat{r}_i = \sqrt{\alpha_i} \hat{R}_i' / \hat{R}_i$,
and  $\hat{c}_j = \sqrt{\beta_j} \hat{C}_j'/\hat{C}_j$. This linear transformation of variables turns
\eqref{qa2} into
\begin{equation}\label{qa3}
\hat{r}_i = (\Delta-1)\sum_{j=1}^q
\frac{B_{ij}R_iC_j}{\sqrt{\alpha_i\beta_j}} c_j
\quad\mbox{and}\quad \hat{c}_j = (\Delta-1)\sum_{i=1}^q
\frac{B_{ij}R_iC_j}{\sqrt{\alpha_i\beta_j}} r_i.
\end{equation}
Note that \eqref{qa3} is $(\Delta-1)L$ where $L$ is the map
defined by~\eqref{lmp}. The constraint~\eqref{norm2} becomes \eqref{pakoprob}.
\end{proof}

\subsubsection{Connecting attractive fixpoints to maximum entropy configurations}\label{sec:attmax}

Now we are ready to prove Theorem \ref{thm:connection}.

\begin{proof}[Proof of Theorem \ref{thm:connection}]
Let $S$ be the linear subspace defined by~\eqref{pakoprob} (note
that~\eqref{pako} together with~\eqref{spako} define the same subspace). The
constraint for the fixpoint to be Jacobian attractive is that
$(\Delta-1)\L$ on $S$ has spectral radius less than $1$. The
constraint for the critical point to be Hessian maximum is that
the eigenvalues of $(\I + \L)( (\Delta-1) \L - \I)$ on $S$ are
negative (see equation~\eqref{liq}).

Note that $\L$ is symmetric and it is a result of tensor product
with the matrix $(\begin{smallmatrix} 0&1\\ 1&0 \end{smallmatrix})$.
Hence $\L$ has symmetric real spectrum (symmetry means that if $a$
is an eigenvalue then so is $-a$). Note that $S$ is invariant
under $\L$ and hence the spectrum of $\L$ on $S$ is a subset of the
spectrum of $\L$ (it is still symmetric real; the restriction wiped
out a pair of eigenvalues $-1$ and $1$).

The constraint for the fixpoint to be Jacobian attractive, in
terms of eigenvalues, is: for each eigenvalue $x$ of $\L$ on $S$
\begin{equation}\label{ooo1}
-1 < (\Delta-1) x < 1.
\end{equation}
The constraint for the critical point to be Hessian maximum, in
terms of eigenvalues, is: for each eigenvalue $x$ of $\L$ on $S$
\begin{equation}\label{ooo2}
(1+x)\big((\Delta-1)x-1\big) < 0\quad\mbox{and}\quad (1-x)\big(-(\Delta-1)x-1\big)
< 0,
\end{equation}
where the second constraint comes from the symmetry of the
spectrum (thus $-x$ is an eigenvalue). Note that
conditions~\eqref{ooo1} and~\eqref{ooo2} are equivalent (since
$(1+x)\big((\Delta-1)x-1\big)$ is negative for  $-1< x< 1/(\Delta-1)$).
\end{proof}

\section{Reduction for Colorings}
\label{sec:reduction}

In this section we outline our proof of Theorem \ref{thm:colorings}.
We start by reviewing the main components of the reduction for 2-spin
systems (as carried out in \cite{Sly,SS}) and in particular the
hard-core model. This will allow us to isolate the parts of the
argument which do not extend to the multi-spin case and motivate our
reduction scheme. The first step is a  reduction from max-cut to a
so-called phase labeling problem that we introduce. To present the
main ideas of this particular key reduction we first present it in
this section in the simplified setting of the colorings problem (see
Lemma \ref{lem:maxcuttophasecolorings}).

The basic gadget in the reduction is a bipartite random graph, which we denote by $G$. The sides of the bipartition have an equal number of vertices, and the sides are labelled with $+$ and $-$. Most vertices in $G$ have degree $\Delta$ but there is also a small number of degree $\Delta-1$ vertices (to allow to make connections between gadgets without creating  degree $\Delta+1$ vertices). For $s=\{+,-\}$, let the vertices in the $s$-side be $U^s\cup W^s$ where the vertices in $U=U^+\cup U^-$ have degree $\Delta$ and the vertices in $W=W^+\cup W^-$ have degree $\Delta-1$. The phase of an independent set $I$ is $+$ (resp. $-$) if $I$ has more vertices in $U^+$ (resp. $U^-$). Note that the phase depends only on the spins of the ``large" portion of the graph, i.e., the spins of vertices in $U$.

In non-uniqueness regimes, the gadget $G$ has two important properties, both of which can be obtained by building on the second moment analysis of Section~\ref{sec:secondmomentanalysis}. First, the phase of a random independent set $I$ is equal to $+$ or $-$ with probability roughly equal to $1/2$. Second, conditioned on the phase of a random independent set $I$, the spins of the vertices in $W$ are approximately independent, i.e., the marginal distribution on $W$ is close to a product distribution.
In this product distribution
if the phase is $+$ (resp. $-$), a vertex in $W^+$ is in $I$ with probability $p^+$ (resp. $p^-$), while a vertex in $W^-$ is in $I$ with probability $p^-$ (resp. $p^+$). The values $p^{\pm}$ correspond to maxima of the function $\Psi_1$ and, crucially (as we shall demonstrate shortly), they satisfy $p^+\neq p^-$.

Using the second moment analysis of Section~\ref{sec:secondmomentanalysis} and in particular Theorem~\ref{thm:second-moment}, we can prove that an analogous phenomenon takes place for the $k$-colorings model in the semi-translation non-uniqueness regime (the precise statement of the gadget's properties are given in Lemma~\ref{lem:gadgetcrucial}). The main difference is that, instead of two phases, the number of phases is equal to the number of maximizers of the function $\Psi_1$ (as described in Theorem~\ref{thm:fase}). In particular, for $k$ even, the phase of a coloring is determined by the dominant set of $k/2$ colors on $U^+$, i.e., the $k/2$ colors with largest frequencies among vertices of $U^{+}$. Each of the $\binom{k}{k/2}$ phases appears with roughly equal probability and given the phase, the marginal distribution on $W$ is close to a product distribution, which we now describe.
We can compute explicit values $a'=a'(k,\Delta), b'=b'(k,\Delta)$ such that for a phase  $T\in \binom{[k]}{k/2}$ the probability
mass function $\x$ of a vertex in $W^+$ has its $i$-th entry equal to $a'$ if $i\in T$ and equal to $b'$ if $i\notin T$.
Similarly, the probability mass function $\y$ of a vertex in $W^-$ has its $i$-th entry equal to $b'$ if $i\in T$ and equal to $a'$ if $i\notin T$. (The values $a',b'$ correspond to the values $a,b$ described in Item~\ref{itt:dominant} of Theorem~\ref{thm:fase}, the correspondence is obtained using \eqref{jqwe} in Theorem~\ref{new:zako1}.\footnote{\label{foot:correspondence}In particular, $a',b'$ can be readily obtained from $a,b$ using the relations $a=a'^{\Delta/\Delta-1}/S$, $b=a'^{\Delta/\Delta-1}/S$, $\frac{q}{2}(a'+b')=1$,  where $S:=\frac{q}{2}(a'^{\Delta/\Delta-1}+b'^{\Delta/\Delta-1})$.})

Let $\Qc$ be the union of the pairs $(\x,\y)$ over all dominant phases. Hereafter, we will identify the phases with elements of $\Qc$. Note that if $(\x,\y)\in \Qc$, then $(\y,\x)\in \Qc$ as well. We also denote by $\Qc'$ the union of unordered elements of $\Qc$.
Elements of $\Qc'$ are called unordered phases (we use $\p$ to denote unordered phases).
Given a phase $\p=\{\x,\y\}$ an ordering of the pair will be
called ``assigning spin to the phase''. The two ordered phases corresponding to the unordered phase $\p$ will be denoted by $\p^+$ and $\p^-$.

The conditional independence property is crucial, it allows us to quantify the effect of using vertices of $W$ as terminals to make connections between copies of the gadget $G$. For example, consider the following type of connection, which we refer to as parallel.  Let  $v^+\in W^+$, $v^-\in W^-$ and consider  two copies of the gadget $G$, say $G_1,G_2$. For $i=1,2$ denote by $v^+_i,v^-_i$ the images of $v^+,v^-$ in $G_i$. Now add the edges  $(v^+_1,v^+_2)$ and $(v^-_1,v^-_2)$ and denote the final graph by $G_{12}$.  Thus, a parallel connection corresponds to joining the $+,+$ and $-,-$ sides of two copies of the gadget.

Clearly, random colorings of $G_{12}$ can be generated by first generating random colorings of $G_1,G_2$ and keeping the resulting coloring if $v^\pm_1,v^\pm_2$ have different colors. We thus have that the partition function of $G_{12}$ is equal to $(Z_G)^2$ times the probability that $v^\pm_1,v^\pm_2$ have different colors in random colorings of $G_1,G_2$. The latter quantity can easily be computed if we condition on the phases $(\x_1,\y_1), (\x_2,\y_2)$ of the colorings in $G_1,G_2$, and this is equal to $(1-\x_1^{\T}\x_2)(1-\y_1^{\T}\y_2)$.

By taking logarithms, we can assume  a parallel connection between gadgets with phases $(\x_1,\y_1)$ and $(\x_2,\y_2)$  incurs an (additive) weight
\begin{equation*}
w_p((\x_1,\y_1), (\x_2,\y_2))=\ln
(1-\x_1^{\T}\x_2) + \ln (1-\y_1^{\T} \y_2).
\end{equation*}

In the hard-core model, parallel connections are sufficient to give hardness. In this case, we have that $\Qc'=\{\p\}$ and $\Qc=\{\p^+,\p^-\}$ and the respective function $w_p(\cdot,\cdot)$ satisfies
\begin{equation}\label{eq:antioptimal}
w_p(\p^+,\p^+)=w_p(\p^-,\p^-)<w_p(\p^+,\p^-).
\end{equation}
Thus, in this case, $w_p(\cdot,\cdot)$ takes only two values and neighboring gadgets prefer to have different phases. Now assume that $H$ is an instance of $\textsc{Max-Cut}$ and replace each vertex in $H$ by a copy of the gadget $G$, while for each edge of $H$, connect the respective gadgets in parallel. The partition function of the final graph is dominated from phase assignments which correspond to  large
 cuts in $H$. This intuition is the basis of the reduction in \cite{Sly,SS}.

For the colorings model, reducing from \textsc{Max-Cut} poses an extra challenge. While for every unordered phase $\p$
equation \eqref{eq:antioptimal} continues to hold, a short calculation shows that the optimal configuration for a triangle of gadgets connected in parallel is to give all three gadgets different phases. To bypass this entanglement, we need to introduce some sort of ferromagnetism  in the reduction to enforce gadgets corresponding to vertices of $H$ to use a single (unordered) phase. To achieve this, we use \emph{symmetric} connections, which correspond to having not only $(+,+), (-,-)$ connections of the gadgets, but also $(+,-)$ and $(-,+)$. Thus, a symmetric connection whose endpoints have phases $(\x_1,\y_1), (\x_2,\y_2)$ incurs (additive) weight
\begin{equation*}
w_s((\x_1,\y_1), (\x_2,\y_2))=w_p((\x_1,\y_1), (\x_2,\y_2))+w_p((\x_1,\y_1), (\y_2,\x_2)).
\end{equation*}
Symmetric connections will allow us to enforce a single unordered phase to all gadgets, while parallel connections will allow us to recover a maximum-cut partition. To have some modularity in our construction, rather than reducing from \textsc{Max-Cut} directly, we use the following ``phase labeling problem".

\noindent
{\sc Colorings Phase Labeling Problem($\B,\Qc)$:}\\
INPUT: undirected edge-weighted multigraph $H=(V,E)$ and a partition of the edges $\{E_p,E_s\}$.\\
OUTPUT: $\MLwt(H):=\max_{\Yc}\Lwt_{H}(\Yc)$, where the maximization is over all possible phase labelings $\Yc:V\rightarrow {\cal Q}$ and
\[\Lwt_{H}(\Yc):=\sum_{\{u,v\}\in E_s} w_s(\Yc(u),\Yc(v))+\sum_{\{u,v\}\in E_p} w_p(\Yc(u),\Yc(v)).\]

Edges in $E_p$ (resp. $E_s$) correspond to parallel (resp. symmetric) connections and we shall refer to them as parallel (resp. symmetric) edges. The arguments in \cite{SS}, which we sketched earlier, can easily be adapted to show that an algorithm for approximating the partition function to an  arbitrarily small exponential factor yields a PTAS for the phase labeling problem,
see Lemma \ref{lem:phaselabelinglemma} and its proof in Section \ref{sec:slysunstuff}.
It then remains to prove that a PTAS for the phase labeling problem yields a PTAS for \textsc{Max-Cut} on 3-regular graphs. This is the scope of the next lemma, which we focus on proving in the remainder of this section.

\begin{lemma}\label{lem:maxcuttophasecolorings}
A (randomized) algorithm that approximates  the solution to the {\sc Colorings Phase Labeling Problem($\B,\Qc)$} on bounded degree graphs  within a factor of $1-o(1)$ yields a (randomized) algorithm that approximates \textsc{MaxCut} on 3-regular graphs within a factor of $1-o(1)$.
\end{lemma}

Our reduction relies on the following gadget which ``prefers'' the unordered phase of two
distinguished vertices $u$ and $v$ to agree.
For a phase assignment $\Yc$ with ordered phases, we denote by $\Yc'$ the respective phase assignment with unordered phases.

\begin{lemma}\label{hrk2col}
A constant sized gadget $J_1$ with two distinguished vertices $u,v$ can be constructed with the following property: all edges of $J_1$ are symmetric and the following is true,
\begin{equation}\label{eq:property}
\max_{\Yc;\, \Yc'(u)=\Yc'(v)} \Lwt_{J_1}(\Yc) > \eps_1 + \max_{\Yc;\, \Yc'(u)\neq \Yc'(v)} \Lwt_{J_1}(\Yc),
\end{equation}
where $\eps_1>0$ is a constant depending only on $k$ and $\Delta$.
\end{lemma}

We give the proof of the critical Lemma~\ref{hrk2col} after the (simpler) proof of Lemma~\ref{lem:maxcuttophasecolorings}.
\begin{proof}[Proof of Lemma~\ref{lem:maxcuttophasecolorings}]
Let $\epsilon_1$ be as in Lemma~\ref{hrk2col} and  
\[t := 2 \lceil (\max_{\p_1,\p_2} w_p(\p_1,\p_2) - \min_{\p_1,\p_2} w_p(\p_1,\p_2) )/\eps_1\rceil.\]

Given a $3$-regular instance $H=(V,E)$ of \textsc{Max-Cut}, we first declare all edges of $H$  to be parallel. Moreover, for every edge $(u',v')$ of $H$, take $t$ copies of gadget $J_1$ from Lemma~\ref{hrk2col}, identify (merge) their $u$ vertices with $u'$, and identify (merge) their $v$ vertices with $v'$. Let $H'$ be the final graph.  

To find the optimal phase labeling of $H'$, we may focus on the phase assignment restricted to vertices in $H$, since each gadget $J_1$ can be independently set to its optimal value conditioned on the phases for its distinguished vertices $u$ and $v$. We claim that 
\begin{equation}\label{eq:mlwthp}
\MLwt(H')=C_1\textsc{MaxCut}(H)+(C_2 +C_3 t) |E|,
\end{equation}
for constants $C_1,C_2,C_3$ to be specified later (depending only on $k,\Delta$). Using the trivial bound $\textsc{MaxCut}(H)\geq |E|/2=3|V|/4$, the lemma follows easily from \eqref{eq:mlwthp}. We thus focus on proving \eqref{eq:mlwthp}.

The key idea is that for any phase labeling $\Yc: V\rightarrow \Qc$, changing the unordered phases of vertices in $H$ to the same unordered phase $\p\in \Qc'$, while keeping the spins, can only increase the weight of the labeling. 
Indeed, for $(u,v)\in E$ such that $\Yc'(u)=\Yc'(v)$, no change in the weight of the labeling occurs, using \eqref{eq:property}. For $(u,v)\in E$ such that $\Yc'(u)\neq \Yc'(v)$, the potential (weight) loss from the parallel edge $(u,v)$ is compensated by the gain on the $t$ copies of $J_1$ by \eqref{eq:property} and the choice of $t$. 

For phase labelings which assign vertices of $H$ the same unordered phase $\p$, to attain the maximum weight for a phase labeling, we only need to choose the spins, in order to maximize the contribution from parallel edges (the edges of $H$). The same argument we discussed  for the hard-core model, \eqref{eq:antioptimal} yields  that the optimal choice of spins to the phases induces a maximum-cut partition of $H$. For such a spin assignment, the contribution from parallel edges is $C_1\textsc{MaxCut}(H)+C_2|E|$, where  \[C_1:=w_p(\p^+,\p^-)-w_p(\p^-,\p^-)\mbox{ and }C_2:=w_p(\p^-,\p^-).\] The contribution from symmetric edges is $C_3t|E|$, where 
\[ C_3:=\max_{\Yc;\, \Yc'(u)=\Yc'(v)=\p} \Lwt_{J_1}(\Yc).
\] This proves \eqref{eq:mlwthp}.
\end{proof}

We conclude this section by giving  the proof of Lemma~\ref{hrk2col}.

\begin{proof}[Proof of Lemma~\ref{hrk2col}]
Let $\Qc':=\{\p_1,\hdots,\p_{Q'}\}$ and $\p_i:=\{\x_i,\y_i\}$ for $i\in [Q']$. Denote by $K$ the multigraph on $Q'$ vertices $b_1,b_2,\hdots,b_{Q'}$ with the following symmetric edges: self-loop on $b_i$ for $i\in [Q']$ and two edges between $b_i$ and $b_j$ for every $i,j\in [Q']$ with $i\neq j$.  We first prove that the optimal phase assignments $\Yc$ of $K$ are those which assign each vertex $b_i$ a distinct phase from $\Qc'$ (note that the spin of the phase does not matter since all edges of $K$ are symmetric). The desired gadget $J_1$ will be constructed afterwards.

Let $\Yc$ be a phase labeling of $K$ and  $s_{i}$ be the number of vertices assigned phase $\p_i$.
Denote by $\s$ the vector $(s_1,\hdots,s_{Q'})^{\T}$.  Note that $\ones^{\T}\s=Q'$, where $\ones$ is the all one vector with dimension $Q'$. Then
\begin{equation*}
\Lwt_K(\Yc)=\sum_{i,j\in[Q']}s_is_jw_s(\p_i,\p_j)=\s^{\T}\A\s,
\end{equation*}
where $\A$ is the $Q'\times Q'$ matrix whose $(i,j)$ entry equals $w_s(\p_i,\p_j)$. Note that $\A$ is symmetric and
$\ones$ is an eigenvector of $\A$ (because of the transitive symmetry of phases).
Moreover, if we let $\s'=\s - \ones$, then $\ones^{\T}\s'=0$. It follows that
\begin{equation}
\label{eq:allones}
\s^{\T}\A\s=\ones^{\T}\A\ones+(\s')^{\T}\A\s'.
\end{equation}
If $\A$ is negative definite, equation \eqref{eq:allones} shows that the all ones labeling is better than any other labeling.
Hence the result will follow if we prove that $\A$ is negative definite.

Let $\z_1,\hdots,\z_Q:=\x_1,\hdots,\x_{Q'},\y_1,\hdots,\y_{Q'}$ and let $\hat{\A}$ be the $Q\times Q$ matrix whose $ij$-entry is $\ln(1-\z_i^{\T}\z_j)$. Using the definition of the weights $w_s(\cdot,\cdot)$, it is easy to check that for any vector $\s$ it holds that
\[\s^{\T}\A\s= (\s,\s)^{\T}\hat{\A}(\s,\s),\]
so it suffices to prove that $\hat{\A}$ is negative definite. We will show here that $\hat{\A}$ is negative semi-definite; the proof that $\hat{\A}$ is regular (and hence negative definite) is trickier and is given in the proof of the more general Lemma~\ref{lem:antiferromagneticneg}.  Note that the entries of $\hat{\A}$ are obtained by applying $z\mapsto \ln(1-z)$ to
each entry of the Gram matrix of the vectors $\z_1,\hdots,\z_Q$.
Since for $|z|<1$ we have $\ln(1-z)=-z-z^2/2-z^3/3-\hdots$, by Schur's product theorem
(see Corollary 7.5.9 in~\cite{HJ}) we obtain that $\hat{\A}$ is negative semi-definite, as desired.

To construct the gadget $J_1$, we overlay two copies of $K$ as follows. Let $K_u$ (resp. $K_v$) be a copy of $K$, where the image of $b_{Q'}$ is renamed to $u$ (resp. $v$). Overlay $K_u,K_v$ by identifying the images of $b_1,\hdots,b_{Q'-1}$ in the two copies.  Thus, the resulting graph $J_1$ has two self loops on $b_i$ for $i\in[Q'-1]$, four edges between $b_i$ and $b_j$ for every $i,j\in [Q'-1]$ with $i\neq j$, two edges between $u$ and $b_i$ for $i\in [Q'-1]$, two edges between $v$ and $b_i$ for $i\in [Q'-1]$ and a self loop on $u,v$.

Note that for every phase labeling $\Yc$ of $J_1$, we have $\Lwt_{J_1}(\Yc)=\Lwt_{K_u}(\Yc)+\Lwt_{K_v}(\Yc)$ and hence $\MLwt(J_1)\leq 2\MLwt(K)$. Using that the optimal phase labelings for $K$ are those which assign each vertex a distinct phase from $\Qc'$, we obtain that the inequality holds at equality for  those (and only those) phase labelings which assign $u,v$ a common phase $\p\in \Qc'$ and vertices $b_1,\hdots,b_{Q'-1}$ a distinct phase from $\Qc'-\{\p\}$. This yields the $\epsilon_1$ in the statement of the lemma. Note that $\epsilon_1$ depends only on $\Qc'$, which in turn is completely determined by $k,\Delta$.
\end{proof}

\section{General Reduction} 
\label{sec:generalreduction}

\subsection{Phase labeling Problem}\label{sec:phaseproblem}
We first introduce the phase labeling problem  for a general antiferromagnetic spin system (which satisfies the hypotheses of Theorem~\ref{thm:general-inapprox}). As in the case for the colorings model (see Section~\ref{sec:reduction}), we let $\Qc'$ be the union of $\{\x,\y\}$ over all phases, i.e.,
$${\cal Q}'=\{\{\x_1,\y_1\},\dots,\{\x_{Q'},\y_{Q'}\}\}.$$
Henceforth, we will refer to elements of ${\cal Q}'$ as phases. Note that for fixed $q,\Delta,\B$ the global maxima of $\Psi_1$ correspond to fixpoints of \eqref{kkrtko} and hence can be approximated to any desired polynomial accuracy of their values. The values of $\x,\y$ may then be recovered using \eqref{jqwe} (see Footnote~\ref{foot:correspondence} for an explicit description of the correspondence in the case of colorings). The assumption of Theorem~\ref{thm:general-inapprox} translates into $\x_i\neq \y_i$ for all $i\in [Q']$.




Given an unordered phase $\{\x,\y\}$ an ordering of the pair will be
called ``assigning spin to the phase''. Let
$${\cal Q}=\{ (\x_1,\y_1),\dots,(\x_{Q},\y_{Q}) \}$$
be the collection of ordered phases.  Note that $Q=2Q'$. We will denote unordered phases using $\p$; the two ordered phases corresponding to the unordered phase $\p$ will be denoted by $\p^+$ and $\p^-$. Given a graph $H$ with vertex set $V$ we will assign ordered phases to its vertices---the labeling
(called phase assignment) will be denoted by $\Yc:V\rightarrow{\cal Q}$.
The corresponding labeling by unordered phases (where the ordering is removed) will be
denoted by $\Yc'$.

Now we define the weight of a phase assignment. We will have two types of edges in $H$: parallel or symmetric; the type of an edge will only impact the weight of a
phase assignment. In particular, a parallel edge whose endpoints
have labels $(\x_1,\y_1)$ and $(\x_2,\y_2)$ incurs weight
\begin{equation*}
w_p((\x_1,\y_1), (\x_2,\y_2))=\ln
(\x_1^{\T}\B \x_2) + \ln (\y_1^{\T} \B \y_2),
\end{equation*}
while a symmetric edge incurs weight
\begin{equation*}
w_s((\x_1,\y_1),
(\x_2,\y_2))=w_p((\x_1,\y_1),
(\x_2,\y_2))+w_p((\x_1,\y_1), (\y_2,\x_2)).
\end{equation*}
Note that if we flip $(\x_1,\y_1)$, that is, replace it by
$(\y_1,\x_1)$, the weight of the symmetric edge does not change.

We will use the following problem in our reduction.

\noindent
{\sc Phase Labeling Problem($\B,\Qc)$:}\\
INPUT: undirected edge-weighted multigraph $H=(V,E)$ and a partition of the edges $\{E_p,E_s\}$.\\
OUTPUT: $\MLwt(H):=\max_{\Yc}\Lwt_{H}(\Yc)$, where the maximization is over all possible phase labelings $\Yc:V\rightarrow {\cal Q}$ and
\[\Lwt_{H}(\Yc)=\sum_{\{u,v\}\in E_s} w_s(\Yc(u),\Yc(v))+\sum_{\{u,v\}\in E_p} w_p(\Yc(u),\Yc(v)).\]

The motivation for the {\sc Phase Labeling problem} is the following lemma. The proof roughly follows the lines of \cite{SS} and is given in Section~\ref{sec:slysunstuff}.

\begin{lemma}\label{lem:phaselabelinglemma}
In the setting of Theorem~\ref{thm:general-inapprox}, the following holds. A  (randomized) algorithm that approximates the partition function on triangle free $\Delta$-regular graphs  within an arbitrarily small exponential factor yields a (randomized) algorithm that approximates  the solution to the phase labeling problem with parameters $\B,\Qc$ on bounded degree graphs  within a factor of $1-o(1)$.
\end{lemma}

The following lemma requires more work in our setting and is proved in Section~\ref{sec:reductioncore}.
\begin{lemma}\label{lem:maxcuttophase}
A (randomized) algorithm that approximates  the solution to the phase labeling problem with parameters $\B,\Qc$ on bounded degree graphs  within a factor of $1-o(1)$ yields a (randomized) algorithm that approximates \textsc{MaxCut} on 3-regular graphs within a factor of $1-o(1)$.
\end{lemma}

Using Lemmas~\ref{lem:phaselabelinglemma} and~\ref{lem:maxcuttophase}, we obtain Theorem~\ref{thm:general-inapprox}.
\begin{proof}[Proof of Theorem~\ref{thm:general-inapprox}]
Suppose that there exists a (randomized) algorithm to approximate the partition function on $\Delta$-regular graphs with interaction matrix $\B$ up to an arbitrarily small exponential factor. Then, combinining Lemmas~\ref{lem:phaselabelinglemma} and~\ref{lem:maxcuttophase}, we obtain a (randomized) algorithm to approximate \textsc{MaxCut} on 3-regular graphs within a factor of $1-o(1)$. This contradicts the result of \cite{APXhard}.
\end{proof}

\subsection{Properties of Antiferromagnetic Spin Systems}
\label{sec:antiferromagnetic}

In this section we prove two basic properties of antiferromagnetic systems that will be used in our general reductions.

As a consequence of the Perron-Frobenius theorem and the antiferromagnetism definition (cf. Definition~\ref{def:antiferromagnetic}), we may decompose the interaction matrix $\B$ of an antiferromagnetic model as
\begin{equation}\label{eq:antiferromagneticB}
\B= \u \u^{\T} - \P^{\T} \P,
\end{equation}
where the vector $\u$ has positive entries and $\P$ is a square matrix. Using the decomposition \eqref{eq:antiferromagneticB}, we prove the following two lemmas which are used in the reduction.

\begin{lemma}\label{lem:antiferromagnetic}
For antiferromagnetic $\B$, and vectors $\z_1,\z_2\in \mathbb{R}^q_{\geq0}$ with $\norm{\z_1}_1=\norm{\z_2}_1=1$, we have
\[(\z_1^{\T} \B \z_1)(\z_2^{\T} \B \z_2)\leq (\z_1^\T \B \z_2)^2.\]
Equality holds iff $\z_1=\z_2$.
\end{lemma}

\begin{proof}
Set $\w_1=\P \z_1,\, \w_2=\P \z_2,\, a_1=\u^{\T}\z_1,\, a_2=\u^{\T}\z_2$. Then
\[\z_1^{\T} \B \z_1=a^2_1-\w_1^{\T}\w_1,\ \z_2^{\T} \B \z_2=a_2^2-\w_2^{\T}\w_2, \ \z_1^{\T} \B \z_2=a_1a_2-\w_1^\T\w_2.\]
Since $\B,\z_1,\z_2$ have nonnegative entries, the above equalities imply $a^2_1-\w_1^{\T}\w_1,a_2^2-\w_2^{\T}\w_2,a_1a_2-\w_1^{\T}\w_2\geq0$. The inequality reduces to
\[ \big(a^2_1-\w_1^{\T}\w_1\big)\big(a_2^2-\w_2^{\T}\w_2\big)\leq\big(a_1a_2-\w_1^{\T}\w_2\big)^2.\]
This is known as Acz\'{e}l's inequality. The fastest proof goes as follows: set $b^2_1=a^2_1-\w_1^{\T}\w_1$ and $b_2^2=a_2^2-\w_2^{\T}\w_2$, so that by Cauchy-Schwarz $a_1a_2\geq b_1b_2+\w_1^{\T}\w_2$, implying the inequality.

Equality can only hold if $a_1=\lambda a_2$ and $\w_1=\lambda \w_2$, yielding $\u^{\T}(\z_1-\lambda\z_2)=0$ and $\P(\z_1-\lambda\z_2)=0$. We easily obtain $\B(\z_1-\lambda\z_2)=0$ and since $\B$ is invertible, $\z_1=\lambda \z_2$. The assumption $\norm{\z_1}_1=\norm{\z_2}_1=1$ implies $\lambda=1$, as wanted.
\end{proof}

\begin{corollary}\label{col:antiferromagnetic}
By plugging in the inequality of Lemma~\ref{lem:antiferromagnetic} the vectors with a single 1 in the positions $i$ and $j$ respectively, we obtain that any two spins $i,j$ induce an antiferromagnetic two-spin system.
\end{corollary}

\begin{lemma}\label{lem:antiferromagneticneg}
Let $\z_1,\dots,\z_n\in\R^d$ be a collection of distinct non-negative vectors such that
$\|\z_i\|_1=1$ for $i\in [n]$. Let $a_i=\z_i^{\T} \u$, where $\u$ is as in \eqref{eq:antiferromagneticB}. Let $\A'$ be the $n\times n$ matrix whose $ij$-th entry is $\ln (\z_i^{\T}\B \z_j)-\ln(a_i)-\ln(a_j)$. Then $\A'$ is negative definite.
\end{lemma}

\begin{proof}
Let $\w_i = \frac{1}{a_i} \P \z_i$ and let $\W$ be the $q\times n$ matrix whose columns are $\w_1,\dots,\w_n$. We
first argue $\w_i\neq \w_j$ for $i\neq j$. Suppose $\w_i=\w_j$. Let $\z=\frac{1}{a_i} \z_i - \frac{1}{a_j} \z_j$.
We have $\P \z = \w_i - \w_j = 0$ and $\u^{\T} \z = 1 - 1 = 0$ and hence $\B \z =0$. Since $\B$ is regular we
have $\z=0$. Thus $0 = \z^{\T} \oneb = \frac{1}{a_i}-\frac{1}{a_j}$ which implies $a_i=a_j$ which in turn
implies $\z_i=\z_j$, a contradiction. Thus $\w_i\neq \w_j$ for $i\neq j$.

Note that we have
$$
\ln (1-\w_i^{\T} \w_j) = \ln (a_i a_j - \z_i^{\T} \P^{\T} \P \z_j) - \ln (a_i a_j) = A'_{ij}.
$$
Thus the $ij$-th entry in $\A'$ is obtained by applying $z\mapsto\ln (1-z)$ to each entry of the
Gramm matrix $\W^{\T} \W$ . Note that for $|z|<1$ we have $\ln (1-z) = -z - z^2/2 - z^3/3 - \dots$
and hence by Schur product theorem $\A'$ is negative semi-definite (see Corollary 7.5.9 in~\cite{HJ}).

Now we argue that $\A'$ is regular (and hence negative definite). We have
\begin{equation}\label{uako}
- \A' = \sum_{k=1}^\infty \frac{1}{k} \W_k^{\T} \W_k,
\end{equation}
where $\W_k$ is the $q^k\times n$ matrix whose columns are $w_1^{\otimes k},\dots,w_n^{\otimes k}$. Note that
if $\A'$ is singular then there exists a non-zero vector $v$ such that $\v^{\T} \A' \v=0$ and for this to happen
we would have to have
\begin{equation}\label{pin1}
\W_k \v = 0
\end{equation}
for all $k\geq 1$ (the terms on the right-hand side of~\eqref{uako}
are non-negative and if even one of them is positive then $\v^{\T} \A' \v < 0$).

There exists a vector $\r\in\R^q$ such that $\alpha_i = \r^{\T} \w_i, i=1,\dots,n$ are distinct real numbers
(the $\w_i$'s are distinct and hence for any $i\neq j$ the measure of $r\in [0,1]^q$ such that $\r^{\T} \w_i = \r^{\T} \w_j$
is zero). Note that $(\r^{\otimes k})^{\T} \W_k$ is $(\alpha_1^k,\dots,\alpha_n^k)$. From~\eqref{pin1}
we obtain that for every integer $k\geq 1$ we have $(\alpha_1^k,\dots,\alpha_n^k) \v = 0$
and hence $\v=0$ (by considering the Vandermonde matrix $\{\alpha_{i}^k\})$, a contradiction.
Hence $\A'$ is regular and negative definite.
\end{proof}

\subsection{Reducing \textsc{MaxCut} to \textsc{Phase Labeling}}
\label{sec:reductioncore}

In this section, we prove Lemma~\ref{lem:maxcuttophase}.

\subsubsection{An intermediate gadget}
We will use the following gadget which ``prefers'' the unordered
phase of two vertices to agree. 
\begin{lemma}\label{hrk2}
A constant sized gadget $J_1$ with two distinguished vertices $u,v$ can be constructed with the following property: all edges of $J_1$ are symmetric and the following is true,
\begin{equation}\label{ooow1}
\max_{\Yc;\, \Yc'(u)=\Yc'(v)} \Lwt_{J_1}(\Yc) > \eps_1 + \max_{\Yc;\, \Yc'(u)\neq \Yc'(v)} \Lwt_{J_1}(\Yc),
\end{equation}
where $\eps_1>0$ is a constant depending only on the spin model and $\Delta$.
\end{lemma}
Note that Lemma \ref{hrk2col} which was proved in Section~\ref{sec:reduction} is a special case of Lemma \ref{hrk2} in the case of the colorings model. The proof of Lemma \ref{hrk2} follows roughly the same lines with slightly more intricate technical details.
\begin{proof}[Proof of Lemma~\ref{hrk2}.]
Let $\z_1,\dots,\z_Q := \x_1,\dots,\x_{Q'},\y_1,\dots,\y_{Q'}$. Let $\u$ be defined as in Equation~\eqref{eq:antiferromagneticB}.
In Section~\ref{sec:antiferromagnetic}, Lemma~\ref{lem:antiferromagneticneg}
it is proved that the $Q\times Q$ matrix $\hat{\A}$ whose $ij$-th entry is $\ln (\z_i^{\T}\B \z_j)-\ln(\z_i^{\T} \u)-\ln(\z_j^{\T} \u)$ is negative definite.
Let $\A'$ be the $Q'\times Q'$ matrix obtained by the following ``folding'' of $\hat{\A}$:
$$\A'_{ij} = \hat{\A}_{i,j} + \hat{\A}_{i+Q',j} + \hat{\A}_{i,j+Q'} + \hat{\A}_{i+Q',j+Q'}.$$
We have that $\A'$ is also negative definite (since $\x^{\T}\A'\x = \y^{\T} \hat{\A} \y'$, where $\y^{\T} = (\x^{\T},\x^{\T})$).
Note that
$$
\A'_{ij} = w_s((\x_i,\y_i),(\x_j,\y_j)) - a'_i - a'_j,
$$
where $a'_i:=2 \ln (\x_i^{\T} \u) + 2\ln (\y_i^{\T} \u)$.

Let $\lambda_1$
be largest eigenvalue of $-\A'$ and let $\lambda_2$ be the smallest eigenvalue of $-\A'$. Note that $0<\lambda_2\leq\lambda_1$.
Define $\A$ to be the $Q'\times Q'$ matrix with $A_{ij} = A'_{ij} + a'_i + a'_j$ and consider the following maximization
problem
\begin{equation}\label{hoju1}
\max_{\x; \x^{\T}\oneb = 1, \x\geq 0} \x^{\T} \A \x.
\end{equation}
 Note that for $\x$ with $\x^{\T} \oneb = 1$ we have
\begin{equation}\label{hoj}
\x^{\T} \A \x = 2 \a'^{\T} \x + \x^{\T} \A' \x,
\end{equation}
where $\A'$ is negative definite. Note that if $\x$ and $\y$ are distinct optimal solutions of~\eqref{hoju1}
then $(x+y)/2$ satisfies all the constraints, and from~\eqref{hoj} and negative definiteness
of $A'$ we have
$$
( (\x+\y)/2)^{\T} \A ((\x+\y)/2) > \left(\x^{\T} \A \x + \y^{\T} \A \y\right)/2,
$$
a contradiction (with optimality of both $x$ and $y$). Thus~\eqref{hoju1} has a unique maximum; let $\x^*$ be the value of $\x$ achieving it.
Let $O^*$ be $(\x^*)^{\T} \A \x^*$. Let $S$ be the set of non-zero coordinates in $\x^*$.

Let $\y\in \mathbb{R}^{Q'}$ be such that $\y^{\T} \oneb = 0$ and $\y$ is zero on coordinates outside $S$. Then
from (local) optimality of $\x^*$ we have
\begin{equation}\label{zoj}
(\x^*+\y)^{\T} \A (\x^*+\y) = O^* + 2 (\a'^{\T}  + (\x^*)^{\T} \A') \y + \y^{\T} \A' \y = O^* + \y^{\T} \A' \y \geq O^* - \lambda_1 \|\y\|_2^2.
\end{equation}
Equation~\eqref{zoj} tells us that moving slightly from the optimum the objective
decreases at most quadratically in the length of $\y$.

Let $\y\in \mathbb{R}^{Q'}$ be such that $\y^{\T} \oneb =0$ and $\y$ is non-negative on coordinates outside $S$.
Then
from (local) optimality of $\x^*$ we have
\begin{equation}\label{zoj2}
(\x^*+\y)^{\T} \A (\x^*+\y) = O^* + 2 (\a'^{\T}  + (\x^*)^{\T} \A') \y + \y^{\T} \A' \y = O^* + \y^{\T} \A' \y \geq O^* - \lambda_2 \|\y\|_2^2.
\end{equation}
Equation~\eqref{zoj} tells us that moving slightly from the optimum the objective
decreases at least quadratically in the length of $\y$.

Let $Z\geq (4 Q'\lambda_1/\lambda_2)^{Q'}$. Note that $Z$ is a constant depending only on the
model and $\Delta$. Let $z_1/z,\dots,z_{Q'}/z$ be the optimal simultaneous Diophantine approximation
of $x^*_1,\dots,x^*_{Q'}$ with $z_1,\dots,z_{Q'},z\in{\mathbb Z}$ and $1\leq z\leq Z$. By Dirichlet's theorem we have
\begin{equation}\label{hamp9}
\left| z x^*_i - z_i\right| \leq Z^{-1/Q'} < 1.
\end{equation}
Note that~\eqref{hamp9} implies
\begin{equation}\label{mm12}
\mbox{if $x^*_i=0$ then $z_i=0$.}
\end{equation}
Also note that
$$
\left|\sum_{i=1}^{Q'} z x^*_i - \sum_{i=1}^{Q'} z_i\right| \leq \sum_{i=1}^{Q'} \left | z x^*_i - z_i\right|  \leq Q' Z^{-1/Q'} < 1,
$$
and since $z$ and $z_i$'s are integers and $(\x^*)^{\T} \oneb =1$ we have
\begin{equation}\label{vone}
\sum_{i=1}^{Q'} \frac{z_i}{z} = 1.
\end{equation}
From~\eqref{mm12} and~\eqref{vone} we have that for $\y:=(z_1/z,\dots,z_{Q'}/z) - \x^*$
we can apply~\eqref{zoj} and hence
\begin{equation}
(z_1/z,\dots,z_{Q'}/z) \A (z_1/z,\dots,z_{Q'}/z)^{\T} \geq O^* - \lambda_1 Q' Z^{-2/Q'} z^{-2}.
\end{equation}
Now we are ready to construct the gadget $J_1$. First, let $K$ be the multigraph on $z$ vertices $b_1,b_2,\hdots,b_{z}$ with the following symmetric edges: self-loop on $b_i$ for $i\in [z]$ and two edges between $b_i$ and $b_j$ for every $i,j\in [z]$ with $i\neq j$. To obtain $J_1$, we overlay two copies of $K$ as follows. Let $K_u$ (resp. $K_v$) be a copy of $K$, where the image of $b_{z}$ is renamed to $u$ (resp. $v$). Overlay $K_u,K_v$ by identifying the images of $b_1,\hdots,b_{z-1}$ in the two copies.  Thus, the resulting graph $J_1$ has $z+1$ vertices and the following edges: two self loops on $b_i$ for $i\in[z-1]$, four edges between $b_i$ and $b_j$ for every $i,j\in [Q'-1]$ with $i\neq j$, two edges between $u$ and $b_i$ for $i\in [z-1]$, two edges between $v$ and $b_i$ for $i\in [z-1]$ and a self loop on $u,\, v$.

Note that the weight of a phase assignment on $J_1$ is the sum of the induced phase assignments on $K_u$ and $K_v$. Consider an assignment of phases $\Yc_{o}$ such that in each complete graph
$z_i$ vertices get phase $i$ (note that this forces the phases of $u$ and $v$ to be the same).
The weight of the phase assignment $\Yc_{o}$ is
\begin{equation}\label{lowbnd}
\Lwt_{J_1}(\Yc_{o})={\mathrm S_1} := 2 (z_1,\dots,z_{Q'}) \A (z_1,\dots,z_{Q'})^{\T} \geq 2 z^2 O^* - 2 \lambda_1 Q' Z^{-2/Q'}.
\end{equation}
Now suppose that we have a phase assignment $\Yc$ for $J_1$ where the phases of $u$ and $v$
are different. Let $\hat{\u}$ be the vector with $\hat{u}_i$ counting the number of vertices with
phase $i$ in $K_u$ and define similarly $\hat{\v}$.

Note that $\|\hat{\u}-\hat{\v}\|_2^2 =2$ (since $\hat{\u}$ and $\hat{\v}$
differ in two coordinates---the phases of $u$ and $v$ in the assignment). By triangle inequality we have
$\|\hat{\u}/z-\x^*\|_2\geq 1/(z\sqrt{2})$ or $\|\hat{\v}/z-\x^*\|_2\geq 1/(z\sqrt{2})$
(otherwise we would have $\|\hat{\u}/z-\hat{\v}/z\|_2<\sqrt{2}/z$). W.l.o.g. assume that $\hat{\u}/z$ has  the greater distance
from $\x^{*}$. We have
\begin{equation}\label{inbnd}
\Lwt_{J_1}(\Yc_{o})={\mathrm S_2} := \hat{\u}^{\T} \A \hat{\u} + \hat{\v}^{\T} \A \hat{\v} \leq z^2 ( 2 O^* - \lambda_2 / (2z^2) ) = 2 z^2 O^* - \lambda_2 /2.
\end{equation}
By our choice of $Z$ we have $S_1>S_2$ and hence in an optimal phase assignment for $J_1$ we have that $u$
and $v$ get the same phase. Note that we did not show which phase assignment is optimal; we only found a phase
assignment in which $u, v$ have the same phase that is better than any assignment in which $u,v$ have
different phases.
\end{proof}

\subsubsection{The reduction}

In Section~\ref{sec:antiferromagnetic}, Lemma~\ref{lem:antiferromagnetic} we proved that for a parallel edge and any phase $\p$ we have
$w( \p^{+} , \p^{+} )=w_p(\p^{-},\p^{-}) < w_p(\p^{+}, \p^{-})$ and hence there exists a constant $\eps_2>0$ depending only on the model and $\Delta$
such that for every phase $\p\in \Qc$ we have
\begin{equation}\label{crk}
w_p(\p^{+},\p^{+})=w_p(\p^{-},\p^{-}) < w_p(\p^+,\p^-) - \eps_2.
\end{equation}
Combining Lemma~\ref{hrk2} with equation~\eqref{crk} we can construct a
gadget that ``prefers'' the unordered phase of two vertices to agree and
also ``prefers'' the spin assignment to disagree.

\begin{lemma}\label{hrk3}
A constant sized gadget $J_2$ can be constructed with two distinguished vertices $u,v$
and the following property: there exists a phase $\p\in {\cal Q}'$ satisfying simultaneously all of the following:
\begin{enumerate}
\item $A_1(\p)=\MLwt(J_2)$, where
\begin{equation}\label{e1wwq}
A_1(\p) := \max_{\Yc;\, \Yc(u)=\p^+, \Yc(v)=\p^-} \Lwt_{J_2}(\Yc) = \max_{\Yc;\, \Yc(u)=\p^-, \Yc(v)=\p^+} \Lwt_{J_2}(\Yc).
\end{equation}
\label{it:maximappp}

\item Among $\p$ that satisfy Item~\ref{it:maximappp}, $\p$ maximizes
\begin{equation}\label{e2wwq}
A_2(\p) := \max_{\Yc;\, \Yc(u)=\p^+, \Yc(v)=\p^+} \Lwt_{J_2}(\Yc) = \max_{\Yc; \Yc(u)=\p^-, \Yc(v)=\p^-} \Lwt_{J_2}(\Yc).
\end{equation}
\item The following inequalities hold
\begin{equation}\label{hhu4}
A_1(\p) > A_2(\p) + \eps_3\quad\mbox{and}\quad A_2(\p) > \eps_3 + \max_{\Yc; \Yc'(u) \neq \Yc'(v)}\Lwt_{J_2}(\Yc),
\end{equation}
where $\eps_3>0$ is a constant (depending only on the model and $\Delta$).
\end{enumerate}
\end{lemma}

\begin{proof}
To construct $J_2$ we take $t := 3 \lceil (\max_{\p_1,\p_2} w_p(\p_1,\p_2) - \min_{\p_1,\p_2} w_p(\p_1,\p_2) )/\eps_1\rceil$ copies of gadget $J_1$ from Lemma~\ref{hrk2},
identify (merge) their $u$ vertices, and identify (merge) their $v$ vertices. Finally
we add a parallel edge between $u$ and $v$.

Let $\p$ be the unordered phase that is the common value of $\Yc'(u)$ and $\Yc'(v)$ for which the maximum on the left-hand
side of~\eqref{ooow1} is achieved (note that $\p$ is not unique; we just take one such $\p$). Let
$$
A_4 := \max_{\Yc;\, \Yc'(u)=\p, \Yc'(v)=\p} \Lwt_{J_2}(\Yc)\quad\mbox{and}\quad
A_5 := \max_{\Yc;\, \Yc'(u)\neq \Yc'(v)} \Lwt_{J_2}(\Yc).
$$
Then applying~\eqref{ooow1} on each copy of $J_1$ in $J_2$ we obtain
\begin{equation}\label{sakpp}
A_4 > A_5 + 2(\max_{\p_1,\p_2} w_p(\p_1,\p_2) - \min_{\p_1,\p_2} w_p(\p_1,\p_2) ).
\end{equation}
Thus the maximizer of $\max_{\Yc} \Lwt_{J_2}(\Yc)$ happens for $\Yc$ with $\Yc'(u)=\Yc'(v)$. Only the
parallel edge is influenced by the spin and hence, by~\eqref{crk}, we have
\begin{equation}\label{sako112}
\max_{\Yc} \Lwt_{J_2}(\Yc) = \max_{\p} \max_{\Yc:\, \Yc(u)=\p^+, \Yc(v)=\p^-} \Lwt_{J_2}(\Yc).
\end{equation}
Let $\p$ be the maximizer on the right-hand side of~\eqref{sako112} that (secondarily)
maximizes the second expression in~\eqref{e2wwq}. Note that $\p$ satisfies the
first and second condition of the lemma. The first part of the third condition
is satisfied for any $\eps_3\leq \eps_2$ (using~\eqref{crk}). Recall that $\eps_2>0$. The second part of the
third condition is satisfied for $\eps_3\leq \max_{\p_1,\p_2} w_p(\p_1,\p_2) - \min_{\p_1,\p_2} w_p(\p_1,\p_2)$.
Recall that $\max_{\p_1,\p_2} w_p(\p_1,\p_2) - \min_{\p_1,\p_2} w_p(\p_1,\p_2)>0$.
Thus we can take $\eps_3>0$ to be the smaller of the two upper bounds (each of which is
a constant depending on the model and $\Delta$ only).
\end{proof}

\begin{lemma}\label{lem:finalpiece}
Let $\B$ be the interaction matrix of an antiferromagnetic spin model. Let $A_1,A_2$ be the constants defined in Lemma~\ref{hrk3}. There  exists constants $D_1,D_2,D_3$ depending only on the model and $\Delta$
such that the following is true. Given a cubic graph $H$ we can, in polynomial-time, construct a max-degree-$D_1$ graph $G$ with $|V(G)| \leq D_2 |V(H)|$ such that
$$\MLwt(G) = (A_1-A_2) {\mbox{\textsc{MaxCut}}}(H) + A_2 |E(H)|+A_1 D_3 |V(H)|.$$
\end{lemma}

We can now go back and prove the inapproximability result for the phase labeling problem.

\begin{proof}[Proof of Lemma~\ref{lem:maxcuttophase}]
Since $A_1,A_2,D_3$ are constants depending only on the model and $\Delta$, the trivial algorithm gives the bound $\textsc{MaxCut}(H)\geq 1/2|E(H)|=3/4|V(H)|$. Together with Lemma~\ref{lem:finalpiece} we obtain the result.
\end{proof}

\begin{proof}[Proof of Lemma~\ref{lem:finalpiece}]
Replace each edge of $H$ by gadget $J_2$ and for each vertex $w\in V(H)$ add $D_3$ new vertices $w_1,\dots,w_{D_3}$ and add a gadgets $J_2$ between $w$ and $w_i$
(for $i\in [D_3]$),  where $D_3$ will be determined shortly.

The purpose of the $D_3$ copies of $J_2$ is to force phase $\p$ (from Lemma~\ref{hrk3}) to be used on the distinguished vertices in a labeling of $G$ with maximum weight. A phase $\r\neq \p$ can have
\begin{equation}\label{ert1}
\ell_1(\r) := \max_{\Yc;\, \Yc(u)=\r^+, \Yc(v)=\r^+} \Lwt_{J_2}(\Yc) - \max_{\Yc;\, \Yc(u)=\p^+, \Yc(v)=\p^+} \Lwt_{J_2}(\Yc) > 0,
\end{equation}
but then by the choice of $\p$
\begin{equation}\label{ert2}
\ell_2(\r):= \max_{\Yc;\, \Yc(u)=\p^+, \Yc(v)=\p^-} \Lwt_{J_2}(\Yc) - \max_{\Yc;\, \Yc(u)=\r^+, \Yc(v)=\r^-} \Lwt_{J_2}(\Yc) > 0.
\end{equation}
Let
$$
D_3 = 4 + 3\left\lceil\max_{\r} \frac{\ell_1(\r)}{\ell_2(\r)}\right\rceil,
$$
where the maximum is taken over $\r$ such that~\eqref{ert1} is satisfied (if no such $\r$ exists we can take $D_3=0$). Note that $D_3$
is a constant depending on the model and $\Delta$ only.

Now we want to find the maximum weight labeling of $G$. We are only going to focus on labeling of the distinguished vertices ($u$'s and $v$'s in the $J_2$ gadgets),
since once those are fixed one just finds the optimal labeling in each gadget (conditioned on the labels of distinguished vertices). Let $W$ be
a labeling of the distinguished vertices that leads to the maximum weight labeling of $G$. Let $\hat{W}$ be the labeling obtained from $W$ by changing
the phase  of each distinguished vertex to $\p$ while (1) keeping the original spin on the vertices of $H$, and (2) making the spin of $w_1,\dots,w_{D_3}$
the opposite of the spin of $w$ (for each $w\in V(H)$). Now we compare $W$ and $\hat{W}$ for each $J_2$ gadget corresponding to edge of $H$:
\begin{itemize}
\item if in $W$ the phase of $u$ and $v$ were different then $\hat{W}$ has higher weight than $W$ on the gadget, using~\eqref{hhu4};
\item if in $W$ the phase of $u$ and $v$ is the same but the spin is different then $\hat{W}$ has greater or equal weight than $W$ on the gadget, using~\eqref{e1wwq};
\item if in $W$ the phase of $u$ and $v$ is the same and the spin is the same that the loss of $\hat{W}$ on the gadget (compared to $W$) is $\ell_1(\r)$
(where $\r$ is the phase of $u$, $v$ in $W$).
\end{itemize}
For the $J_2$ gadgets connecting $w$ to $w_1,\dots,w_{D_3}$ we have
\begin{itemize}
\item if the phase of $w$ in $W$ was $\r$ such that $\ell_1(\r)>0$ then the gain of $\hat{W}$ on each gadget (compared to $W$) is at least $\ell_2(\r)$;
\item otherwise, by~\eqref{e1wwq} then $\hat{W}$ has greater or equal weight than $W$ on the gadget.
\end{itemize}
For each vertex whose phase in $W$ was $\r$ such that $\ell_1(\r)>0$ there are $3$ edges where $\hat{W}$ can lose $\ell_1(\r)$ (compared to $W$)
but there are $D_3$ edges where $\hat{W}$ gains $\ell_2(\r)$ (compared to $W$). Since $D_3 \ell_2(\r) > 3 \ell_1(\r)$ we have that $\hat{W}$ has
at least as large weight as $W$ (and hence is also optimal).

Now we just argue how the spins should be assigned. The largest number of $J_2$ gadgets with opposite spins on the distinguished vertices arises
when we take the max-cut of $H$ and assign the spin according to the cut.\end{proof}

\subsection{Connection between approximating the partition function and the phase labeling problem}\label{sec:slysunstuff}

In this section we prove Lemma~\ref{lem:phaselabelinglemma}. 

Let $H=(V,E)$ be an instance of the phase labeling problem, where $\{E_p,E_s\}$ is a partition of the edges of $H$.  Let $|V|=m$. The degree of a vertex $v\in V$ will be defined as $2d_s+d_p+4l_s+2l_p$, where $d_s,d_p$ are the numbers of symmetric and parallel edges joining $v$ to a distinct vertex $u$ and $l_s,l_p$ are the numbers of symmetric and parallel loops from $v$ to itself.  The bounded  degree assumption means  there is an absolute constant $D$ (not depending on $m$) which bounds the degree of any $v\in V$.

To approximate the phase labeling problem on $H$ with parameters $\B,\Qc$, we will replace each vertex in the graph $H$  by a suitable graph in a family of gadgets $\Fc$. The construction has a parameter $k$ which roughly controls the accuracy of the approximation we want to achieve.  The family $\Fc$ will be of the form  $\{G^d\}_{d\in [D]}$ and the gadget for a vertex $v$ will be $G^{d}$ where $d$ is the degree of $v$. Note that the cardinality of $\Fc$ is bounded by the absolute constant $D$. The gadgets $G^{d}$ are selected from  a graph distribution $\Gc^{kd}_{n}$ for some appropriate $n$ to be specified later. For integer $r,n$ satisfying $n>r\geq 0$,  we next describe the graph distribution $\Gc^r_{n}:=\Gc^r_{n}(\Delta)$.
\begin{enumerate}
\item $\Gc^r_{n}$ is supported on bipartite graphs. The two parts of the bipartite graph are labeled by $+,-$ and each is partitioned as $U^{s}\cup W^{s}$ where $|U^{s}|=n$, $|W^{s}|=r$ for $s=\{+,-\}$. $U$ denotes the set $U^{+}\cup U^{-}$ and similarly $W$ denotes the set $W^{+}\cup W^{-}$.
\item To sample $G\sim\Gc^r_{n}$,  sample uniformly and independently $\Delta$ matchings: (i) $(\Delta-1)$ perfect matchings between $U^+\cup W^+$ and $U^-\cup W^-$,  (ii) a $n$-matching between $U^+$ and $U^-$. The edge set of $G$ is the union of the $\Delta$ matchings. Thus, vertices in $U$ have degree $\Delta$, while vertices in $W$ have degree $\Delta-1$.
\end{enumerate}
Note that in the special case $r=0$, the distribution $\Gc^r_{n}$ is identical to the graph distribution $\Gc_n$ defined in Section~\ref{sec:derivations}.

Before further specifying the family $\mathcal{F}$, we first describe the properties that a gadget in $\mathcal{F}$ should have. We assume throughout that $r$ is an arbitrarily large constant (independent of $n$).  Let $G\sim \Gc^r_n$ and denote by  $\mu_G$ the Gibbs distribution on  $G$ with interaction matrix $\B$. Note that $G$ is a random graph on $2(n+r)$ vertices. 

For $\sigma:\, U\cup W\rightarrow [q]$, the footprint of $\sigma$ is a pair of $q$-dimensional vectors  $(\alphab_{\sigma},\betab_\sigma)$. The $i$-th entry of $\alphab_{\sigma}$ (resp. $\betab_\sigma$) is equal to $|\sigma^{-1}(i)\cap U^+|/n$ (resp. $|\sigma^{-1}(i)\cap U^-|/n$). Let $\p\in \Qc$ and recall that $\p$ corresponds to a dominant phase $(\alphab,\betab)$ of $\Psi_1$. The phase of  a configuration $\sigma:\, U\cup W\rightarrow [q]$ will be denoted by $Y(\sigma)$ and equals $\p$ if the closest\footnote{See  Appendix~\ref{sec:momentasymptotics}, equation \eqref{eq:phaseofconfiguration} for the precise definition.} dominant phase to the footprint $(\alphab_\sigma,\betab_\sigma)$ of $\sigma$ is  $(\alphab,\betab)$. Note that the phase of $\sigma$ depends only on the spins of vertices in $U$.  

We shall display shortly that, conditioned on  $Y(\sigma)=\p$, the marginal distribution of $\mu_G$ on the vertices in $W$ can be well approximated by an appropriate product measure $\nu^{\otimes}_\p(\cdot)$. To do this, recall that every phase $\p\in \Qc$ corresponds to a fixpoint of the tree recursions \eqref{kkrtko}. Let $(\hat{R}_1,\hdots,\hat{R}_q)$ be a scaled version of $(R_1,\hdots,R_q)$ so that $\sum_{i}\hat{R}_i=1$ (and define similarly $\hat{C}_1,\hdots,\hat{C}_q$). We now define a product measure $\nu^{\otimes}_\p(\cdot)$ on the space of spin assignments to vertices in $W$. For $\eta:\, W\rightarrow [q]$ and $\p\in \Qc$, let
\begin{equation}\label{eq:nudefinition}
\nu^{\otimes}_\p(\eta)=\prod_{i\in[q]}(\hat{R}_i)^{|\eta^{-1}(i)\cap W^{+}|}\prod_{j\in[q]}(\hat{C}_j)^{|\eta^{-1}(j)\cap W^{-}|}.
\end{equation}

For $\sigma:\, U\cup W\rightarrow [q]$, denote by $\sigma_W$ the restriction of $\sigma$ to vertices in $W$.
\begin{lemma}\label{lem:gadgetcrucial}
Let $r$ be an arbitrarily large constant. In the setting  of Theorem~\ref{thm:general-inapprox}, for every $\epsilon>0$, there exists $N(\epsilon)$ such that for $n\geq N$, a random graph $G\sim \Gc^r_n$ satisfies with positive probability the following:
\begin{enumerate}
\item The graph $G$ is simple.\label{it:simplicity}
\item For each $\p\in \Qc$, $(1-\epsilon)/|\Qc|\leq\mu_G(Y(\sigma)=\p)\leq (1+\epsilon)/|\Qc|$. That is, the phases in $\Qc$ appear with roughly equal probability.\label{it:independence}
\item Let $\sigma\sim \mu_G$. Then, $\mu_G(\sigma_W=\eta\, |\, Y(\sigma)=\p)/\nu^{\otimes}_\p(\eta)\in[1-\epsilon,1+\epsilon]$ for all $\eta:\, W\rightarrow[q]$. That is, conditioned on the phase $\p$ of the configuration, the spins of the vertices in $W$ are roughly independent and the marginal distribution on them can be approximated by the distribution $\nu^{\otimes}_\p(\cdot)$.\label{it:prodapprox}
\item There is no edge between $W^+$ and $W^-$. Moreover, there is no vertex in $G$ which has two neighbors in $W^+\cup W^-$.\label{it:trianglefree}
\end{enumerate}
\end{lemma}

Lemma \ref{lem:gadgetcrucial} is proved in Section~\ref{sec:proof-small-graph}.
An immediate consequence of Lemma~\ref{lem:gadgetcrucial} is the following.
\begin{corollary}\label{cor:family}
Let $k$ be an arbitrarily large constant.  For $d\in [D]$, let $G^d\sim \Gc^{kd}_n$ and set $\mathcal{F}=\{G^d\}_{d\in [D]}$. Then, for all sufficiently large $n$, $G^d$ satisfies Items~\ref{it:simplicity},~\ref{it:independence},~\ref{it:prodapprox} and \ref{it:trianglefree} of Lemma~\ref{lem:gadgetcrucial} with positive probability for every $d\in[D]$.
\end{corollary}
Corollary~\ref{cor:family} also yields a trivial randomized algorithm  to construct the family $\Fc$ for an arbitrary constant $k$. In fact, since all the parameters  are constants, one can construct the family $\Fc$ by brute force search. With the family $\Fc$ in our hands, we can now give the details of the construction.

The first step consists of replacing each vertex $v\in H$ with degree $d$ with a distinct copy of the gadget $G^{d}\in \mathcal{F}$. We will denote the gadget corresponding to vertex $v$ by $G_v$ and the images of the sets $W,W^{\pm},U^{\pm}$ in $G_v$ by $W_v,W^{\pm}_v,U^{\pm}_v$. Further, denote by $\widehat{H}$ the graph obtained by the disconnected copies of the gadgets.

The second step consists of encoding the edges of $H$ in $\widehat{H}$, that is, making connections between the gadgets. The final graph will be denoted by $H_{\Fc}$. The edges we are going to place will form a perfect matching on $\cup_{v\in H} W_v$ and as a result $H_{\Fc}$ will be $\Delta$-regular. Every parallel edge of $H$ corresponds to $2k$ edges in $H_{\Fc}$, while every symmetric to $4k$. Roughly, parallel and symmetric indicate which parts of two gadgets get connected (recall that the gadgets are bipartite). Loops are treated as if they were connecting distinct vertices.

In detail, let $(u,v)$ be an edge $e$ of $H$.  Suppose first that $u\neq v$. If  $e$ is parallel, place $k$ edges between $W^s_u$ and $W^s_v$ for $s\in\{+,-\}$.  If $e$ is symmetric, place $k$ edges between $W^s_u$ and $W^s_v$ and $k$ edges between $W^s_u$ and $W^{-s}_v$ for $s\in \{+,-\}$. Suppose now that $u=v$. If  $e$ is parallel, place $k$ edges between distinct vertices in $W^+_v$ and $k$ edges between distinct vertices in $W^-_v$. If $e$ is symmetric, place $2k$ edges between $W^+_v$ and $W^-_v$, $k$ edges between distinct vertices in $W^+_v$ and $k$ edges between distinct vertices in $W^-_v$.

The first step of the construction guarantees that the second step can be done in a (deterministic) way so that $H_{\Fc}$ is $\Delta$-regular.  Moreover, by Corollary~\ref{cor:family} and item \ref{it:trianglefree} of Lemma~\ref{lem:gadgetcrucial}, $H_{\Fc}$ is a simple, triangle-free graph.
\begin{proof}[Proof of Lemma~\ref{lem:phaselabelinglemma}]
Using Corollary~\ref{cor:family} and specifically items~\ref{it:independence} and~\ref{it:prodapprox} of Lemma~\ref{lem:gadgetcrucial}, the argument in \cite[Lemma 4.3]{SS} almost verbatim gives
\[ \frac{(1-\epsilon)^{2m}}{|\Qc|^{m}}\leq \frac{Z_{H_{\Fc}}/Z_{\widehat{H}}}{\exp(k\cdot\MLwt(H))}\leq (1+\epsilon)^m.\]
This can be rearranged into
\begin{equation}\label{eq:qualityapprox}
\frac{1}{k}\log\Big(\frac{Z_{H_{\Fc}}}{Z_{\widehat{H}}}\Big)-\frac{m}{k}\log(1+\epsilon)\leq\MLwt(H)\leq \frac{1}{k}\log\Big(\frac{Z_{H_{\Fc}}}{Z_{\widehat{H}}}\Big)-\frac{m}{k}\big[2\log(1-\epsilon)-\log |\Qc|\big].
\end{equation}

The argument in \cite[Proof of Theorems 1 and 2]{SS} gives the desired result. We give the short details. The graph $\widehat{H}$ consists of $m$ disconnected subgraphs, each of constant size. Hence, we can compute $Z_{\widehat{H}}$ exactly in polynomial time.  Assume now that $Z_{H_{\Fc}}$ can be approximated within a factor of $\exp\big(c|\widehat{H}|\big)$ in polynomial time for any $c>0$.  Since $\log\big(Z_{H_{\Fc}}\big)$ is bounded above by $O(|\widehat{H}|)$, the ratio $\log\big(Z_{H_{\Fc}}/Z_{\widehat{H}}\big)$ can be approximated within an additive $O(c|\widehat{H}|)=O[c m(n+kD)]=O(cnm)$ since $n>kD$. Thus, by \eqref{eq:qualityapprox}, we obtain upper and lower bounds for $\MLwt(H)$ which differ by $O[(cn+1)m/k]$. A random phase labeling yields the lower bound $\MLwt(H)\geq \Omega(m)$. Thus, the final approximation is within a multiplicative factor $O[(cn+1)/k]$ of $\MLwt(H)$. To make the multiplicative factor arbitrarily small, we need to take $k$ large. This might increase $n$, but we can compensate by taking $c$ small. This concludes the proof.
\end{proof}

\subsubsection{Proof of Lemma~\ref{lem:gadgetcrucial}}\label{sec:proof-small-graph}
Let $G\sim\Gc^r_n$. To get a handle on Items~\ref{it:independence} and~\ref{it:prodapprox} of Lemma~\ref{lem:gadgetcrucial}, we first define the partition functions conditioned on a phase $\p\in \Qc$. Similar definitions appear in~\cite{Sly}. Let $\Omega^\p$ be the configurations $\sigma\in \Omega$ whose phase $Y(\sigma)$ equals $\p$, i.e., 
\begin{equation}\label{eq:omegap}
\Omega^\p=\{\sigma\in \Omega \mid Y(\sigma)=\p\}.
\end{equation}

Similarly, for a configuration $\eta:W\rightarrow [q]$, let 
\begin{equation}\label{eq:omegapeta}
\Omega^\p(\eta)=\{\sigma\in \Omega \mid Y(\sigma)=\p, \sigma_W=\eta\}.
\end{equation}
Note that $\Omega^\p=\cup_\eta\, \Omega^\p(\eta)$ and $\Omega=\cup_{\p\in \Qc}\,\Omega^{\p}$. The conditioned partition functions $Z^{\p}_G$ and $Z^{\p}_G(\eta)$ are defined as
\begin{equation}
Z^{\p}_{G}(\eta):=\sum_{\sigma\in \Omega^{\p}(\eta)}w_G(\sigma),\quad Z^{\p}_{G}:=\sum_{\sigma\in \Omega^{\p}}w_G(\sigma)=\sum_{\eta: W\rightarrow [q]}Z^{\p}_G(\eta).\label{eq:defcondpart}
\end{equation}
The following equalities display the relevance of these quantities  to Lemma~\ref{lem:gadgetcrucial}.
\begin{equation}
\mu_{G}(Y(\sigma)=\p)=\frac{Z^{\p}_{G}}{\sum_{\p \in \Qc} Z^{\p}_{G}},\quad \mu_{G}(\sigma_W=\eta\, | Y(\sigma)=\p)=\frac{Z^{\p}_{G}(\eta)}{Z^{\p}_{G}}.\label{eq:reformp}
\end{equation}

Note that the definition of $Z^{\p}_G$ also makes sense in the case $r=0$. Note that for $r=0$ there are no vertices of degree $\Delta-1$ (and hence no set $W$), so the graph distribution $\Gc^0_n$ is identical to the graph distribution $\Gc_n$ defined in Section~\ref{sec:derivations}.

To start, we are going to show  that Items~\ref{it:independence} and~\ref{it:prodapprox} of Lemma~\ref{lem:gadgetcrucial} hold in expectation. This is the scope of the following lemma which expresses $\E_{\Gc^r_n}\big[Z^{\p}_G\big],\E_{\Gc^r_n}\big[Z^{\p}_G(\eta)\big]$ in terms of $\E_{\Gc_n}\big[Z^{\p}_G\big]$. Note that $o(1)$ refers to quantities that tend to 0 as $n\rightarrow \infty$.
\begin{lemma}\label{lem:itemexpect}
Let $r$ be a fixed constant and let $\p$ be a phase, i.e., $\p\in \Qc$. There exists a constant $C(\p)$ such that for every $\eta:W\rightarrow [q]$, it holds that
\begin{equation}\label{eq:itemexpect2}
\E_{\Gc^r_n}\big[Z^{\p}_G(\eta)\big]=\big(1+o(1)\big)C^r\nu^{\otimes}_{\p}(\eta)\E_{\Gc_n}\big[Z^{\p}_G\big], \mbox{\  and thus \ } \max_{\eta:W\rightarrow[q]}\Big|\frac{\E_{\Gc^r_n}\big[Z^{\p}_G(\eta)\big]}{\E_{\Gc^r_n}\big[Z^{\p}_G\big]}-\nu^{\otimes}_{\p}(\eta)\Big|=o(1).
\end{equation}
Moreover, when the phases $\Qc$ are permutation symmetric,  $\E_{\Gc_n}\big[Z^{\p}_G\big]=\big(1+o(1)\big)\E_{\Gc_n}\big[Z^{\p'}_G\big]$ for any two phases $\p,\p'\in \Qc$ and the constant $C$ in \eqref{eq:itemexpect2} does not depend on the particular phase $\p$. Consequently, for $\p,\p'\in \Qc$
\begin{equation}\label{eq:itemexpect1}
\E_{\G^r_n}\big[Z^{\p}_{G}\big]=\big(1+o(1)\big)\E_{\G^r_n}\big[Z^{\p'}_{G}\big], \mbox{\  and thus \ } \frac{\E_{\G^r_n}\big[Z^{\p}_{G}\big]}{\sum_{\p \in \Qc}\E_{\G^r_n}\big[Z^{\p}_G\big]}=\big(1+o(1)\big)\frac{1}{|\Qc|}.
\end{equation}
\end{lemma}
\begin{proof}
The second equalities in each of \eqref{eq:itemexpect2} and \eqref{eq:itemexpect1} follow immediately from the first. The latter may be proved by explicit calculations following the same arguments as in~\cite[Lemma 3.3]{Sly}. It is worthy to note that \eqref{eq:itemexpect2} holds even if the phases are not permutation symmetric, which is not in general true for \eqref{eq:itemexpect1}.
\end{proof}

In light of Equations~\eqref{eq:reformp},~\eqref{eq:itemexpect2} and \eqref{eq:itemexpect1}, the path to obtain Items~\ref{it:independence} and~\ref{it:prodapprox} of Lemma~\ref{lem:gadgetcrucial} is now paved: it suffices to show that the conditioned partition functions $Z^\p_G(\eta)$ are (with positive probability) arbitrarily close to their expectations for large $n$. Note that we want this to be simultaneously true for all $\p$ and $\eta$, that is, for the same graph $G$. This in turn requires using in full strength a theorem by \cite{Janson}, which is an extension of the small subgraph conditioning method introduced by  \cite{Worm2}.

We do an exposition of these theorems and their application in Appendix~\ref{sec:small-graph}. For satisfying the reader who is more interested in the proof of Lemma~\ref{lem:gadgetcrucial}, the following lemma is a distilled version of the results in Appendix~\ref{sec:small-graph}, yet at the same point containing some important bits which will allow us to motivate it.

\begin{lemma}\label{lem:smallgraphstrip}
Let $G\sim \Gc^{r}_n$ and denote by $X_{in}$, $i=1,2,\hdots,$  the number of cycles of length $2i$ in $G$. There exist random variables $W^{\p}_{mn}$, a deterministic function of $X_{1n},X_{2n},\hdots,X_{mn}$, such that for every $\epsilon>0$
\begin{equation}\label{eq:concstronger}
\lim_{m\rightarrow\infty}\limsup_{n\rightarrow\infty}\mathrm{Pr}_{\Gc^r_n}\Big(\bigcup_{\p}\bigcup_{\eta}\Big[\big|\frac{Z_G^{\p}(\eta)}{\E_{\Gc^r_n}[Z_G^{\p}(\eta)]}-W^{\p}_{mn}\big|>\epsilon\Big]\Big)=0.
\end{equation}
There also exists a positive constant $c>0$ such that $W^{\p}_{mn}>c$ uniformly in $m,n$. Moreover, when the phases $\Qc$ are permutation symmetric, the random variables $W^{\p}_{mn}$
do not depend on the phase $\p$.
\end{lemma}

Lemma~\ref{lem:smallgraphstrip} provides a straightforward proof of Lemma~\ref{lem:gadgetcrucial}, so we shall elucidate its most important aspects in an attempt to demystify its rather unintuitive statement.  
Equation \eqref{eq:concstronger} says that for all sufficiently large $m,n$ the random variables $Z_G^{\p}(\eta)/\E_{\Gc^r_n}[Z_G^{\p}(\eta)]$ are well-approximated by the variables $W^{\p}_{mn}$, with large probability. To get a feeling about this statement, it is well known fact that a random $\Delta$-regular graph is locally tree-like and its girth diverges as $n\rightarrow\infty$. That is, as $n$ grows large, for any positive integer $t$, for all but $o(n)$ vertices, the $t$-depth neighborhood of a vertex is isomorphic to the first $t$ levels of the infinite $\Delta$-regular tree. This is in alignment with the fact that $\E_{\Gc^r_n}[Z_G^{\p}(\eta)]$ is determined by the Gibbs measure on the infinite $\Delta$-regular tree associated to the phase $\p$. On the other hand, a graph $G\sim \Gc^r_n$ does have $o(n)$ vertices which are contained in constant sized cycles. Thus, it is reasonable to expect that $Z_G^{\p}(\eta)$  fluctuates from its expectation. It is equally reasonable to expect the fluctuations to depend on the presence of small cycles which occur with small but non-zero probability. Equation \eqref{eq:concstronger} thus provides an explicit handle on these fluctuations, given by the variables $W^{\p}_{mn}$, which are a deterministic function of the small cycle counts in $G$. Crucially for our proof of Lemma~\ref{lem:gadgetcrucial}, when the phases are permutation symmetric, the fluctuations from the expectation are captured by a single random variable, which allows us to control them uniformly over all the phases $\p$ and configurations $\eta$.
 
We should point out that the notation $W^{\p}_{mn}$ should not be confused by any means to the labeling of the degree $\Delta-1$ vertices in $G$, i.e., the set of vertices $W$.

\begin{proof}[Proof of Lemma~\ref{lem:gadgetcrucial}]
We assume that the $\epsilon$ in the statement of the lemma is fixed. Let $\epsilon'>0$ be sufficiently small, to be picked later.

By  Lemma~\ref{lem:smallgraphstrip}, for all $m,n$ sufficiently large the random variables $Z^\p_G(\eta)/\E_{\Gc^r_n}\big[Z^\p_G(\eta)\big]$ are well approximated by $W^\p_{mn}$ with large probability. That is, there exist $M(\epsilon')$, $N(\epsilon')$ such that for $m\geq M$ and $n\geq N$, it holds with probability $1-\epsilon'$ over the choice of the graph $G$ that, for every phase $\p$ and every configuration $\eta:W\rightarrow [q]$,
\begin{equation}\label{eq:firstcondp}
Z^\p_G(\eta)=(W^{\p}_{mn}\pm \epsilon')\E_{\Gc^r_n}\big[Z^\p_G(\eta)\big].
\end{equation}
We will show that whenever this is the case, Items~\ref{it:independence} and~\ref{it:prodapprox} hold. To do this, sum \eqref{eq:firstcondp} over $\eta$ to obtain   that for each phase $\p$, it holds
\begin{equation}\label{eq:secondcondp}
Z^{\p}_G=(W^{\p}_{mn}\pm \epsilon')\E_{\Gc^r_n}\big[Z^{\p}_G\big],
\end{equation}
Using the positive constant $c$ in Lemma~\ref{lem:smallgraphstrip},
we obtain that for  $\epsilon'$ sufficiently smaller than $c$, the ratio $Z^{\p}_G(\eta)/Z^{\p}_G$ is   within a multiplicative $(1\pm\epsilon)$  from $\E_{\Gc^r_n}\big[Z^{\p}_G(\eta)\big]/\E_{\Gc^r_n}\big[Z^{\p}_G\big]$.  This gives Item~\ref{it:prodapprox} of the lemma, when used in conjuction with \eqref{eq:reformp} and \eqref{eq:itemexpect1}. Note that this part of the argument did not use that the phases $\p$ are permutation symmetric.

To obtain Item~\ref{it:independence}, we have to use that the phases $\p$ are permutation symmetric. Then $W^{\p}_{mn}=: W_{mn}$ by the last assertion in Lemma~\ref{lem:smallgraphstrip}. Thus, a summation of \eqref{eq:secondcondp} over $\p\in \Qc$ gives $Z^{\p}_G=(W_{mn}\pm \epsilon')\E_{\Gc^r_n}\big[Z^{\p}_G\big]$. Exactly the same reasoning as before yields the thesis.

It is a standard  union bound to show that Item~\ref{it:trianglefree} holds with probability $1-O(1/n)$ over the choice of the graph $G$, essentially because $G$ is an expander. Perhaps the second assertion there requires a brief proof sketch. Let $v\in V_1$, $w_1,w_2\in W^{-}$ and let $E_i$ be the event that $(v,w_i)$ is an edge of $G$. The events $E_1,E_2$ are negatively correlated since $v$ has a fixed number of edges incident to it, either $\Delta$ or $\Delta-1$. It is also easy to see that $\Pr_{\Gc^{r}_n}(E_i)\leq 1-(1-1/n)^{\Delta}=O(1/n)$, so that $\Pr_{\Gc^{r}_n}(E_1\cap E_2)=O(1/n^2)$. A union bound over the roughly $nr^2=O(n)$ possibilities of the vertices $v,w_1,w_2$ gives the desired bound.

Thus, a graph $G~\sim \Gc^{r}_n$ satisfies Items~\ref{it:independence},~\ref{it:prodapprox} and~\ref{it:trianglefree} with large probability for all sufficiently large $n$. The first assertion in Item~\ref{it:simplicity} of Lemma~\ref{lem:smallgraphstrip} can hence be guaranteed by contiguity, see \cite[Section 2]{Janson}.
\end{proof}

\section{Dominant phases for Potts Model and Colorings}
\label{sec:phase-diagram}

\subsection{Proof outline}
In this section we prove Theorem \ref{thm:fase} which establishes the hypotheses of Theorem \ref{thm:general-inapprox} for the dominant phases of the antiferromagnetic Potts and colorings models on random $\Delta$-regular bipartite graphs (and, as we showed in Section~\ref{sec:generalresults}, Theorems \ref{thm:colorings} and \ref{thm:Potts} follow as corollaries).

 Recall, the interaction matrix $\B$ for the Potts model is completely determined by a parameter $B$, which is equal to $\exp(-\beta)$ where $\beta$ is the inverse temperature in the standard notation for the Potts model. The antiferromagnetic regime corresponds to $0< B<1$. The coloring model is the zero temperature limit of the Potts model and corresponds to the particular case $B=0$ in what follows. We should note that in Statistical Physics terms, the arguments of this section are closely related to the phase diagrams of the models.


By Theorem~\ref{new:zako1} specified to the antiferromagnetic Potts and colorings models, studying the global maxima of $\Psi_1$ is equivalent to studying the global maxima of $\Phi$. Moreover, the global maxima of $\Phi$ and $\Psi_1$ occur at their critical points. Since there is a one-to-one correspondence between the critical points of $\Phi$ and the critical points of $\Psi_1$ (given by \eqref{jqwe}), we will freely interchange our focus  between critical points of $\Phi$ and $\Psi_1$. 

The critical points of $\Phi$, by the first part of  Theorem~\ref{new:zako1}, are given by fixpoints of the tree recursions \eqref{kkrtko}, which for the Potts model read as:
\begin{equation}\label{eq:treePotts}
R_i\propto\left(BC_i+\mbox{$\sum_{j\neq i}$ } C_j\right)^{d}, \quad C_j\propto\left(BR_j+\mbox{$\sum_{i\neq j}$ } R_i\right)^{d},
\end{equation}
where $i,j=1,\hdots,q$ and  $d$ is the notational convenient substitution $d:=\Delta-1\geq 2$. Given a fixpoint of the tree recursions \eqref{eq:treePotts}, we will classify whether it is a Hessian local maximum of $\Psi_1$ using Theorem~\ref{thm:connection}. 

Once we find the global maxima of $\Psi_1$, it will be simple to prove that they are Hessian and permutation symmetric. Finding however the global maxima of $\Psi_1$ is going to be more intricate, mainly because the number of local maxima varies according to the value of $B$. We will thus have to compare the values of $\Psi_1$ at the critical points. Rather than doing this directly (which seems as a difficult task), we solve a relaxed optimisation problem, which for $q$ even can be tied to the maximization of $\Psi_1$. We next give the details. 

We begin our considerations by examining when a fixpoint \eqref{eq:treePotts} is \emph{translation invariant}, i.e., satisfies $R_i\propto C_i$ for every $i\in[q]$.
\begin{lemma}\label{lem:translationinvariant}
Let $0\leq B< 1$ and $\Delta\geq 3$. If a solution of \eqref{eq:treePotts} satisfies  $R_i\propto C_i$ for $i\in [q]$, then it holds that $R_1=\hdots=R_q$ and $C_1=\hdots=C_q$.
\end{lemma}
\begin{proof}[Proof of Lemma~\ref{lem:translationinvariant}]
By the symmetries of the model, we may assume an arbitrary ordering of the $R_i$'s. Since $0\leq B<1$, \eqref{eq:treePotts} easily implies the reverse ordering of the $C_i$'s. Thus, $R_i\propto C_i$ for every $i\in [q]$ yields that the ordering must be trivial, i.e, $R_1=\hdots=R_q$ and $C_1=\hdots=C_q$. 
\end{proof}
\begin{corollary}
Translation invariant fixpoints of \eqref{eq:treePotts} always exist and are unique up to scaling.
\end{corollary}
We next explore in which regimes of $B$, the critical points of $\Phi$ consist solely of translation invariant fixpoints. In this regime, we immediately obtain by Theorem~\ref{new:zako1} that the global maximum of $\Psi_1$ (and hence the global maximum of $\Phi$ as well) is achieved at a translation invariant fixpoint.
\begin{lemma}\label{lem:semi-uniqueness}
Let $0\leq B< 1$ and $q,\Delta\geq 3$. When $B\geq \frac{\Delta-q}{\Delta}$, the solution of the system of equations \eqref{eq:treePotts} satisfies $R_1=\hdots=R_q$ and
$C_1=\hdots=C_q$.
\end{lemma}
The proof of Lemma~\ref{lem:semi-uniqueness} is an extension of an argument in \cite{BW} for colorings and is given in Appendix~\ref{sec:semi-uniqueness}. The next lemma states that in the complementary regime of Lemma~\ref{lem:semi-uniqueness}, the translation invariant fixpoint does not correspond to  a local maximum of $\Psi_1$ and hence, by Theorem~\ref{new:zako1}, the global maximum of $\Psi_1$ occurs at a fixpoint of \eqref{eq:treePotts} which is not translation invariant. In particular, in this regime we have semi-translational non-uniqueness.
\begin{lemma}\label{lem:transinva}
For $0\leq B<\frac{\Delta-q}{\Delta}$, the global maximum of $\Psi_1$ is not achieved at the translation invariant fixpoint.
\end{lemma}
\begin{proof}[Proof of Lemma~\ref{lem:transinva}]
We apply Theorem~\ref{thm:connection} by showing that the translation invariant fixpoint is Jacobian unstable and hence not a local maximum of $\Psi_1$. By Lemma~\ref{ztt2}, for a general interaction matrix $\B$, the condition for Jacobian stability of a fixpoint of the tree recursions is related to the spectrum of $\L =\left[\begin{smallmatrix} \zeros & \A\\ \A^{\T} & \zeros\end{smallmatrix}\right]$, where $\A$ is the $q\times q$ matrix whose $ij$-entry is given by $A_{ij} =B_{ij}R_iC_j/\sqrt{\alpha_i\beta_j}$ and $\alpha_i,\beta_j$ are given by \eqref{eex}. Recall that $\pm 1$ are eigenvalues of $\L$ and the condition for Jacobian stability is that all the other eigenvalues have absolute value less than $1/(\Delta-1)$ (see for details the proof of Theorem~\ref{thm:connection} in Section~\ref{sec:attmax}).

In the setting of the lemma, the matrix $\A$ for the translation invariant fixpoint has off-diagonal entries equal to $1/(B+q-1)$ and diagonal entries equal to $B/(B+q-1)$.  It follows that the eigenvalues of $\L$ are $\pm1$ by multiplicity 1 and $\pm(1-B)/(B+q-1)$ by multiplicity $q-1$. The absolute value of the latter is greater than $\frac{1}{\Delta-1}$ for $0\leq B<\frac{\Delta-q}{\Delta}$, as claimed.
\end{proof}
We summarize the above results into the following corollary.
\begin{corollary}\label{cor:semi-uniqueness}
Let $0\leq B< 1$ and $q,\Delta\geq 3$. When $B\geq \frac{\Delta-q}{\Delta}$, $\Psi_1$ has a unique global maximum for $\alpha_1=\hdots=\alpha_q=\beta_1=\hdots=\beta_q=1/q$ or, in other words, the global maximum of $\Psi_1$ is achieved by the fixpoint which corresponds to the (unique) translation invariant Gibbs measure. In the complementary regime $0\leq B<\frac{\Delta-q}{\Delta}$, the maximum of $\Psi_1$ is not achieved at the translation invariant fixpoint, and hence it is achieved at a semi-translation invariant fixpoint which is not translation invariant.
\end{corollary}

Corollary \ref{cor:semi-uniqueness} is not sufficient to obtain Theorems~\ref{thm:colorings} and  \ref{thm:Potts}, since we need to verify that the global maxima of $\Psi_1$ in semi-translational non-uniqueness are Hessian and permutation symmetric. We do this by identifying the critical points which are maxima of $\Psi_1$. 

To state the result, we first need the following structural statement for the solutions of equations~\eqref{eq:treePotts}, namely that solutions of \eqref{eq:treePotts} are supported on at most 3 values for the $R_i$'s and similarly for the $C_i$'s.
\begin{lemma}\label{oooqw}
Let $(R_1,\dots,R_q, C_1,\dots,C_q)$ be a positive solution of the system~\eqref{eq:treePotts}. Let $t_{R}$ be the number of values on which the $R_i$'s are supported and define similarly $t_{C}$. Then $t_R,t_C\leq 3$ and $t_R=t_C=:t$.
\end{lemma}
The proof of Lemma~\ref{oooqw} is given in Section~\ref{sec:remainingproofsmax}. Lemma~\ref{oooqw} motivates the following definition.

\begin{definition}\label{def:typefixpoint}
From Lemma~\ref{oooqw}, the $R_i$'s and $C_j$'s of a fixpoint of \eqref{eq:treePotts}
attain at most $t\leq3$ different values. Let $\tilde{R}_1,\dots,\tilde{R}_t$
and $\tilde{C}_1,\dots,\tilde{C}_t$ be their values and let
$q_1,\dots,q_t\geq 1$ be their multiplicities. When $t=1$, define
$q_2=q_3=0$; when $t=2$, define $q_3=0$; when $q_i=0$, define the values of $\tilde{R}_i,\tilde{C}_i$ to be zero. The corresponding solution of \eqref{eq:treePotts} or equivalently the fixpoint of the tree recursions is then defined to be of type $(q_1,q_2,q_3)$. Note that $q_1+q_2+q_3=q$ and the $q_i$'s are non-negative integers. Call a $(q_1,q_2,q_3)$-type fixpoint to be $t$-supported if  the number of $q_i$'s which are non-zero equals $t$.
\end{definition}

Finding the types of fixpoints  which correspond (via \eqref{jqwe}) to global maxima of $\Psi_1$ is not a trivial task. While 2-supported fixpoints are simple to handle, this is not the case for 3-supported fixpoints. The main lemma we prove is the following, which identifies the type of fixpoints which maximize $\Psi_1$.

\begin{lemma}\label{lem:mainlemmapotts}
For $0\leq B<\frac{\Delta-q}{\Delta}$ and even $q\geq 3$, the maximum of $\Psi_1$ over $(q_1,q_2,q_3)$-type solutions of \eqref{eq:treePotts} is attained at
fixpoints of  type $(q/2,q/2,0)$. 
\end{lemma}

The final piece is to show that fixpoints of  type $(q/2,q/2,0)$ are Hessian maxima of $\Psi_1$ and permutation symmetric. This is the scope of the next lemma, whose proof is given in Section~\ref{sec:remainingproofsmax}.
\begin{lemma}\label{lem:symmetricHessian}
For $0\leq B<\frac{\Delta-q}{\Delta}$ and even $q\geq 3$,  fixpoints of  type $(q/2,q/2,0)$ are Jacobian stable and hence correspond to  Hessian maxima of $\Psi_1$. The values of $R_i$'s and $C_j$'s for fixpoints of  type $(q/2,q/2,0)$ are unique up to scaling and permutations of the colours.
\end{lemma}

We are now ready  to prove Theorem~\ref{thm:fase}.
\begin{proof}[Proof of Theorem~\ref{thm:fase}]
Item~\ref{itt:uniqueness} follows from Corollary~\ref{cor:semi-uniqueness} (see also Lemmas~\ref{lem:semi-uniqueness} and~\ref{lem:transinva}). Item~\ref{itt:dominant} follows from Lemmas~\ref{lem:mainlemmapotts} and~\ref{lem:symmetricHessian}, after using the correspondence between fixpoints of the tree recursions \eqref{eq:treePotts} and dominant phases of Theorem~\ref{new:zako1} (equation \eqref{jqwe}).
\end{proof}

\subsection{Proof of Lemma~\ref{lem:mainlemmapotts}}\label{sec:relaxation}

In this section, we outline the proof of Lemma~\ref{lem:mainlemmapotts}. We need to find the type(s) of the fixpoints which maximize $\Psi_1$. Let $\qb=(q_1,q_2,q_3)$ specify the type of  a fixpoint of \eqref{eq:treePotts} and let $\rb=(R_1,R_2,R_3)$, $\cb=(C_1,C_2,C_3)$ be the respective values of the $R_i$'s and $C_j$'s, see Definition~\ref{def:typefixpoint}. Note that the $q_i$'s are non-negative integers satisfying $q_1+q_2+q_3=q$.

Using Theorem~\ref{new:zako1}, we obtain that the value of $\Psi_1(\alphab,\betab)$ corresponding to this fixpoint  of \eqref{eq:treePotts} is given by the value of the function $\overline{\Phi^{S}}$, where 
\begin{equation}\label{eq:phiover}
\begin{aligned}
\overline{\Phi^{S}}(\qb,\rb,\cb)&:=(d+1)\ln\Big(\mbox{$\sum^3_{i=1}$}\,q_iR_i\,\mbox{$\sum^3_{j=1}$}\,q_jC_j+(B-1)\mbox{$\sum_{i}$}\,q_i R_iC_i\Big)\\
&\qquad \qquad \qquad-d\ln\Big(\mbox{$\sum^3_{i=1}$}\,q_iR^{(d+1)/d}_i\Big)-d\ln\Big(\mbox{$\sum^3_{j=1}$}\,q_jC^{(d+1)/d}_j\Big),
\end{aligned}
\end{equation}
and $d=\Delta-1$. It is a non-trivial task to directly compare the values of $\overline{\Phi^{S}}$ over fixpoints of \eqref{eq:treePotts}. Instead, we will solve a relaxed version of the problem, seeking to maximize $\overline{\Phi^{S}}$ over non-negative $q_i$'s which satisfy $q_1+q_2+q_3=q$. If this maximum happens to occur for integer $\qb$ and the respective values of $R_i$'s and $C_j$'s are solutions of \eqref{eq:treePotts}, then we have also found the solution to the original maximization problem.  It turns out that all of the above are satisfied iff $q$ is even.


To formalize the argument, for non-negative $q_i$'s such that $q_1+q_2+q_3=q$, define
\begin{equation}\label{eq:overphi}
\overline{\Phi}(\qb):=\max_{\rb,\cb} \overline{\Phi^{S}}(\qb,\rb,\cb)
\end{equation}
where the maximum is over $\rb=(R_1,R_2,R_3)^{\T}$, $\cb=(C_1,C_2,C_3)^{\T}$ which satisfy
\begin{equation}\label{eq:pospos1}
\begin{gathered}
\mbox{$\sum^3_{i=1}$}\, q_iR_i\,\mbox{$\sum^3_{j=1}$}\,q_jC_j+(B-1)\mbox{$\sum^3_{i=1}$}\,q_i R_iC_i>0,\\
R_1, R_2, R_3,C_1,C_2,C_3\geq0.
\end{gathered}
\end{equation}
It is simple to see that in the region \eqref{eq:pospos1}, $\overline{\Phi^S}$ is well defined. It is not completely immediate that the maximum in \eqref{eq:overphi} is well defined since the region \eqref{eq:pospos1} is not compact. This is a consequence of the following scale-free property of $\overline{\Phi^{S}}$ with respect  to $\rb$ and $\cb$:
\begin{equation}\label{eq:scalefree}
\text{for every $c_1,c_2>0$ it holds that }\overline{\Phi^{S}}(\qb,c_1\rb,c_2\cb)=\overline{\Phi^{S}}(\qb,\rb,\cb).
\end{equation}
Using \eqref{eq:scalefree}, it is simple to obtain the following.
\begin{lemma}\label{lem:fractionalmax}
Let $B\geq 0$ and $q\geq 2$. For all $q_1,q_2,q_3\geq 0$ which satisfy $q_1+q_2+q_3=q$, the maximum in \eqref{eq:overphi} is well defined. Moreover, the maximum of $\overline{\Phi}(q_1,q_2,q_3)$ over  all such $q_1,q_2,q_3$ is attained.
\end{lemma}

We next seek to connect the maximizers of \eqref{eq:overphi} with solutions of \eqref{eq:treePotts}. To do this, we first need to consider whether the maximum in \eqref{eq:overphi}  happens on the boundary of the region \eqref{eq:pospos1}; it turns out that the maximum can happen at the boundary $R_i=0$ or $C_i=0$ if $q_i$ is close to zero. While the boundary cases are an artifact of allowing $q_i$'s to be non-integer, we will need to treat them explicitly to find the maximum of $\overline{\Phi}$.
\begin{definition}\label{def:goodtriples}
A triple $\qb=(q_1,q_2,q_3)$ is \emph{good} if  the $\rb,\cb$ which achieve the maximum in  \eqref{eq:overphi} satisfy: for $i=1,2,3$,   $q_i>0$ implies $R_i,C_i>0$. A triple $\qb=(q_1,q_2,q_3)$ is \emph{bad} if it is not good.
\end{definition}
To complete the connection, we need to further restrict the set of triples $\qb$. To motivate this restriction, note that if we consider the region \eqref{eq:pospos1} in the subspace $R_1=R_2$ and $C_1=C_2$, we obtain $\overline{\Phi}(q_1+q_2,q_3,0)\leq \overline{\Phi}(q_1,q_2,q_3)$.  To avoid degenerate cases, we consider only triples $\qb$ where such simple inequalities do not hold at equality.

\begin{definition}
Let $t=2$ or $3$. A triple $\qb=(q_1,q_2,q_3)$ is called $t$-maximal if
exactly $t$ of the $q_i$'s are non-zero and for all distinct $i,j,k\in\{1,2,3\}$ (with $q_iq_j>0$) it holds that $\overline{\Phi}(q_i+q_j,q_k,0)<\overline{\Phi}(\qb)$.
\end{definition}

Our interest is in maximal good triples $\qb=(q_1,q_2,q_3)$. This is justified by the following lemma, whose proof is given in Section~\ref{sec:remainingproofsmax}.
\begin{lemma}\label{lem:maximalgood}
Suppose that $q_1,q_2,q_3$ are non-negative \emph{integers} and the triple $\qb=(q_1,q_2,q_3)$ is $t$-maximal and good. Then, the $\rb,\cb$ which achieve the maximum in \eqref{eq:treePotts} specify a $t$-supported fixpoint of \eqref{eq:treePotts} of type $(q_1,q_2,q_3)$.
\end{lemma}
Thus to prove Lemma~\ref{lem:mainlemmapotts}, it suffices to prove that the triple $(q/2,q/2,0)$ is 2-maximal and good and that the maximum of   $\overline{\Phi}(\qb)$ is achieved at $(q/2,q/2,0)$. The next lemma examines which maximal good triples can be a maximum of $\overline{\Phi}$.
\begin{lemma}\label{lem:goodmax}
Let $q\geq 3$ and $0\leq B<1$. There do not exist 3-maximal good triples $\qb$ which maximize $\overline{\Phi}(\qb)$. The only 2-maximal good triples $\qb$ where a maximum of  $\overline{\Phi}(\qb)$ can occur are $(q/2,q/2,0)$ or its permutations.
\end{lemma}
Lemma~\ref{lem:goodmax} is not sufficient to yield Lemma~\ref{lem:mainlemmapotts} because the maximum of $\overline{\Phi}(\qb)$ can occur at a bad triple $\qb$. This possibility is excluded by the following lemma.
\begin{lemma}\label{lem:badmax}
Let $q\geq 3$ and $0\leq B<\frac{\Delta-q}{\Delta}$. There do not exist bad triples $\qb$ which maximize $\overline{\Phi}(\qb)$.
\end{lemma}
Using Lemmas~\ref{lem:goodmax} and~\ref{lem:badmax}, we can now give the proof of Lemma~\ref{lem:mainlemmapotts}.
\begin{proof}[Proof of Lemma~\ref{lem:mainlemmapotts}]
The maximum of $\overline{\Phi}(\qb)$ over triples $\qb$ is attained by Lemma~\ref{lem:fractionalmax}. This maximum can happpen either at a bad or a good triple $\qb$. Maxima at bad triples $\qb$ are excluded by Lemma~\ref{lem:badmax}. Maxima at 3-maximal good triples  are excluded by the first part of Lemma~\ref{lem:goodmax}. Thus, the maximum must happen at a (good) triple of the form $\qb=(q_1,q_2,0)$. The latter can be either 2-maximal or not. If it is not 2-maximal, the maximum must equal $\overline{\Phi}(\qb)$, which in the regime $0\leq B<\frac{\Delta-q}{\Delta}$ is excluded by Lemma~\ref{lem:transinva}. Thus, the maximum must happen at a 2-maximal good triple, which  Lemma~\ref{lem:goodmax} asserts that it must be the triple $(q/2,q/2,0)$. Finally, for $q$ even, by Lemma~\ref{lem:maximalgood} the $\rb,\cb$ which achieve the maximum in \eqref{eq:overphi} correspond to a  $2$-supported fixpoint of \eqref{eq:treePotts} of type $(q/2,q/2,0)$, as wanted.
\end{proof}
For the proofs of Lemmas~\ref{lem:goodmax} and~\ref{lem:badmax}, we will often perturb the values of $q_i$'s. The following lemma, which is proved in Section~\ref{sec:remainingproofsmax} will be very helpful.
\begin{lemma}\label{lem:derphis}
Let $\qb=(q_1,q_2,q_3)$ and $I=\{i\, |\, q_i>0\}$. Suppose that $\rb,\cb$ achieve the maximum in \eqref{eq:overphi}. Then, for $i\in I$ it holds that
\begin{equation}\label{prd}
\frac{\partial \overline{\Phi^{S}}}{\partial q_i}(\qb,\rb,\cb)=\frac{R_i\sum_jq_j C_j+C_i\sum_j q_j R_j+(d-1)(1-B)R_iC_i}{\sum_j q_j R_j\sum_jq_j C_j+(B-1)\sum_jq_j R_jC_j}.
\end{equation}
Moreover, if there exist $i,j\in I$ such that $\frac{\partial \overline{\Phi^{S}}}{\partial q_i}-\frac{\partial \overline{\Phi^{S}}}{\partial q_j}\neq 0$, the maximum of $\overline{\Phi}$ is \emph{not} achieved at the triple $\qb$.
\end{lemma}

\subsection{Good triples: proof of Lemma~\ref{lem:goodmax}}\label{sec:tequals3}

We first prove the statement of the lemma for 3-maximal good triples $\qb=(q_1,q_2,q_3)$,  the proof for 2-maximal good triples will easily be inferred by appropriately modifying the arguments in the special case $q_2=0$.

Let $\qb=(q_1,q_2,q_3)$ be a 3-maximal good triple. Since $\qb$ is 3-maximal all of the $q_i$'s are positive. Moreover, $\qb$ is good,  and hence the maximum in \eqref{eq:overphi} for $\qb$ is attained at positive $R_i$'s and $C_j$'s. Thus, the $R_i$'s and $C_j$'s satisfy $\partial\overline{\Phi^{S}}/\partial R_i=\partial\overline{\Phi^{S}}/\partial C_j=0$ which give
\begin{equation}\label{zzz}
R_i^{1/d} \propto q_1 C_1 + q_2 C_2 + q_3 C_3 + (B-1)
C_i, \quad C_j^{1/d} \propto q_1 R_1 + q_2 R_2 + q_3
R_3 + (B-1) R_j.
\end{equation}
Since $\qb$ is 3-maximal, we may assume that $\rb$ is such that $R_i\neq R_j$ for all $i\neq j$.  Otherwise, if for  example $R_1=R_2$,  by \eqref{zzz}, we have $C_1=C_2$ as well, so that $\overline{\Phi}(q_1,q_2,q_3)=\overline{\Phi}(q_1+q_2,q_3,0)$, contradicting the 3-maximality of $\qb$. Thus, we may assume a strict ordering of the $R_i$'s, which by \eqref{zzz} implies the reverse ordering of the $C_j$'s. W.l.o.g., we will use the following ordering:
\begin{equation}\label{eq:assumptionsRC}
\begin{gathered}
R_1> R_2> R_3>0\quad \mbox{and} \quad 0<C_1< C_2< C_3.
\end{gathered}
\end{equation}

The following lemma, together with the second part of Lemma~\ref{lem:derphis}, establishes that the maximum of $\overline{\Phi}$ cannot occur at a 3-maximal triple.
\begin{lemma}\label{lem:contradictionargument}
Suppose that $R_i$'s and $C_j$'s satisfy \eqref{zzz} and \eqref{eq:assumptionsRC}. If $R_1/R_3\neq C_3/C_1$ then $\frac{\partial \overline{\Phi^{S}}}{\partial q_1}-\frac{\partial \overline{\Phi^{S}}}{\partial q_3}\neq 0$. If $R_1/R_3=C_3/C_1$ then $\frac{\partial \overline{\Phi^{S}}}{\partial q_1}-\frac{\partial \overline{\Phi^{S}}}{\partial q_2}\neq 0$.
\end{lemma}

We next give the proof of Lemma~\ref{lem:contradictionargument}. We will utilize Lemma~\ref{lem:derphis} by specifying a particular scaling of the $R_i$'s and $C_j$'s which will be beneficial. To do this, set
\begin{equation}\label{eq:subsrc}
r_1^d= R_1/R_3, r_2^d=R_2/R_3, c_2^d=C_2/C_1, c_3^d=C_3/C_1.
\end{equation}
The $R_i$'s  and $C_j$'s may be recovered from $r_i$'s, $c_j$'s using
\begin{equation}\label{eq:subs}
R_1\propto r_1^d,\, R_2\propto r_2^d,\, R_3\propto 1,\quad \mbox{and} \quad
C_1\propto 1,\, C_2\propto c_2^d,\, C_3\propto c_3^d.
\end{equation}
Translating \eqref{eq:assumptionsRC} into $r_1,r_2,c_2,c_3$ gives
\begin{equation}\label{eq:gathered}
\begin{gathered}
r_1>r_2> 1\mbox{ and }c_3> c_2>1.
\end{gathered}
\end{equation}
Moreover, dividing appropriate pairs of \eqref{zzz}, we also obtain
\begin{equation}\label{eq:r1c3r2c2}
\begin{gathered}
r_1=\displaystyle\frac{B+q_1-1+q_2c_2^d+q_3c_3^d}{q_1+q_2c_2^d+(B+q_3-1)c_3^d},\   c_3=\displaystyle\frac{B+q_3-1+q_2r_2^d+q_1r_1^d}{q_3+q_2r_2^d+(B+q_1-1)r_1^d},\\
r_2=\displaystyle\frac{q_1+(B+q_2-1) c_2^d+q_3c_3^d}{q_1+q_2c_2^d+(B+q_3-1)c_3^d},\ 
c_2=\displaystyle\frac{q_3+(B+q_2-1)r_2^d+q_1r_1^d}{q_3+q_2r_2^d+(B+q_1-1)r_1^d}.
\end{gathered}
\end{equation}

It can easily be verified that this system of equations gives
\begin{gather}
q_1= \frac{(1-B) f(r_1,c_3)+q_2P   \left(c_2^d-c_3^d r_2^d\right)}{P \left(r_1^d c_3^d-1\right)},\ q_3=\frac{(1-B)f(c_3,r_1)+q_2P   \left(r_2^d-r_1^d c_2^d\right)}{P \left(r_1^d c_3^d -1\right)}\label{eq:q1q3}\\[0.1cm]
r_2=\frac{r_1c_3^d -1-c_2^d(r_1-1)}{c_3^d-1},\ r_2^d= \frac{r_1^d c_3-1-c_2(r_1^d-1)}{c_3-1}\label{eq:r2c2b},\\
f(x,y):=x^{d+1} y^{d+1}-x^d y^{d+1}-x y^{d+1}+y^d+y-1,\
P:=(r_1-1)(c_3-1)>0.\notag
\end{gather}
We will need the following lemma.
\begin{lemma}\label{lela}
Assume that $q_1,q_2,q_3,r_1,r_2,c_2,c_3$ satisfy \eqref{eq:gathered}, \eqref{eq:q1q3}, \eqref{eq:r2c2b}. If $r_1=c_3$ then $r_2=c_2$ and  $q_1=q_3$.
\end{lemma}
\begin{proof}[Proof of Lemma~\ref{lela}]
We prove that $r_1=c_3$ implies
$r_2=c_2$. Once this is done, \eqref{eq:q1q3} easily gives that
$r_1=c_3$ implies $q_1=q_3$ as well, thus proving the lemma.

So, suppose that $z=r_1=c_3$ and for the sake of contradiction
assume $r_2\neq c_2$. By \eqref{eq:gathered} we obtain that $r_2,c_2\in (1,z)$. Eliminating
$r_2$ from \eqref{eq:r2c2b} we obtain that $c_2$ (and by a
symmetric argument $r_2$) satisfies
\begin{equation*}
g(s):=\left(\frac{z^{d+1}-1 - s^d
(z-1)}{z^d-1}\right)^d+\frac{s(z^d-1)-(z^{d+1}-1)}{z-1}=0.
\end{equation*}
In fact, $g(1)=g(z)=0$ as well, so that $g$ has at least four
distinct roots in $[1,z]$. It follows that $g'(s)=0$ has at least
three distinct solutions in $[1,z]$, say $s_i$ for $i=1,2,3$. As a
consequence of  $g'(s_i)=0$, we easily obtain that the $s_i$'s
satisfy $h(s_i)=c$ where $h(s):=(z^{d+1}-1)s - s^{d+1} (z-1)$ and
$c$ is a constant which depends only on $z,\, d$. Thus, $h'(s)=0$
has at least two distinct solutions in $[1,z]$ which is clearly
absurd.
\end{proof}

\begin{proof}[Proof of Lemma~\ref{lem:contradictionargument}]
Set
\[DIF_{13}:=\frac{\partial \overline{\Phi^{S}}}{\partial q_1}-\frac{\partial \overline{\Phi^{S}}}{\partial q_3}, \, DIF_{12}:=\frac{\partial \overline{\Phi^{S}}}{\partial q_1}-\frac{\partial \overline{\Phi^{S}}}{\partial q_2},\, S:=\mbox{$\sum_{i}$}\,q_i R_i\,\mbox{$\sum_{j}$}\,q_j C_j+(B-1)\mbox{$\sum_{i}$}\,q_i R_i C_i.\]
We use the expressions~\eqref{prd} for the derivatives. The denominators in the expressions are the same, so we may ignore them. Moreover,  the expressions therein are scale-free, consequently in order to write the derivatives with respect to $r_i$'s and $c_j$'s we just need to make the substitutions \eqref{eq:subs}.

To prove the first part of the lemma, we eliminate $q_1,q_3$ from the resulting expression for $DIF_{13}$ using
\eqref{eq:q1q3}. This substitution has the beneficial effect
of eliminating $q_2,r_2$ from the final expression. After
straightforward calculations, we obtain the following:
\begin{equation}\label{eq:dif13}
\begin{aligned}
DIF_{13}&=-\frac{(1-B)\, g(r_1,c_3)}{S(r_1-1)(c_3-1)},\,\mbox{ where } \\ g(r_1,c_3)&:=(r_1-c_3)(r_1^d-1)(c_3^d-1)-d (r_1-1)(c_3-1)(r_1^d-c_3^d).
\end{aligned}
\end{equation}
It can easily be seen that for $r_1,c_3>1$, it holds that $g(r_1,c_3) = 0$ iff $r_1 =c_3$ iff $R_1/R_3=C_3/C_1$ as desired.

We next prove the second part of the lemma. Since $R_1/R_3=C_3/C_1$, we have $r_1=c_3$ and by Lemma~\ref{lela}, $r_2=c_2$ and $q_1=q_3$. Using these, \eqref{eq:q1q3} and \eqref{eq:r2c2b} simplify to
\begin{gather}
q_1=\frac{(1-B) \big(r_1^{d+1}-1\big)-q_2 r_2^d(r_1-1)}{(r_1-1)
\left(r_1^d+1\right)},\ r_2=\frac{r_1^{d+1}
-1-r_2^d(r_1-1)}{r_1^d-1}.\label{eq:q1r2}
\end{gather}
Moreover, using the substitutions \eqref{eq:subs} and $q_1=q_3$, we
obtain
\begin{align*}
DIF_{12}&=\frac{(q_1r_1^d+q_2 r_2^d+q_1)(r_1^d-2r_2^d+1)+(d-1)(1-B)(r_1^d-r_2^{2d})}{S}\\
&=-\frac{(1-B)\big[(d-1)(r_1-1)r_2^{2d}+2r_2^d(r_1^{d+1}-1)-(r_1^{2d+1}+d
r_1^{d+1}-d r_1^d-1)\big]}{(r_1-1)S},
\end{align*}
where in the second equality we substituted the value of $q_1$ from
\eqref{eq:q1r2}. Observe that the numerator is a quadratic
polynomial in $r_2^d$ and, by inspection, for $r_1>1$, its roots are of opposite
sign. Thus, $DIF_{12}=0$ iff
$r_2^d= \rho_1$, where
\begin{equation*}
\rho_1(r_1):=\frac{\sqrt{D}-(r_1^{d+1}-1)}{(d-1)(r_1-1)} \mbox{
and } D:=\big(dr_1^{d+1}- (d-1) r_1^d+1\big)
\big(r_1^{d+1}+(d-1)r_1-d\big).
\end{equation*}
For the sake of contradiction, suppose that $r_2^d=\rho_1$. Then
\eqref{eq:q1r2} gives that $r_2=\rho_2$, where
\begin{equation*}
\rho_2(r_1):=\frac{d(r_1^{d+1}-1)-\sqrt{D}}{(d-1)(r_1^d-1)}.
\end{equation*}
Thus $\rho_1=\rho_2^d$. We obtain a contradiction by showing
that for every $r_1>1$, it holds that $\rho_2^d<\rho_1$ or
equivalently $d\ln \rho_2<\ln \rho_1$. It is easy to see that in
the limit $r_1\downarrow 1$ the inequality is satisfied at
equality, thus it suffices to prove that the derivative of the rhs
w.r.t $r_1$ is greater than the respective derivative of
the l.h.s. for $r_1>1$.

This differentiation is cumbersome but otherwise straightforward. The final result is
\begin{align}
\frac{1}{\rho_1} \frac{\displaystyle \partial \rho_1}{\partial r_1}-\frac{d}{\rho_2} \frac{\partial \rho_2}{\partial r_1}&=\frac{ (d+1) g(r_1)h(r_1)}{2 (r_1-1)\left(r_1^d-1\right) \left(\sqrt{D}-(r_1^{d+1}-1)\right) \left(d (r_1^{d+1}-1)-\sqrt{D}\right)},\label{eq:differenceder}\\
g(r_1):&=r_1^{2 d}-d^2 r_1^{d+1}+2(d^2-1)r_1^{d}-d^2r_1^{d-1}+1,\notag\\
h(r_1):&=(d+1) (r_1^{d+1}-1)-(d-1)( r_1^d-1)-2 \sqrt{D}.\notag
\end{align}
Note that the denominator in the r.h.s. of \eqref{eq:differenceder} is positive
for $r_1>1$: the terms involving $\sqrt{D}$ are positive since
they are the numerators of $\rho_1,\ \rho_2$. The final part of
the proof consists of proving that $g(r_1)>0$ and $h(r_1)>0$ for
$r_1>1$.

The polynomial $g$ has 4 sign changes and hence, by the Descartes' rule of signs has at most 4
positive roots. In fact, a tedious calculation shows that $r_1=1$
is a root by multiplicity 4, thus proving that $g(r_1)>0$ for
$r_1>1$. To prove that $h(r_1)>0$ for $r_1>1$, note the identity
\begin{equation*}
\big[(d+1) (r_1^{d+1}-1)-(d-1)( r_1^d-1)\big]^2-4D=(d-1)^2 (r_1-1)^2 (r_1^d-1)^2.
\end{equation*}
This completes the proof.
\end{proof}

To prove the second part of Lemma~\ref{lem:goodmax}, assume that $\qb=(q_1,q_2,q_3)$ is a 2-maximal good triple. Since $\qb$ is 2-maximal, w.l.o.g. we may assume that $q_2=0$. Note that the values of $R_2,C_2$ do not affect the value of the derivatives $\partial \overline{\Phi^S}/\partial q_1, \partial \overline{\Phi^S}/\partial q_3$ when $q_2=0$. Similarly, \eqref{eq:q1q3} continues to hold even when $q_2=0$. Thus, the proof of the first part of Lemma~\ref{lem:contradictionargument} carries through verbatim. In particular, if $R_1/R_3\neq C_3/C_1$, then $\partial \overline{\Phi^S}/\partial q_1-\partial \overline{\Phi^S}/\partial q_3\neq 0$. By the second part of Lemma~\ref{lem:derphis}, it follows that $\qb=(q_1,0,q_3)$ cannot be a maximum unless $R_1/R_3= C_3/C_1$. In this case, \eqref{eq:q1q3} gives $q_1=q_3$. Since $q_1+q_3=q$, we obtain that the only 2-maximal good triples where the maximum of $\overline{\Phi}$ may occur are $(q/2,q/2,0)$ or its permutations, as desired.

This concludes the proof of Lemma~\ref{lem:goodmax}.

\subsection{Bad triples: proof of Lemma~\ref{lem:badmax}}
To get a handle on bad triples, we first give necessary conditions so that the maximum in \eqref{eq:overphi} happens at the boundary. The proof of the following lemma is given in Section~\ref{sec:remainingproofsmax}.
\begin{lemma}\label{lem:classifybad}
Let $0\leq B<1$. For a triple  $\qb=(q_1,q_2,q_3)$, let $\rb,\cb$ achieve the maximum in \eqref{eq:overphi}. Then, if $q_i>0$, the following implications hold:
\begin{equation*}
R_i=0 \Rightarrow \mbox{$\sum_j$}\, q_j C_j\leq (1-B)C_i, \quad C_i=0 \Rightarrow \mbox{$\sum_j$}\, q_j R_j\leq (1-B)R_i.
\end{equation*}
In particular, if $q_i>1-B$ it holds that $R_i,C_i>0$. Hence, for every $q\geq 3$ there exists $i\in\{1,2,3\}$ such that $R_i,C_i>0$.
\end{lemma}

We next examine bad triples. Note that a bad triple $\qb=(q_1,q_2,q_3)$, by the second part of Lemma~\ref{lem:derphis}, must have at least two positive entries. We consider cases whether the triple $\qb$ has two or three positive entries. We start with the case where exactly two of the $q_i$'s are positive. We assume throughout the rest of the section that $\rb,\cb$ achieve the maximum in \eqref{eq:overphi}.

Let $\qb=(q_1,q_2,0)$ be a bad triple where $q_1,q_2>0$. Since $\qb$ is bad, at least one of $R_1,R_2,C_1,C_2$ is zero. Wlog, we may assume $C_2=0$. By the second part of Lemma~\ref{lem:classifybad}, it follows that $R_1,C_1>0$.  There are two cases to consider.
\begin{equation}\label{eq:twosupport}
\mbox{(I) $R_2=0$, (II) $R_2>0$.}
\end{equation}
Case (I) is straightforward: by the first part of Lemma~\ref{lem:derphis}, we trivially have $\frac{\partial \overline{\Phi^S}}{\partial q_1}>0$ and $\frac{\partial \overline{\Phi^S}}{\partial q_2}=0$, so that the second part of Lemma~\ref{lem:derphis} yields that $\qb$ does not maximize $\overline{\Phi}$.

We next examine case (II). Since $\overline{\Phi^S}$ is scale-free (see \eqref{eq:scalefree}), we may assume that $C_1=1$. Since $R_1,R_2$ are positive, it holds that $\partial \overline{\Phi^S}/\partial R_1=\partial \overline{\Phi^S}/\partial R_2=0$, yielding
\[R_1\propto y^d,\ R_2\propto 1, \mbox{ where } y=(q_1+B-1)/q_1.\]
Expressing $q_1,q_2$ in terms of $y$ and substituting in $\overline{\Phi^S}$, we obtain the value of $\overline{\Phi}(\qb)$:
\begin{equation*}
\overline{\Phi}(\qb)=\log h(y), \mbox{ where } h(y):=\frac{(1-B) \left(q (1-y)-(1-B)(1-y^{d+1})\right)}{(1-y)^2}.
\end{equation*}
Let $I$ be the interval $[0,(q+B-1)/q]$. Note that for any $y\in I$, there exists a positive $q_1\in[0,q]$ such that $y=(q_1+B-1)/q_1$. Obviously, if $\qb$ maximizes $\overline{\Phi}$, it must be the case that $y$ maximizes $h(y)$ in the interval $I$. We compute $h'(y)$.
\begin{equation*}
h'(y)=\frac{(1-B)\,r(y)}{(1-y)^3}, \mbox{ where } r(y):=q(1-y)-(1-B)\big((d-1)y^{d+1}-(d+1)y^d+2\big).
\end{equation*}
It is immediate to see that $r(y)$ is convex for $y\in [0,1]$. Since $r(0)=q-2(1-B)>0$ and $r(1)=0$, we obtain that either
\begin{align*}
\mbox{(i)}& \mbox{ $r(y)>0$ for all $y\in I$,\ \   or}\\[0.1cm]
\mbox{(ii)}& \mbox{ $\exists\, y_o\in I$: $r(y_o)=0$, $r(y)> 0$ iff $y< y_o$.}
\end{align*}

In case (i), $h(y)$ is increasing and hence $h(y)$ is maximized at $y=(q+B-1)/q$. This value of $y$ corresponds to $q_1=q$ and thus $\Phi(\qb)=\Phi(q,0,0)$.

In case (ii), we have $h(y)\leq h(y_o)$.  The value of $q_1$ corresponding to $y_o$ is $q_o:=(1-B)/(1-y_o)$.  We will show that the maximum in \eqref{eq:overphi} does not happen at the boundary $C_2=0$ when $\qb=(q_o,q-q_o,0)$, implying that $h(y_o)$ does not equal $\overline{\Phi}(\qb)$ and hence the maximum of $\overline{\Phi}$ as well. To prove the former, we utilize the first part of Lemma~\ref{lem:classifybad}. In particular, we prove that
\begin{equation}\label{ineq:toprove}
q_o y_o^{d}+(q-q_o)> (1-B).
\end{equation}
Note that $r(y_o)=0$ yields $q=(1-B)\big((d-1)y^{d+1}_o-(d+1)y^d_o+2\big)/(1-y_o)$. Plugging this expression into \eqref{ineq:toprove}, we only need to show that
\begin{equation}
\frac{(d-1)y^{d+1}_o-d y^d_o+1}{1-y_o}> 1 \mbox{ or } (d-1)y^{d}_o+1 > d y^{d-1}_o,
\end{equation}
which holds by the AM-GM inequality for any positive $y_o \neq 1$.

Let $\qb=(q_1,q_2,q_3)$ be a bad triple where all of the $q_i$'s are positive. Since $\qb$ is bad, at least one of the $R_i$'s and $C_j$'s is zero. W.l.o.g. we may assume $C_2=0$. Moreover, by the second part of Lemma~\ref{lem:classifybad}, we may also assume that $R_1,C_1>0$.  There are four cases to consider.
\begin{equation*}
\mbox{(I) $R_2=0$,\ \  (II) $R_2,R_3>0$, $C_3=0$,\ \  (III) $R_2,R_3,C_3>0$, \ \ (IV) $R_2,C_3>0$, $R_3=0$.}
\end{equation*}
We omitted the case $R_2>0$ and $R_3=C_3=0$, which is identical to case (I) after renaming the $q_i$'s.

Case (I) is straightforward: since $R_2=C_2=0$, \eqref{prd} gives $\partial \overline{\Phi^S}/\partial q_2=0$. Since at least one of $\partial \overline{\Phi^S}/\partial q_1,\ \partial \overline{\Phi^S}/\partial q_3$ is positive, the second part of Lemma~\ref{lem:derphis} yields that $\qb$ does not maximize $\overline{\Phi}$.

We next examine case (II). Since $\overline{\Phi^S}$ is scale-free (see \eqref{eq:scalefree}), we may substitute $C_1=1$. Setting the derivatives of $\partial \overline{\Phi^S}/\partial R_1, \partial \overline{\Phi^S}/\partial R_2, \partial \overline{\Phi^S}/\partial R_3$ equal to zero, we obtain
\begin{equation*}
R_1\propto (q_1+B-1)^d/q_1^d,\ R_2\propto 1,\ R_3\propto 1.
\end{equation*}
It follows that $\overline{\Phi}(\qb)=\overline{\Phi}(q_1,q_2+q_3,0)$ and hence the maximum of $\overline{\Phi}$ does not occur at $\qb$ by the argument for case (II) in \eqref{eq:twosupport}.

We next examine case (III). The partial  derivatives of $\overline{\Phi^S}$ with respect to $R_1$, $R_2$, $R_3$, $C_1$,  $C_3$ must vanish so  we obtain
\begin{equation}\label{zzz1}
\begin{gathered}
R_1^{1/d}\propto q_1 C_1+q_2 C_3-(1-B)C_1,\ R_2^{1/d}\propto q_1C_1+q_3 C_3,\ R_3^{1/d}\propto q_1C_1+q_3 C_3-(1-B)C_3, \\
C_1^{1/d}\propto q_1R_1+q_2 R_2+q_3R_3-(1-B)R_1,\ C_3^{1/d}\propto q_1 R_1+q_2R_2+q_3R_3-(1-B) R_3.
\end{gathered}
\end{equation}
If $C_1=C_3$, then $R_1=R_3$ and thus we obtain $\overline{\Phi}(q_1,q_2,q_3)=\overline{\Phi}(q_1+q_3,q_2,0)$, contradicting the maximality of $\qb$ by the argument for case (II) in \eqref{eq:twosupport}. Thus, wlog we may assume $C_1<C_3$. By \eqref{zzz1}, this yields
\begin{equation}\label{eq:assumptionsRC1}
R_2>R_1>R_3, \ C_1<C_3.
\end{equation}
We have the following analogue of Lemma~\ref{lem:contradictionargument}, which proves that the maximum cannot occur at $\qb$ by the second part in Lemma~\ref{lem:derphis}.
\begin{lemma}\label{lem:contradictionargument1}
Suppose that $R_i$'s and $C_j$'s satisfy \eqref{zzz1} and \eqref{eq:assumptionsRC1}. If $R_1/R_3\neq C_3/C_1$ then $\frac{\partial \overline{\Phi^{S}}}{\partial q_1}-\frac{\partial \overline{\Phi^{S}}}{\partial q_3}\neq 0$. If  $R_1/R_3=C_3/C_1$  then $\frac{\partial \overline{\Phi^{S}}}{\partial q_1}-\frac{\partial \overline{\Phi^{S}}}{\partial q_2}\neq 0$.
\end{lemma}

\begin{proof}[Proof of Lemma~\ref{lem:contradictionargument1}]
The proof is analogous to the proof of Lemma~\ref{lem:contradictionargument}, we highlight the main differences. Let $r_1^d=R_1/R_3, r_2^d=R_2/R_3, c_3^d=C_3/C_1$. The $R_i$'s and $C_j$'s may be recovered by the $r_i$'s and $c_j$'s by
\begin{equation}\label{eq:subs1}
R_1\propto r_1^d, R_2\propto r_2^d, R_3\propto 1, \mbox{ and } C_1\propto 1, C_3\propto c_3^d.
\end{equation}
By \eqref{eq:assumptionsRC1}, we have
\begin{equation*}
r_2>r_1>1 \mbox{ and } c_3>1.
\end{equation*}
The expressions for  $r_1,r_2,c_3$ in \eqref{eq:r1c3r2c2} are exactly the same after substituting $c_2=0$. The same is true for \eqref{eq:q1q3}, \eqref{eq:r2c2b}. It follows that the proof for the first part of  Lemma~\ref{lem:contradictionargument} holds verbatim in this case as well (note that the ordering of $r_1,r_2$ is different here but that part of the argument does not use the ordering).

While the proof for the second part of  Lemma~\ref{lem:contradictionargument} does not carry through as simply, the changes are minor. We assume that $r_1=c_3$ and set $DIF_{12}:=\frac{\partial \overline{\Phi^{S}}}{\partial q_1}-\frac{\partial \overline{\Phi^{S}}}{\partial q_2}$. Plugging $r_1=c_3$ and $c_2=0$ in  \eqref{eq:q1q3}, \eqref{eq:r2c2b} and then substituting the resulting expressions in $DIF_{12}$ we obtain
\[DIF_{12}=\frac{(1-B)\, h(r_1)}{r1-1}, \mbox{ where } h(r_1):=(r_1^{2d+1}+d r_1^{d+1}-dr_1^d-1)-\frac{(r_1^{d+1}-1)^{d+1}}{(r_1^d-1)^d}.\]
By a first derivative argument, the function \[g(r_1):=\log\bigg(\frac{(r_1^{d+1}-1)^{d+1}}{(r_1^d-1)^d\,(r_1^{2d+1}+d r_1^{d+1}-dr_1^d-1)}\bigg),\]
is strictly increasing for $r_1>1$ . Thus, $g(r_1)\geq g(+\infty)=0$, which gives $h(r_1)>0$ for all $r_1>1$. This proves that $DIF_{12}\neq 0$, as desired.
\end{proof}

Finally, we examine case (IV). The partial  derivatives of $\overline{\Phi^S}$ with respect to $R_1$, $R_2$, $C_1$, $C_3$ must vanish so  we obtain
\begin{equation}\label{zzz2}
\begin{gathered}
R_1^{1/d}\propto q_1 C_1+q_3 C_3-(1-B)C_1,\ R_2^{1/d}\propto q_1C_1+q_3 C_3, \\
C_1^{1/d}\propto q_1R_1+q_2 R_2-(1-B)R_1,\ C_3^{1/d}\propto q_1 R_1+q_2R_2.
\end{gathered}
\end{equation}
Note that we have $R_1<R_2$ and $C_1<C_3$.
\begin{lemma}\label{lem:finishingproof}
If $R_2/R_1\neq C_3/C_1$ then either $\frac{\partial \overline{\Phi^{S}}}{\partial q_2}-\frac{\partial \overline{\Phi^{S}}}{\partial q_3}\neq0$ or $\frac{\partial \overline{\Phi^{S}}}{\partial q_1}-\frac{\partial \overline{\Phi^{S}}}{\partial q_2}\neq 0$. If $R_2/R_1= C_3/C_1$ and $\frac{\partial \overline{\Phi^{S}}}{\partial q_1}-\frac{\partial \overline{\Phi^{S}}}{\partial q_2}=0$, then the maximum in \eqref{eq:overphi} does not happen at the boundary $C_2=0$.
\end{lemma}
\begin{proof}[Proof of Lemma~\ref{lem:finishingproof}]
The approach for the first part is similar the proof of Lemma~\ref{lem:contradictionargument}. Set $r_2^d=R_2/R_1$ and $c_3^d=C_3/C_1$, so that $r_1,c_3>1$. Dividing appropriate pairs in \eqref{zzz2}, we obtain
\begin{equation}\label{eq:r2c3123}
r_2=\frac{q_1+q_3 c_3^d}{(q_1+B-1)+q_3c_3^d},\quad c_3=\frac{q_1+q_2 r_2^d}{(q_1+B-1)+q_2r_2^d}.
\end{equation}
It follows that
\begin{equation*}
q_2=\frac{ q_1-(q_1+B-1) c_3}{r_2^{d}(c_3-1)},\quad q_3= \frac{ q_1-(q_1+B-1)r_2}{c_3^{d}(r_2-1)}.
\end{equation*}
Using these, we obtain
\begin{align}
\frac{\partial \overline{\Phi^{S}}}{\partial q_2}-\frac{\partial \overline{\Phi^{S}}}{\partial q_3}=0& \Rightarrow f(r_2)=f(c_3),\mbox{ where } f(x):= \frac{x^{d+1}}{x-1},\label{eq:mainnnn1}\\
\frac{\partial \overline{\Phi^{S}}}{\partial q_1}-\frac{\partial \overline{\Phi^{S}}}{\partial q_3}=0& \Rightarrow r_2^{d+1}(c_3-1)-(d+1)r_2c_3+d(r_2+c_3)-(d-1)=0.\label{eq:mainnnn2}
\end{align}
From \eqref{eq:mainnnn2}, we obtain  
\begin{equation}\label{eq:valueofc3}
c_3=g(r_2), \mbox{ where } g(r_2):=\frac{r_2^{d+1}-d r_2+(d-1)}{r_2^{d+1}-(d+1) r_2+d}.
\end{equation}
It follows that $r_2=c_3$ is equivalent to 
\begin{equation}\label{eq:mainnnn3}
r_2^{d+1}=(d+1)r_2-(d-1).
\end{equation}
It is straightforward to check that \eqref{eq:mainnnn3} has exactly one solution for $r_2>1$, say $r_2=x$. Using the expression for $c_3$ from \eqref{eq:valueofc3}, \eqref{eq:mainnnn1} gives
\[h(r_2)=0, \mbox{ where } h(r_2):=r_2^{d+1}-\frac{\left(r_2^{d+1}-d r_2+(d-1)\right)^{d+1}}{\left(r_2^{d+1}-(d+1) r_2+d\right)^d}.\]
A standard calculation (albeit lengthy) shows that $h$ is strictly increasing for $r_2>1$ (and every $d\geq 2$). Moreover, it holds that $h(x)=0$, so that \eqref{eq:mainnnn1} and \eqref{eq:mainnnn2} can only hold simultaneously when $r_2=c_3$, which yields the first part of the lemma.

For the second part, we have $r_2=c_3$ so $r_2$ satisfies \eqref{eq:mainnnn3}. To prove that the maximum does not happen at the boundary $C_2=0$, we use the first part of Lemma~\ref{lem:classifybad}. It suffices to prove that
\begin{equation}\label{eq:mainnnn4}
q_1+q_2r_2^d>(1-B)r_2^d.
\end{equation}
We have that $q_2=q_3=(q-q_1)/2$, so that \eqref{eq:r2c3123} gives
\begin{equation}\label{eq:mainnnn5}
q_1 =\frac{q(r_2^{d+1}- r_2^d)-2r_2(1-B)}{(r_2-1) (r_2^d-2)},\quad q_2=q_3=\frac{q -r_2 (q+B-1)}{( r_2-1) (r_2^d-2)}.
\end{equation}
Plugging \eqref{eq:mainnnn5} into \eqref{eq:mainnnn4} gives  the equivalent inequality
\begin{equation*}
\frac{(1-B) \left(r_2^d+r_2-r_2^{d+1}\right)}{r_2-1}>0
\end{equation*}
To see the latter, use \eqref{eq:mainnnn3} to obtain
\begin{equation*}
r_2^d+r_2-r_2^{d+1}=r_2^d+(d-1)-dr_2>0, \mbox{ for  all } r_2>1 \mbox{ by the AM-GM inequality}.
\end{equation*}
This completes the proof.
\end{proof}

\subsection{Remaining proofs}\label{sec:remainingproofsmax}
\begin{proof}[Proof of Lemma~\ref{oooqw}] W.l.o.g. we may assume that the scaling factors in \eqref{eq:treePotts} are equal to 1. Let $R_i=r_i^d$, $C_i = c_i^{d}$, $r=\sum_{i=1}^q r_i^d$, and $c=\sum_{i=1}^q c_i^d$. We have
\begin{equation*}
r_i = c - (1-B) c_i^d\quad\mbox{and}\quad c_i = r - (1-B) r_i^d,
\end{equation*}
It is clear from this equation that $R_i=R_j$ iff $C_i=C_j$ and hence also $t_R=t_C$. We also obtain that for $i=1,\hdots,q$,
\begin{equation}\label{opopq}
r_i = c - (1-B) (r - (1-B) r_i^d)^d.
\end{equation}
Since $r$ is the sum of $r_i^d$ and the $r_i$ are positive, we have $(1-B)r_i^d <
r$. Fix the values of $r,c$ and let $I$ be the interval where $(1-B)x^d<r$.
Using~\eqref{opopq}, we shall prove that $t_R\leq 3$ by arguing that $f(x) = c - (1-B) (r - (1-B) x^d)^d - x$
has at most $3$ positive roots in the
interval $I$, counted by multiplicities. We have
$$f'(x) = (1-B)^2 d^2 (r - (1-B) x^d)^{d-1} x^{d-1} - 1
= \left(\sum_{i=0}^{d-2} g(x)^i\right) (g(x)-1),$$ where
$$
g(x) = ((1-B) d)^{2/(d-1)} (r - (1-B) x^d) x.
$$
Note that $g(x)>0$ in the interval $I$ and hence all roots of
$f'(x)$ in this interval come from $g(x)-1$. The polynomial
$g(x)-1$ has at most two positive roots by Descartes' rule of signs,
hence $f'(x)$ has at most two positive roots in $I$. Thus, $f(x)$ has at most
three positive roots in $I$, all roots counted by their
multiplicities. This concludes the proof.
\end{proof}

\begin{proof}[Proof of Lemma~\ref{lem:symmetricHessian}]
Let $q'=q/2$. To better align with the results of Section~\ref{sec:tequals3}, let us assume that the fixpoint $(q',0,q')$ maximizes $\Psi_1$. In Section~\ref{sec:tequals3}, we proved that this can be the case only if $R_1/R_3=C_3/C_1$ or (in the parameterization of Section~\ref{sec:tequals3}) $r_1=c_3=:x$ where $x>1$. Equation~\eqref{eq:r1c3r2c2} for $q_2=0$, $q_1=q_3=q'$ gives that $x$ satisfies
\begin{equation}\label{eq:qhalf}
x=\frac{B+q'-1+q'x^d}{q'+(B+q'-1)x^d}.
\end{equation}
It is straightforward to check that \eqref{eq:qhalf} has exactly one solution $x>1$ for all $0\leq B<\frac{\Delta-q}{\Delta}$. The values of $R_1,C_1,R_3,C_3$ may be recovered by \eqref{eq:subs}, which in the case $q_2=0$ give
\begin{equation*}
R_1\propto x^d, R_3\propto 1 \mbox{ and } C_1\propto 1, C_3\propto x^d.
\end{equation*}
This proves the second part of the lemma. For the first part, to check Jacobian stability, we proceed as in the proof of Lemma~\ref{lem:transinva}. The eigenvalues of the matrix $\L$ in this case can be computed easily as well. They are given by 
$\pm1$ by multiplicity 1, $\pm \lambda_1$ by multiplicity $q-2$ and $\pm (B+q-1)\lambda_1^2$ by multiplicity 1, where 
\begin{equation*}
\lambda_1:=\frac{(1-B)x^{d/2}}{\sqrt{(q'+(B+q'-1)x^d)(B+q'-1+q'x^d)}}.
\end{equation*}
To prove that the absolute value of the eigenvalues different from 1 is less than $1/d$, it suffices to prove that $\lambda_1< 1/d$. Use \eqref{eq:qhalf} to solve for $q'$ and plug the value  into the expression for $\lambda_1$. This yields that $\lambda_1$ is equal to $x^{(d-1)/2}(x-1)/(x^d-1)$, which by the AM-GM inequality is less than $1/d$ for $x>1$. 
\end{proof}

\begin{proof}[Proof of Lemma~\ref{lem:fractionalmax}]
For non-negative $\qb=(q_1,q_2,q_3)$ with $q_1+q_2+q_3=q$, consider the function 
\begin{equation}\label{eq:FFFFF}
\overline{F}(\qb)=\max_{\rb,\cb}F(\qb,\rb,\cb), \mbox{ where } F(\qb,\rb,\cb):=\mbox{$\sum^3_{i=1}$}\, q_iR_i\,\mbox{$\sum^3_{j=1}$}\,q_jC_j+(B-1)\mbox{$\sum^3_{i=1}$}\, q_i R_iC_i,
\end{equation}
and the maximum is over the compact region (by restricting to $R_i=C_i=0$ whenever $q_i=0$)
\begin{equation}\label{eq:bbals}
\begin{gathered}
\mbox{$\sum^3_{i=1}$}\, q_iR^{(d+1)/d}_i\leq 1,\ \mbox{$\sum^3_{j=1}$}\, q_jC^{(d+1)/d}_j\leq1,\\
R_1,R_2,R_3,C_1,C_2,C_3\geq 0.
\end{gathered}
\end{equation}
Note that $\overline{F}(\qb)>0$, since we can set all of the $R_i$'s and $C_j$'s equal to  $x$, where $q x^{(d+1)/d}=1$. Clearly, $\overline{\Phi}(\qb)\geq \ln \overline{F}(\qb)$. Since $\overline{\Phi^{S}}(\qb,\rb,\cb)$ is scale-free with respect to $\rb$ and $\cb$ (see \eqref{eq:scalefree}), we may scale $\rb,\cb$ to satisfy \eqref{eq:bbals} and hence $\overline{\Phi}(\qb)=\sup_{\rb,\cb}\overline{\Phi^{S}}(\qb,\rb,\cb)\leq \ln \overline{F}(\qb)$, proving that $\overline{\Phi}(\qb)= \ln \overline{F}(\qb)$ and consequently the supremum is attained.  

To prove that $\sup_{\qb}\overline{\Phi}(\qb)$ is attained, it clearly suffices to prove that $L:=\sup_{\qb}\overline{F}(\qb)$ is attained. This can be accomplished by using variants of Berge's Maximum Theorem and showing that the function $\overline{F}(\qb)$ is upper semi-continuous. We give a more direct argument, which is similar to the proof of Berge's Maximum Theorem and can also easily be adapted to show that $\overline{F}(\qb)$ is upper semi-continuous.

Note first that $L<\infty$ by a simple application of H{\"o}lder's inequality. Let $\qb_n$, $n=1,2,\hdots$ be a sequence such that $\overline{F}(\qb_n)\uparrow L$. Since the $\qb_n$ lie in a compact region, by restricting to a subsequence we may assume that $\qb_n\rightarrow\qb$. Let $\rb_n,\cb_n$ be maximizers for $\overline{F}(\qb_n)$ in \eqref{eq:FFFFF}. 

Suppose first that $\qb$ has positive entries. Then, for sufficiently large $n$, the maximizers $\rb_n,\cb_n$ lie in a compact set and hence a standard diagonalisation argument yields a convergent subsequence $(\qb_{n_k}, \rb_{n_k},\cb_{n_k})\rightarrow (\qb,\rb,\cb)$. By continuity, $\rb,\cb$ must lie in the region \eqref{eq:bbals} defined by $\qb$ and moreover $\overline{F}(\qb_{n_k})=F(\qb_{n_k}, \rb_{n_k},\cb_{n_k})\rightarrow F(\qb,\rb,\cb)$. Thus $L=F(\qb,\rb,\cb)$ and the supremum is attained.

Suppose now that $\qb$ has an entry equal to zero, say $q_1$, so that $q_{1n}\rightarrow 0$ (with the natural notation for entries of the subsequences). In this setting, $R_{1n},C_{1n}$ might escape to infinity, so assume that $R_{1n},C_{1n}\uparrow \infty$, by restricting to a subsequence if necessary. \eqref{eq:bbals} implies $q_{1n} R_{1n}^{(d+1)/d}, q_{1n} C_{1n}^{(d+1)/d}\leq 1$ and hence $q_{1n} R_{1n}, q_{1n} C_{1n}\rightarrow 0$. Note that $q_{1n} R_{1n} C_{1n}\rightarrow 0$ as well; otherwise there  exists a subsequence with $q_{1n_k} R_{1n_k} C_{1n_k}\geq \epsilon>0$. This contradicts that $\rb_{n_k},\cb_{n_k}$ maximize $F(\qb_{n_k},\cdot,\cdot)$, since setting $R_{1,n_k}=C_{1,n_k}=0$ would maintain feasibility in \eqref{eq:bbals} and achieve a bigger value of $F$ for all sufficiently large $k$ (recall that $B<1$). Thus $q_{1n} R_{1n}, q_{1n} C_{1n},q_{1n} R_{1n} C_{1n}\rightarrow 0$, yielding once again $L=F(\qb,\rb,\cb)$. 
\end{proof}

\begin{proof}[Proof of Lemmas~\ref{lem:maximalgood} and~\ref{lem:derphis}]
We first prove Lemma~\ref{lem:derphis}. Let $I_R=\{i\in I\, |\, R_i>0\}$. For $i\in I_R$, it must hold that $\partial \overline{\Phi^{S}}/\partial R_i=0$. Since $q_i>0$ for $i\in I$, it follows that
\begin{equation}\label{eq:extextextbasic}
R_i^{1/d}\propto \mbox{$\sum_j$}\, q_j C_j-(1-B)C_i \quad \mbox{ for all } i\in I_R,
\end{equation}
and hence
\begin{equation*}
R_i^{(d+1)/d}\propto R_i\big(\mbox{$\sum_j$}\, q_j C_j-(1-B)C_i\big)\mbox{ for all } i\in I.
\end{equation*}
Thus, for $i\in I$ it holds that
\begin{equation}\label{eq:simplifyyy1}
\frac{R_i^{(d+1)/d}}{\sum_{j}q_jR_j^{(d+1)/d}}=\frac{R_i\big(\sum_jq_j C_j-(1-B)C_i\big)}{\sum_j q_j R_j\sum_jq_j C_j+(B-1)\sum_jq_j R_jC_j},
\end{equation}
and an analogous argument for the $C_i$'s gives
\begin{equation}\label{eq:simplifyyy2}
\frac{C_i^{(d+1)/d}}{\sum_{j}q_jC_j^{(d+1)/d}}=\frac{C_i\big(\sum_jq_j R_j-(1-B)R_i\big)}{\sum_j q_j R_j\sum_jq_j C_j+(B-1)\sum_jq_j R_jC_j}.
\end{equation}
Moreover, by a direct calculation we have
\begin{equation}\label{eq:directcalculation}
\frac{\partial \overline{\Phi^{S}}}{\partial q_i}=\frac{(d+1)\big(R_i\sum_jq_j C_j+C_i\sum_j q_j R_j+(B-1)R_iC_i\big)}{\sum_j q_j R_j\sum_jq_j C_j+(B-1)\sum_jq_j R_jC_j}-\frac{dR_i^{(d+1)/d}}{\sum_{j}q_jR_j^{(d+1)/d}}-\frac{dC_i^{(d+1)/d}}{\sum_{j}q_jC_i^{(d+1)/d}}.
\end{equation}
Plugging \eqref{eq:simplifyyy1}, \eqref{eq:simplifyyy2} in \eqref{eq:directcalculation} proves the first part of Lemma~\ref{lem:derphis}.

For the second part of Lemma~\ref{lem:derphis}, assume w.l.o.g. that $q_1,q_2>0$ and $\frac{\partial \overline{\Phi^{S}}}{\partial q_1}-\frac{\partial \overline{\Phi^{S}}}{\partial q_2}>0$. For $\epsilon>0$, consider $\qb'=(q_1+\epsilon,q_2-\epsilon,q_3)$. Since $q_1,q_2$ are positive, for small enough $\epsilon$, $\qb'$ has positive entries which sum to $q$. Moreover, for small enough $\epsilon$ the value of $\overline{\Phi^{S}}$ increases, while still maintaining feasibility in the region \eqref{eq:pospos1}. Hence, $\qb$ does not maximize $\overline{\Phi}$, as desired.

Lemma~\ref{lem:maximalgood} follows easily: just use \eqref{eq:extextextbasic} and the fact that $q_1,q_2,q_3$ are integers to get the alignment with \eqref{eq:treePotts}.
\end{proof}

\begin{proof}[Proof of Lemma~\ref{lem:classifybad}]
Suppose that $q_i>0$ and $\sum_j\, q_j C_j> (1-B)C_i$. We look at the derivative $\partial\overline{\Phi^{S}}/\partial R_i$ evaluated at $R_i=0$:
\begin{equation*}
\frac{\partial\overline{\Phi^{S}}}{\partial R_i}=\frac{q_i(q_1 C_1+q_2 C_2+q_3 C_3-(1-B)C_i)}{\sum_j\, q_j R_j\sum_j\, q_j C_j+(B-1)\sum_j\, q_j R_jC_j}>0
\end{equation*}
Thus, increasing the value of  $R_i$ by a sufficiently small amount, increases the value of $\overline{\Phi^{S}}$. Hence, the maximum cannot be obtained at the boundary $R_i=0$. The second part of the lemma follows immediately from the first part.
\end{proof}

\newcommand{\etalchar}[1]{$^{#1}$}

\appendix

\section{The Small Subgraph Conditioning Method}
\label{sec:small-graph}

In this section, we prove Lemma~\ref{lem:smallgraphstrip} by appyling the small subgraph conditioning method. 

\subsection{Overview}
\label{sec:bitssmall}
The small subgraph conditioning method was introduced by \cite{Worm2} to prove that a random $\Delta$-regular contains asymptotically almost surely (a.a.s.) a Hamilton cycle. Roughly speaking, the method provides a way to get a.a.s results when the second moment method fails,  in the particular case (though common in the random regular setting) where the ratio of the second moment of a variable to the first moment squared converges to a constant strictly greater than 1.

The method was first used to analyze spin models on random regular graphs in \cite{MWW} and was subsequently used in \cite{Sly,GSV:arxiv}. In our setting, applying the small subgraph conditioning method of \cite{Worm2} as in the previous works \cite{MWW,Sly,GSV:arxiv} would not be sufficient, since it guarantees  a polynomial multiplicative deviation from the expectation, which is weak in the setting of Lemma~\ref{lem:gadgetcrucial}. We instead use  an extension of the method given by \cite{Janson}.

More generally, the method of \cite{Worm2} is sufficient when the interest is in proving concentration of a variable within a polynomial factor from its expectation. Janson's refinement of the method gives the distributional limit of the variable and explicitly attributes the fluctuations from the expectation to the presence of specific subgraph structures. For the convenience of the reader, we include both versions of the method in Theorem~\ref{thm:smallgraphmethod}, which is a concatenated version of the respective Theorems in \cite{Worm2,Janson}. The theorem can be extrapolated from~\cite{Janson}, after combining~\cite[Lemma 1, Remark 4, Remark 9]{Janson}. The notation $[X]_{m}$ refers to the $m$-th order falling factorial of the variable $X$. We shall discuss the theorem statement afterwards.

\begin{theorem}\label{thm:smallgraphmethod}
Let $S$ be a set of finite cardinality. For $s\in S$ and $i=1,2,\hdots$, let $\mu_i>0$ and $\delta_i^{(s)}>-1$ be constants  and assume that for each $n$ there are random variables $X_{in}$,  $i=1,2,\hdots,$ and $Y^{(s)}_n$, $s\in S$, all defined on the same probability space $\G=\G_n$ such that $X_{in}$ is non-negative integer valued, $Y^{(s)}_n\geq 0$ and $\E\big[Y_n^{(s)}\big]>0$ (for $n$ sufficiently large). Furthermore, for every $s\in S$,  the following hold:
\begin{enumerate}[label=\emph{(A\arabic{enumi})},ref=(A\arabic{enumi})]
\item \label{it:A1} $X_{in}\stackrel{d}{\longrightarrow}Z_i$ as $n\rightarrow\infty$, jointly for all $i$, where $Z_i\sim\mathrm{Po}(\mu_i)$ are independent Poisson random variables;\label{it:poissonconv}
\item \label{it:A2} for every finite sequence $j_1,\hdots,j_m$ of non-negative integers,
\begin{equation}
\frac{\E_{\G}\big[Y^{(s)}_n[X_{1n}]_{j_1}\cdots [X_{mn}]_{j_m}\big]}{\E_{\G}\big[Y^{(s)}_n\big]}\rightarrow \prod^m_{i=1}\Big(\mu_i\big(1+\delta_i^{(s)}\big)\Big)^{j_i} \quad \text{ as } n\rightarrow\infty;\label{eq:prodsmall}
\end{equation}
\item \label{it:A3} $\sum_i \mu_i\big(\delta^{(s)}_i\big)^2<\infty$;
\item \label{it:A4} $\E_{\G}\big[\big(Y^{(s)}_n\big)^2\big]/\big(\E_{\G}\big[Y_n^{(s)}\big]\big)^2\leq\exp\Big(\sum_i \mu_i \big(\delta^{(s)}_i\big)^2\Big)+o(1)$ as $n\rightarrow \infty$;
\end{enumerate}
Then, the following conclusions hold:
\begin{enumerate}[label=\emph{(C\arabic{enumi})},ref=(C\arabic{enumi})]
\item \label{it:C1} Let $r(n)$ be a function such that $r(n)\rightarrow 0$ as $n\rightarrow \infty$. For each $s\in S$, it holds that $Y^{(s)}_n>r(n)\E_{\G}\big[Y^{(s)}_n\big]$ asymptotically almost surely.
\item \label{it:C2} For $s\in S$,
\begin{equation}\label{eq:dlim}
\frac{Y^{(s)}_n}{\E_{\G}\big[Y^{(s)}_n\big]}\stackrel{d}{\longrightarrow}W^{(s)}=\prod^\infty_{i}\Big[1+\delta_i^{(s)}\Big]^{Z_i}\exp\big(-\mu_i\delta_i^{(s)}\big).
\end{equation}
This and the convergence in \ref{it:A1} hold jointly. The infinite product defining $W^{(s)}$ converges a.s. and in $L^2$, with
\[\E[W^{(s)}]=1\text{ and }   \E\big[(W^{(s)})^2\big]=\lim_{n\rightarrow\infty}\E_{\G}\big[\big(Y^{(s)}_n\big)^2\big]/\big(\E_{\G}\big[Y_n^{(s)}\big]\big)^2.\] Moreover, $W^{(s)}>0$ a.s. iff $\delta^{(s)}_i>-1$ for all $i$;
\end{enumerate}
\end{theorem}

The random variables $Y^{(s)}_n$ in Theorem~\ref{thm:smallgraphmethod} are the ones we are interested in obtaining ``concentration" type results, where $s$ is simply an index allowing us to treat simultaneously more than one variables. In our setting, for $G\sim\Gc^r_n$, $Y^{(s)}_n$ are going to be the variables $Z^\p_G(\eta)$ for  phases $\p\in \Qc$ and configurations $\eta$ on $W$. The random variables $X_{in}$, for graphs with no small multicyclic components, correspond to cycles of length $i$. For example, in our setting and because the graph $G$ is bipartite, $X_{in}$ is the number of cycles of length $i$ in $G$ where $i$ is even.

The conclusion \ref{it:C1} of Theorem~\ref{thm:smallgraphmethod} is essentially due to \cite{Worm2}, while the conclusion \ref{it:C2} is an extension of conclusion \ref{it:C1} due to \cite{Janson}. At this point, to obtain Lemma~\ref{lem:smallgraphstrip} (which was the important part to prove Lemma~\ref{lem:gadgetcrucial}) we will not explicitly use either of \ref{it:C1} or \ref{it:C2} but rather the following variant. The variant was observed in \cite[p.5]{Janson}, who discusses it without proof in a specific setting, and is also implicit in \cite{Worm2}. As such, we write and prove a formal statement in the setup of Theorem~\ref{thm:smallgraphmethod}. The proof follows Janson's proof of Theorem~\ref{thm:smallgraphmethod} but uses a different finish.

\begin{lemma}\label{lem:wellapproximated}
Assume that the conditions in Theorem~\ref{thm:smallgraphmethod} hold. For an integer $m>0$ and $s\in S$, let
\begin{equation*}
W^{(s)}_{mn}=\prod^m_{i=1}\big(1+\delta_i^{(s)}\big)^{X_{in}}\exp\big(-\mu_i\delta_i^{(s)}\big).
\end{equation*}
Then, for every $\epsilon>0$, it holds that
\begin{equation}\label{eq:wellapproximated}
\lim_{m\rightarrow \infty}\limsup_{n\rightarrow \infty} \mathrm{Pr}_{\G_n}\bigg(\bigcup_{s\in S}\bigg[\big|\frac{Y^{(s)}_n}{\E_{\G_n}\big[Y^{(s)}_n\big]}-W^{(s)}_{mn}\big|>\epsilon\bigg]\bigg)=0.
\end{equation}
\end{lemma}

\begin{proof}[Proof of Lemma~\ref{lem:wellapproximated}]
We prove the statement for a fixed $s\in S$, the extension of the argument to prove \eqref{eq:wellapproximated} is straightforward (e.g. by a union bound) and is omitted. To lighten notation we will drop $s$ from the notation and w.l.o.g. we also  assume  $\E_{\G_n}\big[Y_n\big]=1$. We will prove that
\begin{equation}\label{eq:convinterm}
\limsup_{n\rightarrow \infty} \mathrm{Pr}_{\G_n}\big(\big[|Y_n-W_{mn}|>\epsilon\big]\big)\leq\frac{1}{4}\epsilon^{-2}\Big[\exp\big(\sum^{\infty}_{i=1}\mu_i\delta_i^2\big)-\exp\big(\sum^{m}_{i=1}\mu_i\delta_i^2\big)\Big].
\end{equation}
This clearly gives the statement of the lemma, since by assumption \ref{it:A3} of Theorem~\ref{thm:smallgraphmethod}, the lhs is finite and goes to 0 as $m\rightarrow \infty$. To prove \eqref{eq:convinterm}, we follow \cite[Proof of Theorem 1]{Janson} up to a certain point but avoid the use of Skorokhod's theorem in the argument. Janson's proof goes as follows. For a positive integer $m$ define the functions
\begin{align}
f_n(x_1,\hdots,x_m)=\E_{\G_n}[Y_n\,|\, X_{1n}=x_1,\hdots, X_{mn}=x_m], \notag\\ f_{\infty}(x_1,\hdots,x_m)=\lim_{n\rightarrow\infty}f_{n}(x_1,\hdots,x_m)=\prod^{m}_{i=1}(1+\delta_i)^{x_i}e^{-\mu_i\delta_i}.\label{eq:funconvergence}
\end{align}
The second equality follows by assumption \ref{it:A2} of Theorem~\ref{thm:smallgraphmethod} and \cite[Lemma 1]{Janson}. Define also the random variable
\begin{equation*}
Y^{(m)}_n=\E_{\G_n}[Y_n\,|\, X_{1n},\hdots, X_{mn}].
\end{equation*}
Using assumptions \ref{it:A1} and \ref{it:A2}, Fatou's Lemma and that $Y^{(m)}_n$ is a conditional expectation of $Y_n$, one obtains
\begin{equation*}
\limsup_{n\rightarrow\infty}\E_{\G_n}\big[|Y_n-Y_n^{(m)}|^2\big]\leq \exp\big(\sum^{\infty}_{i=1}\mu_i\delta_i^2\big)-\exp\big(\sum^{m}_{i=1}\mu_i\delta_i^2\big),
\end{equation*}
see \cite[Equation (5.2)]{Janson} for details. We now give the main deviation point from Janson's proof, which amounts to proving that for fixed $m$, we have
\begin{equation}
\lim_{n\rightarrow\infty}\mathrm{Pr}_{\G_n}\big(\big[|Y^{(m)}_n-W_{mn}|>\epsilon\big]\big)=0.\label{eq:ywmn}
\end{equation}
as $n\rightarrow\infty$. Fix $M>0$. By \eqref{eq:funconvergence}, there is $N$ such that for $n\geq N$ it holds that
\begin{equation*}
|f_n(x_1,\hdots,x_m)-f_{\infty}(x_1,\hdots,x_m)|<\epsilon \mbox{ for all integer } x_1,\hdots,x_m\in[0,M].
\end{equation*}
It follows that for $n\geq N$, we have
\begin{equation*}
\mathrm{Pr}_{\G_n}\big(\big[|Y^{(m)}_n-W_{mn}|>\epsilon\big]\big)\leq\mathrm{Pr}_{\G_n}\Big(\bigcup^m_{i=1}\big[X_{in}> M\big]\Big)
\end{equation*}
Note that as $n\rightarrow\infty$, the rhs by assumption \ref{it:A1} converges to $\mathrm{Pr}\big(\bigcup^m_{i=1}\big[Z_i> M\big]\big)$. The latter can be made arbitrarily small by letting $M\rightarrow \infty$. This proves \eqref{eq:ywmn}.

The final step is to bound
\begin{align*}
\limsup_{n\rightarrow \infty} \mathrm{Pr}_{\G_n}\big(&\big[|Y_n-W_{mn}|>\epsilon\big]\big)\\
&\leq \limsup_{n\rightarrow\infty}\mathrm{Pr}_{\G_n}\big(\big[|Y_n-Y_n^{(m)}|>\epsilon/2\big]\big)+ \limsup_{n\rightarrow \infty} \mathrm{Pr}_{\G_n}\big(\big[|Y_n^{(m)}-W_{mn}|>\epsilon/2\big]\big)\\
&\leq \frac{1}{4}\epsilon^{-2}\Big[\exp\big(\sum^{\infty}_{i=1}\mu_i\delta_i^2\big)-\exp\big(\sum^{m}_{i=1}\mu_i\delta_i^2\big)\Big]+0,
\end{align*}
which finishes the proof of \eqref{eq:convinterm}.
\end{proof}

\subsection{Application of the Small Subgraph Conditioning Method}\label{sec:applicationsmallgraph}

The application of Theorem~\ref{thm:smallgraphmethod}, and  similarly Lemma~\ref{lem:wellapproximated}, requires a verification of its assumptions. This check is routine for the most part, but it is nevertheless technically arduous, mainly because of assumption \ref{it:A3}, which requires precise calculation of the moments' asymptotics.
We suppress the verification in the following lemma whose proof is given later in this section. The lemma includes some  details on a few  quantities which will be relevant for the proof of Lemma~\ref{lem:smallgraphstrip}.

\begin{lemma}\label{lem:smallcond}
Let $G\sim \G^r_n$ and $X_{in}$ be the number of cycles of even length $i$ appearing in $G$, $i=2,4,\hdots$. Let $S=\{(\p,\eta)\, |\, \p\in \Qc,\ \eta:\,W\rightarrow[q]\}$ and for $s\in S$ with $s=(\p,\eta)$, set $Y^{(s)}_n=Z^{\p}_G(\eta)$. In the setting of Theorem~\ref{thm:general-inapprox}, the assumptions of Theorem~\ref{thm:smallgraphmethod} hold.

Further, for $s\in S$ with $s=(\p,\eta)$ and all even  $i\geq 2$, $\delta^{(s)}_i$ satisfies \begin{enumerate*}[label=\emph{(\roman*)},ref=(\roman*)] \item \label{it:id} $\delta^{(s)}_i>0$, \item \label{it:iid} $\delta^{(s)}_i$ depends on $\p$ but not on $\eta$, \item \label{it:iiid} $\sum_i \mu_i\delta^{(s)}_i<\infty$, \item \label{it:ivd} if the phases are permutation symmetric, $\delta^{(s)}_i$ depends on the spin model but not on the particular phase $\p$. \end{enumerate*}
\end{lemma}

Using Lemmas~\ref{lem:wellapproximated} and~\ref{lem:smallcond}, we are ready to prove Lemma~\ref{lem:smallgraphstrip}.
\begin{proof}[Proof of Lemma~\ref{lem:smallgraphstrip}]
To see~\eqref{eq:concstronger}, note that the $W^{(s)}_{mn}$ of Lemma~\ref{lem:wellapproximated} depend on the particular $s$ only through the $\delta^{(s)}_i$'s. By Item \ref{it:iid} of Lemma~\ref{lem:smallcond}, these depend only on $\p$ in general and specifically for the permutation symmetric case, only on the spin model by Item \ref{it:ivd}.

It remains to prove that $W^{\p}_{mn}$ are lower bounded uniformly in $\p$ by a positive constant. Since the number of phases $\p$ is bounded by a constant depending only on the spin model, it suffices to show that this is the case for a fixed phase $\p$. Using Item \ref{it:id} of Lemma~\ref{lem:smallcond} and that the random variables $X_{in}$ are non-negative integer valued, we have everywhere the bound
\begin{equation*}
W^{\p}_{mn}=\prod^m_{i=1}\big(1+\delta_i^{\p}\big)^{X_{in}}\exp\big(-\mu_i\delta_i^{\p}\big)\geq \prod^m_{i=1}\exp\big(-\mu_i\delta_i^{\p}\big)>\prod^{\infty}_{i=1}\exp\big(-\mu_i\delta_i^{\p}\big).
\end{equation*}
Note that we have identified the $\delta^{(s)}_i$'s with the respective $\delta_i^{\p}$'s, this is justified by Item \ref{it:iid} of Lemma~\ref{lem:smallcond}. The last quantity is finite and positive by Item \ref{it:iiid} in Lemma~\ref{lem:smallcond}.
\end{proof}

We next prove Lemma~\ref{lem:smallcond} which amounts to checking the validity of the assumptions \ref{it:A1}-\ref{it:A4} of Theorem~\ref{thm:smallgraphmethod} for $Z_G^\p(\eta)$ for $\p\in \Qc$ and $\eta:W\rightarrow [q]$.

Let us fix first some notation. Recall that a phase $\p\in \Qc$ corresponds to a global maximum $(\alphab,\betab)$ of $\Psi_1$.  Let $\x=(x_{ij})_{i,j\in[q]}$ be as in Lemma~\ref{helma2}, i.e., the unique vector which maximizes $\Upsilon_1(\alphab,\betab,\X)$ when $\alphab,\betab$ are fixed. In the setting of Theorem~\ref{thm:general-inapprox}, we may assume that $(\alphab,\betab)$ is a Hessian local maximum of $\Psi_1$. The following lemma puts together some relevant quantities and information which we derived in Section~\ref{sec:connection} in the course of proving Theorem~\ref{thm:connection}. 

\begin{lemma}\label{lem:bpeigenspace}
For a random $\Delta$-regular graph, suppose that $(\alphab,\betab)$ is a Hessian local maximum of $\Psi_1$. Define the vector $\x=(x_{ij})_{i,j\in[q]}$ as in Lemma~\ref{helma2}.

Let $\Jb$ be the matrix $\left[\begin{smallmatrix} \zeros & \L\\ \L^{\T} & \zeros\end{smallmatrix}\right]$, where $\L$ is the $q\times q$ matrix whose $ij$-entry is given by $x_{ij}/\sqrt{\alpha_i}\sqrt{\beta_j}$.  Then, the spectrum of $\Jb$ is \[\pm1,\pm\lambda_1,\hdots,\pm\lambda_{q-1},\]
for some positive $\lambda_i$ which satisfy $\max_{i}\lambda_i<\frac{1}{\Delta-1}$. Relevant to Lemma~\ref{lem:smallcond}, observe that if the phases $\p$ are permutation symmetric, then the $\lambda_i$'s are \emph{common for all phases}.
\end{lemma}

Let $G\sim\Gc^r_n$ and $X_i:=X_{in}$ be the number of cycles in $G$ of even length $i$. Let $\p=(\alphab,\betab)\in \Qc$. We next verify the assumptions of Theorem~\ref{thm:smallgraphmethod} for the random variables $Z^{\p}_G(\eta)$, $\eta:W\rightarrow [q]$. We have the following lemmas.

\begin{lemma}[Lemma 7.3 in \cite{MWW}]\label{lem:cycle}
Assumption \ref{it:A1} of Theorem~\ref{thm:smallgraphmethod} holds for even
$i$ with
\[\mu_i=\frac{r(\Delta,i)}{i}=\frac{(\Delta-1)^{i}+(-1)^{i}(\Delta-1)}{i},\]
where $r(\Delta,i)$ is the number of ways to properly edge
color a cycle of length $i$ with $\Delta$ colors.
\end{lemma}
\noindent The proof of Lemma~\ref{lem:cycle} is given in
\cite{MWW} and is omitted.

\begin{lemma}\label{lem:ratiosmallgraph}
Let $\lambda_j$, $j\in[q-1]$ be as in Lemma~\ref{lem:bpeigenspace}. Then,  for all even $i\geq 2$ it holds that
\begin{equation*}
\frac{\E_{\Gc^r_n}[Z^{\p}_G(\eta) X_i]}{\E_{\Gc^r_n}[Z^{\p}_G(\eta)]}\rightarrow \mu_i(1+\delta_i)\mbox{ as } n\rightarrow\infty,
\mbox{ where }\delta_i:=\sum^{q-1}_{j=1}\lambda_j^{i}.
\end{equation*}
In particular, $\delta_i$ is positive for every even $i\geq 2$.
\end{lemma}
\noindent The proof of Lemma~\ref{lem:ratiosmallgraph} is given in
Section~\ref{sec:smallgraphregular}.

\begin{lemma}\label{lem:finitesequence}
Let $\delta_i$, $i=2,4,\hdots$  be as in Lemma~\ref{lem:ratiosmallgraph}. For every finite
sequence $m_1,\hdots,m_k$ of nonnegative integers, it holds that
\begin{equation*}
\frac{\E_{\Gc^r_n}\big[Z^{\p}_G(\eta) [X_2]_{m_1}\cdots [X_{2k}]_{m_k}\big]}{\E_{\Gc^r_n}[Z^{\p}_G(\eta)]}\rightarrow \prod^{k}_{i=1}\big(\mu_i(1+\delta_i)\big)^{m_i}\mbox{ as } n\rightarrow\infty.
\end{equation*}
\end{lemma}
\noindent Once we give the proof of
Lemma~\ref{lem:ratiosmallgraph}, the proof of
Lemma~\ref{lem:finitesequence} is identical to \cite[Proof of Lemma 7.5]{MWW} and is omitted.

\begin{lemma}\label{lem:sumasymptotics}
In the notation and setting of Lemma~\ref{lem:bpeigenspace}, it holds that
\begin{equation*}
\exp\Big(\sum_{\mbox{even }i\geq
2}\mu_i\delta^2_i\Big)=\prod^{q-1}_{i=1}\prod^{q-1}_{j=1}\big(1-(\Delta-1)^2\lambda^2_i\lambda^2_j\big)^{-1/2}\prod^{q-1}_{i=1}\prod^{q-1}_{j=1}\big(1-\lambda^2_i\lambda^2_j\big)^{-(\Delta-1)/2}.
\end{equation*}
Moreover, $\sum_{\mbox{\small{even }}i\geq 2}\,\mu_i\delta_i<\infty$.
\end{lemma}
\noindent The proof of Lemma~\ref{lem:sumasymptotics} is given in
Section~\ref{sec:smallgraphregular}.

Finally, we find the asymptotics of the second moment over the first moment squared.
\begin{lemma}\label{lem:momentsratio}
In the notation and setting of Lemma~\ref{lem:bpeigenspace}, it holds that
\begin{equation*}
\lim_{n\rightarrow\infty}\frac{\E_{\G^r_n}[(Z^{\p}_G(\eta))^2]}{\big(\E_{\G^r_n}[Z^{\p}_G(\eta)]\big)^2}=C, \mbox{ where } C:=\prod^{q-1}_{i=1}\prod^{q-1}_{j=1}\big(1-(\Delta-1)^2\lambda^2_i\lambda^2_j\big)^{-1/2}\prod^{q-1}_{i=1}\prod^{q-1}_{j=1}\big(1-\lambda^2_i\lambda^2_j\big)^{-(\Delta-1)/2}.
\end{equation*}
\end{lemma}
The proof of Lemma~\ref{lem:momentsratio} is quite extensive. In Appendix~\ref{sec:momentasymptotics}, we first reduce the lemma to the case $r=0$. This part of the proof is standard and is analogous to the proof of Lemma~\ref{lem:itemexpect}. Then, we compute the asymptotics in terms of determinants of relevant Hessian matrices. These determinants are computed in Appendix~\ref{sec:computations}, where also the proof of Lemma~\ref{lem:momentsratio} is completed.

With Lemmas~\ref{lem:cycle}---\ref{lem:momentsratio} at hand, the proof of Lemma~\ref{lem:smallcond} is immediate.
\begin{proof}[Proof of Lemma~\ref{lem:smallcond}]
We prove the first part of the lemma by verifying the assumptions of Theorem~\ref{thm:smallgraphmethod}. Lemma~\ref{lem:cycle} verifies assumption \ref{it:A1}, Lemma~\ref{lem:finitesequence} verifies assumption \ref{it:A2} and Lemmas~\ref{lem:sumasymptotics} and \ref{lem:momentsratio} verify assumptions \ref{it:A3} and \ref{it:A4}. This proves the first part of the lemma.

For the second part, just use the second parts in Lemmas~\ref{lem:bpeigenspace},~\ref{lem:ratiosmallgraph},~\ref{lem:sumasymptotics} to establish Items \ref{it:id}---\ref{it:ivd}.
\end{proof}

\subsection{Proofs of Lemmas~\ref{lem:ratiosmallgraph}
and~\ref{lem:sumasymptotics}}\label{sec:smallgraphregular}

In this section, we give the proofs of Lemmas~\ref{lem:ratiosmallgraph}
and~\ref{lem:sumasymptotics}.

\begin{proof}[Proof of Lemma~\ref{lem:ratiosmallgraph}]
The proof is close to \cite[Proof of Lemma 7.4]{MWW}, the approach only  needs a few modifications to account for the $q$-spin setting. We make the minor notation change from $X_i$ to $X_\ell$.

We will do the computations for the case of $G\sim\Gc_n$ and the random variables $Z^{\alphab,\betab}_G$, the extension to the case  $G\sim\Gc_n^r$ and the random variables $Z^{\p}_G(\eta)$ has slightly more complicated expressions, but otherwise the derivation is completely analogous (see for example \cite[Proof of Lemma 3.8]{Sly}). We will show that
\[\frac{\E_{\Gc_n}[Z^{\alphab,\betab}_G X_\ell]}{\E_{\Gc_n}[Z^{\alphab,\betab}_G]}\rightarrow \mu_\ell(1+\delta_\ell)\mbox{ as } n\rightarrow \infty,\]
where $\mu_\ell,\delta_\ell$ are as in the statement of the lemma. For simplicity, let $\G:=\Gc_n$.

Let $\Sc=\{S_1,\hdots,S_q\}$ and $\Tc=\{T_1,\hdots,T_q\}$ be
partitions of $V_1$ and $V_2$ respectively such that
$|S_i|=\alpha_i n$ and $|T_j|=\beta_j n$ for all $i,j\in [q]$.
Denote by $Y_{\Sc,\Tc}$ the weight of the configuration $\sigma$
that $\Sc,\Tc$ induce, i.e. for a vertex $v\in V_1$, $\sigma(v)=i$
iff $v\in S_i$ and similarly for vertices in $V_2$.

Fix a specific pair of $\Sc,\Tc$. By symmetry,
\begin{equation}
\frac{\E_{\Gc}[Z^{\alphab,\betab}_GX_\ell]}{\E_{\Gc}[Z^{\alphab,\betab}_G]}=\frac{\E_{\Gc}[Y_{\Sc,\Tc}
X_\ell]}{\E_{\Gc}[Y_{\Sc,\Tc}]}. \label{eq:smallratio}
\end{equation}
We now decompose $X_\ell$ as follows:
\begin{itemize}
\item $\xi$ will denote a proper $\Delta$-edge colored, rooted and
oriented $\ell$-cycle ($r(\Delta,\ell)$ possibilities), in which
the vertices are colored with $\{Y_1,\hdots,Y_q,G_1,\hdots,G_q\}$
and edges are colored with $\{1,\hdots,\Delta\}$.

A vertex colored with $Y_i$ (resp. $G_i$) for some $i\in [q]$ will
be loosely called yellow (resp. green) and signifies that the
vertex belongs to $S_i$ (resp. $T_i$). Since a yellow vertex
belongs to $V_1$, and a green vertex belongs to $V_2$, a vertex
coloring is consistent with the bipartiteness of the random graph
if adjacent vertices of the cycle are not  both yellow or green,
that is, the vertex assignments which are prohibited for neighboring vertices in the cycle
are $(Y_i,Y_j)$ and $(G_i,G_j)$, $\forall(i,j)\in [q]^2$. Note here that
we do not expicitly prohibit assignments $(Y_i,G_j)$ in the presence of
a hard constraint $B_{ij}=0$; this will be accounted otherwise.
The color of the edges will prescribe which of the $\Delta$ perfect
matchings an edge of a (potential) cycle will belong to.

\item Given $\xi$, $\zeta$ denotes a position that an $\ell$-cycle can be, i.e., the exact vertices it traverses in order,
such that the prescription of the vertex colors of $\xi$ is satisfied.

\item $\mathbf{1}_{\xi,\zeta}$ is the indicator function whether a cycle
specified by $\xi,\zeta$ is present in the graph $G$.
\end{itemize}
Note that each possible cycle corresponds to exactly $2\ell$
different configurations $\xi$ (the number of ways to root and
orient the cycle). For each of those $\xi$, the respective sets of
configurations $\zeta$ are the same. Hence, we may write
\[X_\ell=\frac{1}{2\ell}\sum_\xi\sum_\zeta \mathbf{1}_{\xi,\zeta}.\]
Let $p_1:= \Pr_{\Gc}[\mathbf{1}_{\xi,\zeta}=1]$. We have
\begin{align*}
\E_{\Gc}[Y_{\Sc,\Tc}X_\ell]&=\frac{1}{2\ell}\sum_\xi\sum_\zeta
p_1\cdot\E[Y_{\Sc,\Tc}|\mathbf{1}_{\xi,\zeta}=1].
\end{align*}
In light of \eqref{eq:smallratio}, we need to study the ratio
$\E_{\Gc}[Y_{\Sc,\Tc}|\mathbf{1}_{\xi,\zeta}=1]/\E_{\Gc}[Y_{\Sc,\Tc}]$. At
this point, to simplify notation, we may assume that $\xi,\zeta$
are fixed.

We have shown in Section~\ref{sec:derivations} that
\begin{equation}\label{eq:ysimple} \E_{\Gc}[Y_{\Sc,\Tc}]=\bigg(\sum_{\X}\binom{n}{x_{11}n,\hdots,x_{qq}n}^{-1}\prod_{i}\binom{\alpha_{i}n}{x_{i1}n,\hdots,x_{iq}n}\prod_{j}\binom{\beta_{j}n}{x_{1j}n,\hdots,x_{qj}n}\prod_{i,j}B^{nx_{ij}}_{ij}\bigg)^\Delta,
\end{equation}
where the variables $\X=(x_{11},\hdots,x_{qq})$ denote the number
of edges between $\Sc,\Tc$ in one matching. In particular
$nx_{ij}$ is the number of edges between the sets $S_i$ and $T_j$.

To calculate $\E_{\Gc}[Y_{\Sc,\Tc}|\mathbf{1}_{\xi,\zeta}=1]$, we need some
notation. For colors
$c_1,c_2\in\{Y_1,\hdots,Y_q,G_1,\hdots,G_q\}$, we say that an edge
is of type $\{c_1,c_2\}$ if its endpoints have colors $c_1,c_2$.
Let $y_i,g_j$ denote the number of vertices colored with $Y_i,G_j$
respectively. For $k = 1,\hdots,\Delta$, let $a_{ij}(k)$ denote
the number of edges of color $k$ and type $\{Y_i,G_j\}$. Finally,
for $i,j\in[q]$ let $a_{ij}=\sum_k a_{ij}(k)$. By
considering the sum of the degrees of vertices colored $Y_i$, the
sum of the degrees of vertices colored $G_j$ and the total number
of edges of the cycle, we obtain the following equalities.
\begin{equation}\label{eq:degreeconsideration}
\mbox{$\sum_{j}$}\,a_{ij}=2y_i,\ \mbox{$\sum_{i}$}\,a_{ij}=2g_j,\
\mbox{$\sum_{i,j}$}\, a_{ij}=2\ell.
\end{equation}
We are almost set to compute
$\E[Y_{\Sc,\Tc}|\mathbf{1}_{\xi,\zeta}=1]$. We denote by $\X_k$ the
same set of variables as in \eqref{eq:ysimple} but for the $k$-th
matching. Namely, $nx_{ij,k}$ is the number of (undetermined)
edges between sets $S_i$ and $T_j$ in the $k$-th matching. This
number includes the $a_{ij}(k)$ edges prescribed by $\xi,\zeta$.
To simplify the following fomulas,  let
$nx'_{ij,k}=nx_{ij,k}-a_{ij}(k)$ and set $E=\E_{\Gc}[Y_{\Sc,\Tc}|\mathbf{1}_{\xi,\zeta}=1]$.
We have
\begin{equation*}
E=\prod^\Delta_{k=1}\bigg[\sum_{\X_k}\binom{n-\mbox{$\sum_{i,j}$}\,a_{ij}(k)}{x'_{11,k}n,\hdots,x'_{qq,k}n}^{-1}\prod_{i}\binom{\alpha_{i}n-\mbox{$\sum_{j}$}\,a_{ij}(k)}{nx'_{i1,k},\hdots,nx'_{iq,k}}\prod_{j}\binom{\beta_{j}n-\mbox{$\sum_{i}$}\,a_{ij}(k)}{nx'_{1j,k},\hdots,nx'_{qj,k}}\prod_{i,j}B^{nx_{ij,k}}_{ij}\bigg].
\end{equation*}
In the above sums, for any $\epsilon>0$ and all sufficiently large $n$, terms whose $\x_k$'s are $\epsilon$-far from the optimal value of $\X$ given in Lemma~\ref{helma2} have exponentially small contribution and may be ignored. Standard approximations of binomial coefficients, see for example \cite[Lemma 27]{GSV:arxiv}, give
\begin{equation*}
\frac
{\binom{\alpha_{i}n-\Sigma_{j}\,a_{ij}(k)}{x'_{i1,k}n,\hdots,x'_{iq,k}n}}
{\binom{\alpha_{i}n}{x_{i1,k}n,\hdots,x_{iq,k}n}} \sim \frac
{\prod_j\big(x_{ij,k}\big)^{a_{ij}(k)}}
{\alpha_i^{\mbox{$\sum_{j}$}\,a_{ij}(k)}},\ \frac
{\binom{\beta_{j}n-\mbox{$\sum_{i}$}\,a_{ij}(k)}{x'_{1j,k}n,\hdots,x'_{qj,k}n}}
{\binom{\beta_{j}n-\mbox{$\sum_{i}$}\,a_{ij}(k)}{x'_{1j,k}n,\hdots,x'_{qj,k}n}}
\sim \frac {\prod_i\big(x_{ij,k}\big)^{a_{ij}(k)}}
{\beta_j^{\mbox{$\sum_{i}$}\,a_{ij}(k)}}, \ \frac
{\binom{n-\mbox{$\sum_{i,j}$}\,a_{ij}(k)}{x'_{11,k}n,\hdots,x'_{qq,k}n}}
{\binom{n}{x_{11,k}n,\hdots,x_{qq,k}n}} \sim
\prod_{i,j}\big(x_{ij,k}\big)^{a_{ij}(k)}.
\end{equation*}
Thus, we obtain
\begin{equation*}
\frac{\E_{\Gc}[Y_{\Sc,\Tc}|\mathbf{1}_{\xi,\zeta}=1]}{\E_{\Gc}[Y_{\Sc,\Tc}]}
\sim \frac {\prod_{i,j}\big(x_{ij}\big)^{a_{ij}}}
{\prod_i\alpha_i^{\mbox{$\sum_{j}$}\,a_{ij}}\prod_j\beta_j^{\mbox{$\sum_{i}$}\,a_{ij}}}.
\end{equation*}
We have $p_1\sim n^{-\ell}$ and for given $\xi$, the number of
possible $\zeta$ is asymptotic to $n^\ell\prod_i\alpha^{y_i}_i\prod_j\beta^{g_j}$. Thus, for the
given $\xi$, we have
\begin{equation*}
\frac{\sum_\zeta
p_1\E_{\Gc}[Y_{\Sc,\Tc}|\mathbf{1}_{\xi,\zeta}=1]}{\E_{\Gc}[Y_{\Sc,\Tc}]}\sim
\frac
{\prod_i\alpha^{y_i}_i\prod_j\beta^{g_j}_j\prod_{i,j}\big(x_{ij}\big)^{a_{ij}}}
{\prod_i\alpha_i^{\mbox{$\sum_{j}$}\,a_{ij}}\prod_j\beta_j^{\mbox{$\sum_{i}$}\,a_{ij}}}
=
\prod_{i,j}\Big(\frac{x_{ij}}{\sqrt{\alpha_i\beta_j}}\Big)^{a_{ij}}
\end{equation*}
Note that the rhs evaluates to 0 whenever there exist $i,j$ such that $B_{ij}=0$ but $a_{ij}\neq 0$, since then we have $x_{ij}=0$. This is in complete accordance with the fact that the configuration induced by the partition $\{\Sc,\Tc\}$ has zero weight. Thus, by \eqref{eq:smallratio}, we have
\begin{equation*}
\frac{\E_{\Gc}[Z^{\alphab,\betab}_GX_\ell]}{\E_{\Gc}[Z^{\alphab,\betab}_G]}\sim\frac{r(\Delta,\ell)}{2\ell}\cdot
\sum_\xi N_\a
\Big(\frac{x_{ij}}{\sqrt{\alpha_i\beta_j}}\Big)^{a_{ij}},
\end{equation*}
where $\a=\{a_{11},\hdots,a_{qq}\}$ and $N_\a$ is the number of
possible $\xi$ with $a_{ij}$ edges having assignment $(Y_i,G_j)$. To
analyze this sum, we employ a technique given in \cite{Janson}.
The idea is to define a weighted transition matrix and view it as the (weighted) adjacency matrix of 
a weighted graph. The powers of the matrix count the
(multiplicative) weight of walks in the graph and a closed walk in
this graph will correspond to a specification $\xi$. By defining
the weights appropriately, one can also ensure that each closed
walk will correctly capture the weight of the specification $\xi$.

In our setting, the transition matrix is simply the matrix $\Jb$ of Lemma~\ref{lem:bpeigenspace}.
The first $q$ rows and $q$ columns correspond to
the colors $Y_i$ and the remaining rows and columns to
colors $G_j$. The total weight of closed walks of length $\ell$ is given by
$\mathrm{Tr}(\Jb^\ell)$. Using the description of the eigenvalues given in  Lemma~\ref{lem:bpeigenspace}, we obtain that for even $\ell$, $\mathrm{Tr}(\Jb^\ell)=2\Big(1+\sum^{q-1}_{i=1}\lambda^\ell_i\Big)$.
This concludes the proof.
\end{proof}

\begin{proof}[Proof of Lemma~\ref{lem:sumasymptotics}]
Using Lemma~\ref{lem:cycle}, we have
\begin{equation*}
\sum_{\mbox{\small{even} }i\geq
2}\mu_i\delta^2_i=\sum_{\mbox{\small{even} }i\geq
2}\frac{r(\Delta,i)}{i}\cdot
\bigg(\sum^{q-1}_{j=1}\lambda^i_j\bigg)^2=\sum_{\mbox{\small{even}
}i\geq
2}\frac{(\Delta-1)^i+(\Delta-1)}{i}\cdot\bigg(\sum^{q-1}_{j=1}\sum^{q-1}_{j'=1}\lambda^i_j\lambda^i_{j'}\bigg).
\end{equation*}
Observe that $\sum_{j\geq
1}\frac{x^{2j}}{2j}=-\frac{1}{2}\ln(1-x^2)$ for all $|x|<1$. By Lemma~\ref{lem:bpeigenspace},
the $\lambda_j$'s  satisfy
$(\Delta-1)\lambda_j< 1$ for all $j$, so that
$(\Delta-1)\lambda_j\lambda_{j'}<1$ for all $j,j'$. It follows that
\[\sum_{\mbox{\small{even} }i\geq 2}\mu_i\delta^2_i=-\frac{1}{2}\Big(\sum_{i,j}\ln\big(1-(\Delta-1)^2\lambda^2_i\lambda^2_j\big)+(\Delta-1)\sum_{i,j}\ln\big(1-\lambda^2_i\lambda^2_j\big)\Big),\]
thus proving the first part of the lemma. The proof of $\mbox{$\sum_{i}$}\,\mu_i\delta_i<\infty$ is completely analogous.
\end{proof}

\section{Moment Asymptotics}
\label{sec:momentasymptotics}
In this section, we prove Lemma~\ref{lem:momentsratio}. For the purposes of this section, we will identify $\Qc$ with the dominant phases of a random $\Delta$-regular bipartite graph. Thus, we will use $\p\in \Qc$ to denote a dominant phase $(\alphab,\betab)$. 

We first recall  some relevant definitions from Section~\ref{sec:slysunstuff}. For $r\geq 0$, let $G\sim \Gc_n^r$ and $\sigma:U\cup W\rightarrow [q]$ be a configuration on $G$.  The \emph{footprint}  of $\sigma$ is a pair of $q$-dimensional vectors $\alphab_\sigma,\betab_\sigma$ whose $i$-th entries are equal to $|\sigma^{-1}\cap U^+|/n, |\sigma^{-1}\cap U^-|/n$, respectively. The \emph{phase} $Y(\sigma)$ of $\sigma$ is the dominant phase $(\alphab,\betab)$ which is closest to $(\alphab_\sigma,\betab_\sigma)$, precisely:
\begin{equation}\label{eq:phaseofconfiguration}
Y(\sigma)=\arg\max_{\p=(\alphab,\betab)\in\Qc}\big(\norm{\alphab-\alphab_\sigma}^2_2+\norm{\betab-\betab_\sigma}^2_2\big)^{1/2}
\end{equation}
Finally, recall that for $\p\in \Qc$, $Z^\p_G$ is the partition function ``conditioned on the phase $\p$'', i.e, the contribution to the partition function of $G$ from configurations $\sigma$ with $Y(\sigma)=\p$, and for $\eta:W\rightarrow[q]$, $Z^\p_G(\eta)$ is the contribution to the partition function of $G$ from configurations $\sigma$ with $Y(\sigma)=\p$ and $\sigma_W=\eta$, see \eqref{eq:defcondpart} for more details. 

In the setting of Lemma~\ref{lem:momentsratio}, we need to compute the asymptotics of $\E_{\G^r_n}[(Z^{\p}_G(\eta))^2]/(\E_{\G^r_n}[Z^{\p}_G(\eta)])^2$ for $\p\in \Qc$ and configurations $\eta:W\rightarrow[q]$. The following lemma  reduces the computation to the case $r=0$. Note that for $r=0$, the set of vertices $W$ is empty and the distribution $\G^r_n$ coincides with the distribution $\G:=\G_n$ on random $\Delta$-regular bipartite graphs from Section~\ref{sec:derivations}.

\begin{lemma}\label{lem:redrzero}
Let $\p=(\alphab,\betab)\in \Qc$ be a Hessian dominant phase. Then, for every fixed $r>0$, for every $\eta:W\rightarrow [q]$ it holds that
\[\lim_{n\rightarrow\infty}\frac{\E_{\G^r_n}[(Z^{\p}_G(\eta)\big)^2]}{(\E_{\G^r_n}[Z^{\p}_G(\eta)])^2}=\lim_{n\rightarrow\infty}\frac{\E_{\G_n}[(Z^{\p}_G)^2]}{(\E_{\G_n}[Z^{\p}_G])^2}.\]
\end{lemma}
\begin{proof}
By Lemma~\ref{lem:itemexpect}, we have
\begin{equation*}\tag{\ref{eq:itemexpect2}}
\E_{\Gc^r_n}\big[Z^{\p}_G(\eta)\big]=\big(1+o(1)\big)C^r\nu^{\otimes}_{\p}(\eta)\E_{\Gc_n}\big[Z^{\p}_G\big],
\end{equation*}
where $C(\p)$ is the constant in Lemma~\ref{lem:itemexpect} and $\nu^{\otimes}_{\p}(\eta)$ is defined in \eqref{eq:nudefinition}.  As in (the proof of) Lemma~\ref{lem:itemexpect}, we also obtain
\begin{equation}\label{eq:secasd}
\E_{\Gc^r_n}\big[(Z^{\p}_G(\eta))^2\big]=\big(1+o(1)\big)C^{2r}(\nu^{\otimes}_{\p}(\eta))^2\E_{\Gc_n}\big[(Z^{\p}_G)^2\big], 
\end{equation}
where $C(\p)$ is again the constant in Lemma~\ref{lem:itemexpect}. Combining \eqref{eq:itemexpect2} and \eqref{eq:secasd} proves the lemma.
\end{proof}

In light of Lemma~\ref{lem:redrzero}, we need to compute the limiting ratio of $\E_{\G_n}[(Z^{\p}_G)^2]/(\E_{\G_n}[Z^{\p}_G])^2$ for $\p\in \Qc$. We do this by computing separately the asymptotics of $\E_{\G_n}[Z^{\p}_G]$ and $\E_{\G_n}[(Z^{\p}_G)^2]$. We begin with an observation that will allow us to deduce the asymptotics of $\E_{\G}[(Z^{\p}_G)^2]$ from the asymptotics of $\E_{\G}[Z^{\p}_G]$ applied to the spin system with interaction matrix $\B\otimes\B$. 
\begin{lemma}\label{lem:secsecsec}
Let $\p=(\alphab,\betab)\in \Qc$ be a dominant phase for the spin system with interaction matrix $\B$. Then $\p'=(\alphab\otimes \alphab,\betab\otimes \betab)$ is a dominant phase for the spin system with interaction matrix $\B\otimes\B$. 

Let  $Z^{\p'}_G$ equal the partition function for  the spin system with interaction matrix $\B\otimes \B$ conditioned on the phase $\p'$. Then, $\lim_{n\rightarrow\infty}\frac{\E_{\G}[(Z^{\p}_G)^2]}{\E_{\G}[Z^{\p'}_G]}=1$.
\end{lemma}
\begin{proof}[Proof of Lemma~\ref{lem:secsecsec}]
The first part of the lemma is an immediate consequence of Lemma~\ref{lem:maxphi2}. In fact, the proof of Lemma~\ref{lem:maxphi2} shows the stronger fact that the set of dominant phases for the spin system with interaction matrix $\B\otimes\B$ is given by $\Qc^{\otimes 2}:=\{(\alphab\otimes \alphab',\betab\otimes\betab')\mid (\alphab,\betab),(\alphab',\betab')\in \Qc\}$. 

For the second part, let 
\begin{equation}\label{eq:sigma1p}
\Sigma_1=\Big\{(\alphab',\betab')\mid \alphab',\betab'\in \triangle_q,\, \p=\arg \min_{\p^*=(\alphab^*,\betab^*)\in \Qc} (\norm{\alphab'-\alphab^*}^2+\norm{\betab'-\betab^*}^2)^{1/2}\Big\},
\end{equation}
so that
\[Z^\p_G=\sum_{(\alphab',\betab')\in \Sigma_1}Z_G^{\alphab',\betab'}\mbox{ and }(Z^\p_G)^2=\sum_{(\alphab',\betab'),(\alphab'',\betab'')\in \Sigma_1}Z_G^{\alphab',\betab'}Z_G^{\alphab'',\betab''}.\]
It follows that $\E_{\G}[(Z^{\p}_G)^2]$ is given by the sum in the  r.h.s. in \eqref{eq:secondmoment}, but now the sum is over $\gammab,\deltab$ which satisfy 
\begin{equation*}
\begin{aligned}
\mbox{$\sum_{k}$}\,\gamma_{ik}&=\alpha_i'   & &\big(\forall i\in [q]\big),& \mbox{$\sum_{l}$}\,\delta_{jl}&=\beta_j'& &\big(\forall j\in [q]\big),\\
\mbox{$\sum_{i}$}\,\gamma_{ik}&=\alpha_k''   & &\big(\forall k\in [q]\big),& \mbox{$\sum_{j}$}\,\delta_{jl}&=\beta_l''& &\big(\forall l\in [q]\big),
\end{aligned}
\end{equation*}
and $(\alphab',\betab'),(\alphab'',\betab'')$ range over $\Sigma_1$. Note that by the definition of $\Sigma_1$, $\p=(\alphab,\betab)$ is the unique dominant phase (of the spin system with interaction matrix $\B$) contained in $\Sigma_1$. By Lemma~\ref{lem:maxphi2}, for any $\epsilon>0$ and all sufficiently large $n$, terms in the sum with $(\norm{\gammab-\alphab\otimes \alphab}^2_2+\norm{\deltab-\betab\otimes \betab}_2^2)^{1/2}\geq \epsilon$ have exponentially small contribution and hence may be ignored. Similarly, for the spin system with interaction matrix $\B\otimes \B$, $\E_{\G}[Z^{\p'}_G]$ is given by the sum in the  r.h.s. in \eqref{eq:secondmoment}, where now the sum is over $\gammab,\deltab$ with $(\gammab,\deltab)\in \Sigma_2$, where
\[\Sigma_2:=\Big\{(\gammab',\deltab')\mid \gammab',\deltab'\in \triangle_{q^2},\, \p'=\arg \min_{\p^*=(\gammab^*,\deltab^*)\in \Qc^{\otimes 2}} (\norm{\gammab'-\gammab^*}^2+\norm{\deltab'-\deltab^*}^2)^{1/2}\Big\}.\] 
Once again, by Lemma~\ref{lem:maxphi2}, for any $\epsilon>0$ and all sufficiently large $n$, terms in the sum with $(\norm{\gammab-\alphab\otimes\alphab}^2_2+\norm{\deltab-\betab\otimes\betab}^2_2)^{1/2}\geq \epsilon$ have exponentially small contribution and hence may be ignored. It follows that for all sufficiently small $\epsilon>0$, the remaining terms in the two sums are identical which completes the proof.
\end{proof}

As a consequence of Lemma~\ref{lem:secsecsec}, we may focus on the asymptotics of the  first moment $\E_{\G}[Z^{\p}_G]$. Let
\begin{equation}
P_1=\big\{(i,j)\in[q]^2\,\big|\, B_{ij}> 0\big\}.\label{eq:defpone}
\end{equation}
In the presence of a hard constraint $B_{ij}=0$, edge assignments $(i,j)$ yield a zero-weight configuration. In the maximization of $\Upsilon_1$, the hard constraint $B_{ij}=0$ was not directly relevant,  since for $x_{ij}>0$  the function $\Upsilon_1$ evaluates to $-\infty$. Indeed, we found that the optimal $x_{ij}$ is of the form $B_{ij}R_iC_j$ and hence zero. However, the asymptotics of $\E_{\G}[Z^{\p}_G]$ include products of the optimal values of the $x_{ij}$ and to correctly capture them, we need to explicitly rule out the zero values.

To do so, in the formulation \eqref{eq:constraintfirst}, we hard-code $x_{ij}=0$ for a pair $(i,j)\notin P_1$ and hence the variables $\alphab,\betab,\X$ are restricted to the space
\begin{equation}\label{eq:spacefirst}
\begin{gathered}
\begin{aligned}
\mbox{$\sum_{i}$}\,\alpha_i&=1,& & & \mbox{$\sum_{j}$}\,\beta_j&=1,& &\\
\mbox{$\sum_{j}$}\,x_{ij}&=\alpha_i& & \big(\forall i\in[q]\big),& \mbox{$\sum_{i}$}\,x_{ij}&=\beta_j& &\big(\forall j\in[q]\big),\\
x_{ij}&=0& &\big(\forall(i,j)\in [q]^2\backslash P_1\big),& x_{ij}&\geq 0& &\big(\forall(i,j)\in P_1\big).\\
\end{aligned}
\end{gathered}
\end{equation}
We will also need to have a set of affinely independent variables which describe the polytope \eqref{eq:spacefirst}. Note that the dimension of the polytope \eqref{eq:spacefirst} is $(q^2+2q)-(2q+1)-(q^2-|P_1|)=|P_1|-1$. To get affinely independent variables $\alphab,\betab,\X$, we use the equalities in \eqref{eq:spacefirst} and substitute an appropriate set of $(q+1)^2-|P_1|$ variables. We will not need to understand these substitutions till Appendix~\ref{sec:hesformulations}, yet in the integrations which follow it is preferable to have  integration variables rather than integrate over subspaces.

After this process, we are going to have $|P_1|-1$ variables lying in a full dimensional space.  We refer to this set of variables as the full dimensional representation of \eqref{eq:spacefirst}. For simplicity, we will still use $\alphab,\betab,\X$ for these variables and refer, e.g., to $x_{ij}$ even if $x_{ij}$ is not in the full dimensional representation of \eqref{eq:spacefirst}, under the understanding that this is just a shorthand for the substituted expression. Using these conventions, we may view $\Upsilon_1(\alphab,\betab,\X)$ as a function of the full dimensional representation of \eqref{eq:spacefirst}, and we will refer to this setup as the full dimensional representation of $\Upsilon_1$.

The following lemma expresses the asymptotics of $\E_{\G}[Z^{\p}_G]$ in terms of suitable determinants. The computation of these determinants is given in Section~\ref{sec:computations}, where also the proof of Lemma~\ref{lem:momentsratio} is completed.

\begin{lemma}\label{lem:firasympdet}
Let $\p=(\alphab^*,\betab^*)$ be a dominant phase, i.e., $(\alphab^*,\betab^*)$ maximizes $\Psi_1(\alphab,\betab)$. Let $\x^*$ be the (unique) maximizer of $\Upsilon_1(\alphab^*,\betab^*,\X)$ (given in Lemma~\ref{helma2}). Denote by $\H_1^f$ be the Hessian of the full dimensional representation
of $\Upsilon_1(\alphab,\betab,\X)$ scaled by $1/\Delta$ (evaluated at $\alphab^*,\betab^*,\X^*$) and by
$\H_{1,\X}^f$ the square submatrix of $\H_1^f$ corresponding to rows
and columns indexed by $\X$. Then
\[\lim_{n\rightarrow\infty}\frac{\E_{\G}[Z^{\p}_G]}{e^{n\Upsilon_1(\alphab^*,\betab^*,\X^*)}}=
\frac{\Big(     \prod_{i} \alpha_{i}^{*}   \prod_{j} \beta_{j}^{*}
\Big)^{(\Delta-1)/2}\Big(     \prod_{(i,j)\in P_1}x_{ij}^{*}
\Big)^{-\Delta/2}}
{{\Delta^{q-1} \big(\Det(-\H_1^f)\big)^{1/2}
\big(\Det(-\H_{1,\X}^f)\big)^{(\Delta-1)/2}}}.\]
\end{lemma}

%
%
\begin{proof}[Proof of Lemma~\ref{lem:firasympdet}]
We assume that $\alphab,\betab,\X$ is a full dimensional representation of \eqref{eq:spacefirst}. We denote by $(\alphab^*,\betab^*,\X^*)$ the optimal vector which maximizes the full dimensional representation of $\Upsilon_1(\alphab,\betab,\X)$. We have that $\alpha^*_{i},\beta^*_{j}>0$ for all $i,j$ and $x^*_{ij}>0$ for $(i,j)\in P_1$. Pick $\delta$ sufficiently small such that:
\begin{equation*}
\norm{(\alphab,\betab,\X)-(\alphab^*,\betab^*,\X^*)}_2\leq \delta\text{ implies }\alpha_{i},\beta_{j}>0\text{ for all }i,j\text{ and }x_{ij}>0\text{ for }(i,j)\in P_1.
\end{equation*}

Since $\Upsilon_1$ has the unique global maximum $(\alphab^*,\betab^*,\X^*)$ at the intersection of the spaces \eqref{eq:sigma1p} and \eqref{eq:spacefirst}, standard compactness arguments imply that there exists $\epsilon(\delta)>0$ such that \linebreak $\norm{(\alphab,\betab,\X)-(\alphab^*,\betab^*,\X^*)}_2\nobreak\geq \delta$ implies $\Upsilon_1(\alphab^*,\betab^*,\X^*)-\Upsilon_1(\alphab,\betab,\X)\geq\epsilon$. It follows that the contribution of terms with $\norm{(\alphab,\betab,\X)-(\alphab^*,\betab^*,\X^*)}_2\geq \delta$ to $\E_{\G}[(Z^{\p}_G)^2]$ is exponentially small and may be ignored. Hence we may restrict our attention to $\alphab,\betab,\X$ satisfying $\norm{(\alphab,\betab,\X)-(\alphab^*,\betab^*,\X^*)}_2< \delta$. Moreover, using Taylor's expansion, we may choose $\delta$ small enough such that $\Upsilon_1$ decays quadratically in a $\delta$-ball around $(\alphab^*,\betab^*,\X^*)$.

Utilizing the choice of $\delta$ and Stirling's approximation for factorials, we thus obtain
\begin{align*}
\frac{\E_{\G}[Z^{\p}_G]}{e^{n\Upsilon_1(\alphab^*,\betab^*,\X^{*})}}
&=\Big(1+O\big(n^{-1}\big)\Big)
\sum_{\alphab,\betab}
\Big(    \frac{1}{\sqrt{2\pi n}}     \Big)^{2(q-1)}
\Big(    \prod_{i} \alpha_{i}   \prod_{j} \beta_{j}      \Big)^{(\Delta-1)/2}
\notag\\
&\bigg[
\sum_\X    \Big(   \frac{1}{\sqrt{2\pi n}}     \Big)^{|P_1|-(2q-1)}
\Big(   \prod_{(i,j) \in P_1}  \frac{1}{\sqrt{x_{ij}}}     \Big)
e^{n    \big(    \Upsilon_1(\alphab,\betab,\X)-\Upsilon_1(\alphab^*,\betab^*,\X^*)     \big)     /\Delta}
\bigg]^{\Delta}.
\end{align*}
We now compute
\[L:=\lim\limits_{n\rightarrow\infty}
\frac{\E_{\G}[Z^{\p}_G]}{e^{n\Upsilon_1(\alphab^*,\betab^*,\X^{*})}}.
\]
Standard techniques of rewriting sums as integrals and an application of the dominated convergence theorem (see for example \cite[Section 9.4]{JLR}) ultimately give
\begin{align}
L&=
\Big(     \prod_{i} \alpha_{i}^*   \prod_{j} \beta_{j}^*      \Big)^{(\Delta-1)/2}
\Big(     \prod_{(i,j)\in P_1}x^*_{ij}       \Big)^{-\Delta/2}
\label{eq:secmain}
\\
&
\Big(    \frac{1}{\sqrt{2\pi}}     \Big)^{2(q-1)}
\int^{\infty}_{-\infty}   \cdots \int^{\infty}_{-\infty}
\bigg[
\Big(   \frac{1}{\sqrt{2\pi}}     \Big)^{|P_1|-(2q-1)}
\int^{\infty}_{-\infty}   \cdots \int^{\infty}_{-\infty}
e^{   \frac{1}{2}   (\alphab,\betab,\X)\cdot   \mathbf{H}    \cdot(\alphab,\betab,\X)^{\T}    }
d\X
\bigg]^\Delta
d\alphab d\betab,
\notag
\end{align}
where $\H$ denotes the Hessian matrix of $\Upsilon_1$ evaluated at $(\alphab^*,\betab^*,\X^*)$ scaled by $1/\Delta$ and the operator~$\cdot$ stands for matrix multiplication.

We thus focus on computing the integral in \eqref{eq:secmain}. We begin with the inner integration. Let
\[I_1=
\bigg[
\Big(   \frac{1}{\sqrt{2\pi}}     \Big)^{|P_1|-(2q-1)}
\displaystyle\int^{\infty}_{-\infty}   \cdots \displaystyle\int^{\infty}_{-\infty}
e^{   \frac{1}{2}   (\alphab,\betab,\X)\cdot   \mathbf{H}    \cdot(\alphab,\betab,\X)^{\T}   }
d\X
\bigg]^\Delta.\]
To calculate $I_1$, we first decompose the exponent to isolate the terms involving $\X$. We obtain
\begin{equation*}
\frac{1}{2}(\alphab,\betab,\X)\cdot\H\cdot(\alphab,\betab,\X)^{\T}
=
\frac{1}{2}(\alphab,\betab)\cdot\H_{\alphab,\betab}\cdot(\alphab,\betab)^{\T}
-\frac{1}{2}\X\cdot(-\H_\X)\cdot\X^{\T}+ \Tb\cdot\X^{\T},
\end{equation*}
where $\H=\Big[\begin{array}{cc} \H_{\alphab,\betab} & \H_{\alphab\betab,\X}\\ \H^{\T}_{\alphab\betab,\X}& \H_\X\end{array}\Big]$ and $\Tb=(\alphab,\betab) \cdot \H_{\alphab\betab,\X}$. Specifically:
\begin{itemize}
\item $\H_{\alphab,\betab}$ is the square submatrix of $\H$ corresponding  to the rows indexed by $\alphab,\betab$ and  the columns indexed by $\alphab,\betab$,
\item $\H_\X$ is the square submatrix of $\H$ corresponding to the rows indexed by $\X$ and  the columns indexed by $\X$,
\item $\Tb=(\alphab,\betab) \cdot \H_{\alphab\betab,\X}$, where $\H_{\alphab\betab,\X}$ is the submatrix of $\H$ corresponding to the rows indexed by $\alphab,\betab$ and  the columns indexed by $\X$.
\end{itemize}
Note that $\H_\X$ is the Hessian of $g_1(\X)$ evaluated at $\X^*$. Since $g_1(\X)$ is concave, we have that $\H_\X$ is negative definite. Utilizing this decomposition, we obtain
\begin{align*}
I_1 &=e^{\frac{\Delta}{2}(\alphab,\betab)\cdot\H_{\alphab,\betab}\cdot(\alphab,\betab)^{\T}}
\bigg[
\Big(   \frac{1}{\sqrt{2\pi}}     \Big)^{|P_1|-(2q-1)}
\int^{\infty}_{-\infty}   \cdots \int^{\infty}_{-\infty}
e^{-\frac{1}{2}\X\cdot(-\H_\X)\cdot\X^{\T}+ \Tb\cdot\X^{\T}}d\X
\bigg]^\Delta
\\
&=\frac{1}{\big(\mathrm{Det}(-\H_\X)\big)^{\Delta/2}}e^{\frac{\Delta}{2}\big(\Tb\cdot (-\H_\X)^{-1}\cdot\Tb^{\T}+(\alphab,\betab)\cdot\H_{\alphab,\betab}\cdot(\alphab,\betab)^{\T}\big)}.
\end{align*}
We are left with the task of computing the integral
\begin{equation}
I_2=\Big(    \frac{1}{\sqrt{2\pi}}     \Big)^{2(q-1)}
\int^{\infty}_{-\infty}   \cdots \int^{\infty}_{-\infty}
e^{\frac{\Delta}{2}\left(\Tb\cdot (-\H_\X)^{-1}\cdot\Tb^{\T}+(\alphab,\betab)\cdot\H_{\alphab,\betab}\cdot(\alphab,\betab)^{\T}\right)}
d\alphab d\betab.
\label{eq:firstgaussian}
\end{equation}
Using the definition of $\Tb$, we have
\begin{equation*}
\Tb\cdot (-\H_\X)^{-1}\cdot\Tb^{\T}+(\alphab,\betab)\cdot \H_{\alphab,\betab}\cdot(\alphab,\betab)^{\T}=(\alphab,\betab)\cdot \big(\H_{\alphab,\betab}- \H_{\alphab\betab,\X}\cdot \H_\X^{-1}\cdot\H_{\alphab\betab,\X}^{\T}\big)\cdot(\alphab,\betab)^{\T}.
\end{equation*}
The matrix $\M=\H_{\alphab,\betab}- \H_{\alphab\betab,\X}\cdot \H_\X^{-1}\cdot\H_{\alphab\betab,\X}^{\T}$ is the Schur complement of the block $\H_\X$ of $\H$. In fact, we have the identity $\mathrm{Det}(\H)=\mathrm{Det}(\H_\X)\mathrm{Det}(\M)$ and in particular $\M$ is negative definite. A Gaussian integration then yields
\begin{equation}
I_2=\Big(\frac{1}{\Delta^{2(q-1)}\mathrm{Det}(-\M)}\Big)^{1/2}=\Big(\frac{\mathrm{Det}(-\H_\X)}{\Delta^{2(q-1)}\mathrm{Det}(-\H)}\Big)^{1/2}.\label{eq:secondgaussian}
\end{equation}

Combining equations \eqref{eq:secmain}, \eqref{eq:firstgaussian}, \eqref{eq:secondgaussian}, we obtain the statement of the Lemma.
\end{proof}

\subsection{The Determinants}
\label{sec:comdet}

This section addresses the computation of the determinants of the Hessians in Lemma~\ref{lem:firasympdet}. The calculations are quite complex since one has to make a choice of free variables, do the substitutions, differentiate, and then hope that the structure of the problem will prevail in the determinants. Pushing this procedure in our setting leads to complications since the choice of free variables takes away much of the combinatorial structure of the problem. We follow a different path, which amongst other things, reveals that the determinants, via the matrix-tree theorem, correspond to counting weighted trees in appropriate graphs.

The proof has two parts. The first part connects different formulations of the Hessian of a constrained maximization in an abstract setting. Essentially, this puts together well known concepts from optimization in a way that will allow to stay as close as possible to the combinatorial structure of the determinants. The second part specialises the work of the first part to compute the required determinants and is unavoidably more computational.

\subsubsection{Hessian formulations for Constrained problems}\label{sec:hesformulations}
The setting of this section is the following: we are given $\Upsilon$, a function of $\z\in\mathbb{R}^n$, subject to the linear constraints $\A\z=\b$, where $\A\in \mathbb{R}^{m\times n}$. The assumption of linear constraints stems from the setting of Lemma~\ref{lem:firasympdet}, yet the arguments extend to other constraints as well by considering gradients of these constraints at the point $\z_0$ and implicit functions. W.l.o.g., we will also assume that $\b=\mathbf{0}$.

We are interested in the Hessian $\H^f$ of a full dimensional representation of $\Upsilon$. A full dimensional representation of $\Upsilon$ consists essentially of substituting an appropriate subset of the variables $\z$ using the constraints $\A\z=\mathbf{0}$.  Note that the representation is not as much tied to $\Upsilon$ as it is tied to the space $\A\z=\mathbf{0}$. Specifically, assume that the row rank of $\A$ is $r$. In all the relevant constrained functions we consider, the constraints are not linearly independent so such an assumption is necessary. A full dimensional representation of $\Upsilon$ is specified by two submatrices of $\A$ denoted by $(\A_{f},\A_{fs})$. The matrix $\A_f$ is  a submatrix of $\A$ consisting of $r$ linearly independent rows of $\A$, so that $\A\z=\mathbf{0}$ iff $\A_f\,\z=\mathbf{0}$.  Then, $\A_{fs}$ is an $r\times r$ submatrix of $\A_f$ which is invertible. The variables corresponding to columns of $\A_{fs}$ are denoted by $\z_s$. The remaining variables $\z_f$ are called free and $\A_{ff}$ is the submatrix of $\A_f$ induced by the columns indexed by $\z_f$. Renaming if needed,  the equation $\A_f\,\z=\mathbf{0}$ may be naturally decomposed as
\begin{equation*}
\big[\begin{array}{cc} \A_{ff} & \A_{fs}\end{array}\big]\Big[\begin{array}{c} \z_f\\ \z_s\end{array}\Big]=\mathbf{0},\mbox{ so that }  \z=\Big[\begin{array}{c} \z_f\\ \z_s\end{array}\Big]=\Big[\begin{array}{c} \I\\ -(\A_{fs})^{-1}\A_{ff}\end{array}\Big]\z_f.
\end{equation*}
Thus, we can now think of $\Upsilon$ as a function which is completely determined by the variables $\z_f$ which, in contrast with the variables $\z$, span a full dimensional space.

Denote by $\H$ the unconstrained Hessian of $\Upsilon$ with respect to the variables $\z$ and by $\H^f$ the Hessian of the full dimensional representation of $\Upsilon$ with respect to the variables $\z_f$. The Hessians $\H,\ \H^f$ are connected by the following equation, which follows by straightforward matrix calculus and its proof is omitted.
\begin{equation}\label{eq:fullHform}
\H^f=
\S^{\T}\,\H\,\S, \mbox{ where }\S=\Big[\begin{array}{c} \I\\ -(\A_{fs})^{-1}\A_{ff}\end{array}\Big].
\end{equation}

Note that $\H^f$ is different, though closely related,  from the constrained Hessian $\H^c$ of $\Upsilon$ in the subspace $\A\z=\mathbf{0}$, see for example \cite[Chapter 10]{Luen}. The constrained Hessian $\H^c$ has infinitely many matrix representations, all of which correspond to similar matrices, that is, matrices with the same set of eigenvalues. A matrix representation may be obtained by first picking an orthonormal basis of the $(n-r)$-dimensional space $\{\z\,|\,\A\z=\mathbf{0}\}$. Let $\E$ denote the $n\times (n-r)$ matrix whose columns are the vectors in the basis. Then a matrix representation of $\H^c$ is given by
\begin{equation}\label{eq:Hcmatrix}
\H^c=\E^{\T}\,\H\, \E,
\end{equation}
where $\H$ is as before the unconstrained Hessian of $\Upsilon$ with respect to the variables $\z$. We are ready to prove the following. It is useful to recall here that congruent matrices have the same number of negative, zero and positive eigenvalues.
\begin{lemma}\label{lem:HfHccon}
$\H^f$ is congruent to any matrix representation of $\H^c$. Moreover, it holds that
\[\Det\big(\H^f\big)=\left.\Det\big(\H^c\big)\,\Det\big(\A_f\A_f^{\T}\big)\middle/\Det\big(\A_{fs}\big)^2\right. .\]
\end{lemma}
\begin{proof}[Proof of Lemma~\ref{lem:HfHccon}]
The columns of the matrix $\S$ defined in equation \eqref{eq:fullHform} form a basis of the space $\{\z\,|\, \A\z=\mathbf{0}\}$. Indeed, $\S$ has clearly full column rank and also $\A_f\,\S=\mathbf{0}$ implying $\A\,\S=\mathbf{0}$ as well. For future use, by a direct evaluation
\begin{equation*}
\S^{\T}\S=\I+\A_{ff}^{\T}\big(\A_{fs}\A_{fs}^{\T}\big)^{-1}\A_{ff}, \mbox{ so } \Det(\S^{\T}\S)=\Det\Big(\I+\A_{ff}\A_{ff}^{\T}\big(\A_{fs}\A_{fs}^{\T}\big)^{-1}\Big),
\end{equation*}
where the latter equality uses Sylvester's determinant theorem. This clearly yields
\begin{equation}\label{eq:mtransm}
\Det(\S^{\T}\S)=\left.\Det\big(\A_f\A_f^{\T}\big)\middle/\Det(\A_{fs})^2\right. .
\end{equation}
Comparing \eqref{eq:fullHform} and \eqref{eq:Hcmatrix}, the only difference is that $\S$ does not necessarily encode an orthonormal basis. Nevertheless, there clearly exists an invertible matrix $\Pb$ such that $\S\, \Pb$ consists of orthonormal columns, for example by the Gram-דchmidt process on the columns of $\S$. It  follows that $\Pb^{\T}\,\H^f\,\Pb$ is a matrix representation of $\H^c$. This proves the first part of the lemma and also gives $\Det\big(\H^c\big)=\Det\big(\H^f\big)\,\Det(\Pb)^2$.

For the second part, the selection of $\Pb$ implies that $(\S\, \Pb)^{\T}\S\,\Pb$ is the identity matrix and hence
$\Det(\S^{\T}\S)\,\Det(\Pb)^2=1$. The desired equality follows.
\end{proof}

Lemma~\ref{lem:HfHccon} allows us to focus on the determinant of $\H^c$ or equivalently the product of its eigenvalues. The latter may be handled using bordered Hessians. Specifically, let $\A_{f}$ be any submatrix of $\A$ induced by $r$ linearly independent rows. Then, $\lambda$ is an eigenvalue of $\H^c$ iff it is a root of the polynomial
\begin{equation}\label{eq:guarantee}
p(\lambda)=\Det\bigg(\Big[\begin{array}{cc} \mathbf{0} & \A_{f}\\ -\A^{\T}_{f} & \H-\lambda \I_n\end{array}\Big]\bigg).
\end{equation}

In our case, deleting rows of $\A$ to obtain $\A_{f}$ would cause undesirable complications. In the following, we circumvent such deletions by adding suitable ``perturbations''. We will also allow for certain degrees of freedom to select the perturbations which will be exploited in the computations. We first prove the following.

For a polynomial $p(s)$, $[s^t]p(s)$ denotes the coefficient of $s^t$ in $p(s)$.
\begin{lemma}\label{lem:perturbeffect}
Let $\M\in \mathbb{R}^{m\times m}$ be a symmetric matrix with rank  $r$ and let $\mu_i$, $i=1,\hdots,m$ be the eigenvalues of $\M$ with corresponding unit eigenvectors $\v_i$, where $\{\v_1,\hdots,\v_m\}$ is an orthonormal basis of $\mathbb{R}^m$. Then, for any symmetric matrix $\Tb\in\mathbb{R}^{m\times m}$, it holds that
\begin{equation}\label{eq:tformular}
[\epsilon^{m-r}]\,\Det\big(\epsilon \Tb+\M\big)=\prod_{i;\,\mu_i\neq 0}\mu_i\prod_{i;\,\mu_i=0}\v_i^{\T}\,\Tb\,\v_i,
\end{equation}
In particular, if $\Tb$ is positive semidefinite and $[\Tb\ \M]$ has full row rank, the rhs of \eqref{eq:tformular} is non-zero.
\end{lemma}

\begin{proof}[Proof of Lemma~\ref{lem:perturbeffect}]
Let $\M(\epsilon)=\epsilon \Tb+\M$ and denote by $\mu_i(\epsilon),\v_i(\epsilon)$ the eigenvalues and unit eigenvectors of $\M(\epsilon)$. Rellich's theorem asserts that $\mu_i(\epsilon)$ and $\v_i(\epsilon)$ are analytic functions of $\epsilon$ around $\epsilon=0$. By Hadamard's first variation formula, we have
$\displaystyle\frac{\partial\mu_i}{\partial\epsilon}=\displaystyle\v_i^{\T}\frac{\partial \M}{\partial \epsilon}\v_i$. At $\epsilon=0$, $\M$ has rank $r$ and hence exactly $m-r$ eigenvalues are zero. Thus, for small enough $\epsilon$,
\[\Det(\M(\epsilon))=\epsilon^{m-r}\prod_{i:\,\mu_i\neq 0}\mu_i\prod_{i:\,\mu_i=0}\v_i^{\T}\Tb\v_i+O(\epsilon^{m-r+1}).\]
Hence, $[\epsilon^{m-r}]\Det(\M(\epsilon))\neq0$ if for every $\v_i\neq\mathbf{0}$ such that $\M \v_i=\mathbf{0}$, we have $\v_i^{\T}\Tb\v_i\neq0$. The latter is true. Otherwise, using the positive semidefiniteness of $\Tb$, we obtain $\v_i^{\T}[\Tb\ \M]=\mathbf{0}$, contradicting that $[\Tb\ \M]$ has full row rank.
\end{proof}

The following lemma gives the promised extension of \eqref{eq:guarantee}.
\begin{lemma}\label{lem:epsilondet}
Suppose that $\Tb$ is a diagonal positive semidefinite $m \times m$ matrix such that $[\Tb\ \A]$ has full row rank. Let $\H$ (resp. $\H^c$) be the unconstrained (resp. constrained) Hessian of $\Upsilon$ evaluated at a point $\z_0$. Then, $\lambda$ is an eigenvalue of $\H^c$ iff it is a root of the polynomial
\begin{equation}\label{eq:rootsfreedom}
p(\lambda)=[\epsilon^{m-r}]\,\Det(\H_\lambda) \mbox{ where } \H_\lambda=\Big[\begin{array}{cc} \epsilon \Tb & \A\\ -\A^{\T} & \H-\lambda \I_n\end{array}\Big].
\end{equation}
Further, if $\H$ is invertible, then $\Det\big(\H^c\big)=(-1)^r\Det(\H)\displaystyle \frac{[\epsilon^{m-r}]\,\Det\big(\epsilon\Tb+\A\H^{-1}\A^{\T}\big)}{[\epsilon^{m-r}]\,\Det\big(\epsilon\Tb-\A\A^{\T}\big)}$.
\end{lemma}

\begin{proof}[Proof of Lemma~\ref{lem:epsilondet}]
Let $\Tb=(t_{i,j})_{i,j\in[m]}$ and $\H_\lambda=(h_{i,j})_{i,j\in[m+n]}$. Let $\mathcal{W}=\binom{[m]}{m-r}$ and for $W\in \mathcal{W}$ let  $P_W=\{\sigma\in S_{m+n}\, |\, \{i\in[m]\,| \, \sigma(i)=i\}=W\}$. Since $\Tb$ is diagonal, by Leibniz's formula,
\begin{equation}\label{eq:plexpression}
p(\lambda)=[\epsilon^{m-r}]\Det(\H_\lambda)=\sum_{W\in \mathcal{W}}\prod_{i\in W}t_{i,i}\sum_{\sigma\in P_W}\sgn(\sigma)\prod_{i\in[m+n]\backslash W}h_{i,\sigma(i)}.
\end{equation}
Let $\A_{[m]\backslash W}$ be the $r\times n$ submatrix of $\A$ which is obtain by excluding the rows indexed by $W$. Identifying permutations in $P_W$ with permutations of $[n+r]$ in the natural way, we obtain
\begin{equation}\label{eq:polyalign}
\sum_{\sigma\in P_W}\sgn(\sigma)\prod_{i\in[m+n]\backslash W}h_{i,\sigma(i)}=\Det\bigg(\Big[\begin{array}{cc} \mathbf{0} & \A_{[m]\backslash W}\\ -\A_{[m]\backslash W}^{\T} & \H-\lambda \I_n\end{array}\Big]\bigg)\equiv q_W(\lambda).
\end{equation}
If $\A_{[m]\backslash W}$ has row rank $<r$, then $q_W(\lambda)$ is 0. Otherwise, the roots of $q_W(\lambda)$ are the eigenvalues of $\H^c$, c.f. \eqref{eq:guarantee}. By \eqref{eq:plexpression}, this is also the case for $p(\lambda)$, provided it is not identically zero.

To prove that $p(\lambda)$ is nonzero, we prove that the leading coefficient of $p(\lambda)$ is nonzero. Starting from \eqref{eq:polyalign} and plugging into \eqref{eq:plexpression}, the leading coefficient of $p(\lambda)$ can easily be seen to equal
\begin{equation*}
[\epsilon^{m-r}]\Det\bigg(\Big[\begin{array}{cc} \epsilon \Tb & \A\\ -\A^{\T} & -\I_n\end{array}\Big]\bigg)=[\epsilon^{m-r}](-1)^n\Det(\epsilon \Tb-\A \A^{\T}),
\end{equation*}
where in the latter equality we used the Schur complement of the block $-\I_n$. The last expression is non-zero by Lemma~\ref{lem:perturbeffect}.

The determinant of $\H^c$ is the product of its eigenvalues, which in turn equals $(-1)^{n-r}p(0)$ divided by the leading coefficient of $p(\lambda)$. The latter has already been computed. The former, using the Schur complement of the invertible $\H$, is equal to $[\epsilon^{m-r}]\Det(\H)\,\Det\big(\epsilon\Tb-\A\H^{-1}\A^{\T}\big)$. This concludes the proof.
\end{proof}

Finally, we combine the above lemmas to obtain the following.
\begin{lemma}\label{lem:blackbox}
Let $\Upsilon$ be a function of $\z\in\mathbb{R}^n$ subject to the linear constraints $\A\z=\b$, where $\A\in \mathbb{R}^{m\times n}$ and $\A$ has rank $r$. Let $\big(\A_{f}, \A_{fs}\big)$ specify a full dimensional representation of $\Upsilon$ and let $\H^f$ be the corresponding Hessian of $\Upsilon$ evaluated at a point $\z_0$.

Suppose $\Tb$ is a positive semidefinite diagonal matrix with dimensions $m\times m$ such that $[\Tb\ \A]$ has full row rank.  Let $\H$  be the unconstrained Hessian of $\Upsilon$ evaluated at $\z_0$. If $\H$ is invertible, then
\begin{equation}\label{eq:blackbox}
\Det\big(-\H^f\big)=\frac{L\big(\A_f,\A,\Tb\big)}{\Det\big(\A_{fs}\big)^2}\,\Det(-\H)\, [\epsilon^{m-r}]\,\Det\big(\epsilon\Tb-\A\H^{-1}\A^{\T}\big),
\end{equation}
where $L\big(\A_{f}, \A,\Tb\big)=(-1)^r\left.\Det\big(\A_{f}\A_{f}^\T\big)\middle/[\epsilon^{m-r}]\,\Det\big(\epsilon\Tb-\A\A^{\T}\big)\right.$.
\end{lemma}
\begin{proof}[Proof of Lemma~\ref{lem:blackbox}]
Just combine Lemmas~\ref{lem:HfHccon} and~\ref{lem:epsilondet}. The minor sign change $-\H^f$ in the statement can easily be accounted by applying the lemmas to the function $-\Upsilon$.
\end{proof}

The rhs of \eqref{eq:blackbox} has two qualitatively different factors: the factor $L(\A_f,\A,\Tb)/\Det\big(\A_{fs}\big)^2$ depends on the specific full dimensional representation, while the remaining factor is tied to the Hessian of $\Upsilon$. The technical convenience of Lemma~\ref{lem:blackbox} is dual: first, it gives an explicit formula for $\Det\big(-\H^f\big)$ without doing substitutions which would hinder the combinatorial view of the constraints $\A$; second, it isolates the deletions of rows of $\A$ in the factor $L(\A_f,\A,\Tb)$ and leaves untouched the more complicated matrix $\A\H^{-1}\A^{\T}$.

\subsubsection{The Computations}\label{sec:computations}
In this section, we utilize Lemma~\ref{lem:blackbox} to compute the determinants in  Lemma~\ref{lem:firasympdet}. \vskip 0.2cm

\textbf{Notation}: For a vector $\z\in \mathbb{R}^n$ we denote by $\z^D$ the $n\times n$ diagonal matrix $\diag\{z_1,\hdots,z_n\}$. For vectors $\z_i\in \mathbb{R}^{m_i}$, $i=1,\hdots,t$ we denote by $[\z_1,\hdots,\z_{t}]^{\T}$ the $\mathbb{R}^{\sum_i m_i}$ vector which is the concatenation of the vectors $\z_1,\hdots,\z_t$. For matrices $\A$ and $\B$, $\A\otimes \B$ will denote the Kronecker product of $\A,\B$, while  $\A\oplus\B$ is the direct sum of $\A,\B$, that is, the block diagonal matrix $\diag\{\A,\B\}$. The expression $\oplus_2 \A$ is a shorthand for $\A\oplus \A$. Further, $\I_{n}$ denotes  the identity matrix of dimensions $n\times n$. Finally, $\ones_n,\mathbf{0}_n$ denote the all-one and all-zero $n$-dimensional vector.

To start, the equality constraints in \eqref{eq:spacefirst} may be written in the form
\begin{equation*}
\A_1\,\big[\alphab,\, \betab,\, \X\big]^{\T}=\mathbf{0}.
\end{equation*}
The matrix $\A_1$ has dimensions  $(2q+2)\times (|P_1|+2q)$ (cf. \eqref{eq:defpone} for the definition of $P_1$). Note that we exclude from consideration variables $x_{ij}$ which are hard-coded to zero. This is done to ensure that the unconstrained Hessians are invertible, so that Lemma~\ref{lem:blackbox} applies directly.  It will be useful to decompose the matrix $\A_1$ as
\begin{equation}\label{eq:A2blockform}
\A_1=\Big[\begin{array}{cc} \A_{1,\alphab\betab}& \mathbf{0}\\ -\I_{2q} & \A_{1,\X}\end{array}\Big],
\end{equation}
where $\A_{1,\alphab\betab}, \A_{1,\X}$ have dimensions $2\times 2q$ and $2q\times |P_1|$, respectively.

The easiest way to handle the matrix $\A_{1,\X}$ is  as the incidence matrix of a bipartite graph $G_\x$. First, we introduce some notation: for an undirected graph $G$, we denote by $\A_G$ the 0,1  incidence matrix of $G$, by $\Rb_G$ the adjacency matrix of $G$, by $\Db_G$ the diagonal matrix whose diagonal entries are equal to the degrees of the vertices in $G$ and by $\boldsymbol{\Lambda}_G$ the matrix $\Db_G+\Rb_G$.
We will also be interested in the case where the graph $G$ is weighted, in which case we assume that the weights on the edges are given by the diagonal entries of a square diagonal matrix $\W_G$. We denote by $\Rb^{w}_G, \Db^{w}_G,\boldsymbol{\Lambda}^{w}_G$ the weighted versions of the matrices  $\Rb_G, \Db_G, \boldsymbol{\Lambda}_G$. It is well known that
\begin{equation}\label{eq:adjinc}
\A_G\A_G^{\T}=\boldsymbol{\Lambda}_G,\ \A_G\W_G\A_G^{\T}=\boldsymbol{\Lambda}^{w}_G.
\end{equation}
The bipartite graph $G_{\x}$ has vertex bipartition $([q],[q])$ and an edge $(i,j)$ is present iff $(i,j)\in P_1$, that is, $B_{ij}>0$. Since $\B$ is symmetric and irreducible, $G_{\X}$ is undirected and connected.  An edge $(i,j)$ in $G_{\X}$ has weight $x_{ij}$. In the languange of \eqref{eq:adjinc}, $\W_{G_\X}=\X^{D}$ (the choice of $W_{G_\X}$ will become apparent when we consider the unconstrained Hessian). Applying \eqref{eq:adjinc} to the graph $G_{\X}$ is useful to do explicitly in order to decompose the resulting matrices. In particular, since these graphs are undirected and bipartite, we have
\begin{gather}
\boldsymbol{\Lambda}^{w}(G_{\X})=\Big[\begin{array}{cc} \alphab^{D}& \Sb_{\X} \\ \Sb^{\T}_{\X}&\betab^{D}\end{array}\Big],
\label{eq:almostlaplacians1}
\end{gather}
where $\Sb_{\X}$ is the $q\times q$ matrix whose $(i,j)$ entry is $x_{ij}$. Note that the total weight of the edges incident to vertices in $G_{\X}$ (in other words, the diagonal entries of the matrix $\Db^{w}$) was substituted using \eqref{eq:spacefirst}.

We next state the unconstrained Hessians that will be of interest to us. From Lemma~\ref{lem:firasympdet}, these are: (i) $\H_{1,\X}$, the Hessian of $\Upsilon_1/\Delta$ with respect to $\X$ when $\alphab,\betab$ are fixed, (ii) $\H_{1}$, the Hessian of $\Upsilon_1/\Delta$ with respect to $\alphab,\betab,\X$. These two matrices are all diagonal and by inspection one can check that
\begin{equation}\label{eq:unconstrained}
\begin{gathered}
(\H_{1,\X})^{-1}=-\X^{D},\ \ (\H_{1})^{-1}=\frac{\Delta}{\Delta-1} \alphab^{D}\oplus \frac{\Delta}{\Delta-1} \betab^{D}\oplus (\H_{1,\X})^{-1}, \\
\Det(-\H_{1,\X})^{-1}=\prod_{(i,j)\in P_1}x_{ij},\quad \Det(-\H_{1})^{-1}=\Det(-\H_{1,\X})^{-1}\,\Big(\frac{\Delta}{\Delta-1}\Big)^{2q} \prod_{i\in[q]}\alpha_{i}\prod_{j\in[q]}\beta_{j}.
\end{gathered}
\end{equation}

We are now ready to evaluate these matrices at a global maximum $(\alphab^*,\betab^*, \X^*)$ of $\Upsilon_1$. Henceforth, we will not explicitly use asterisks in the notation with the understanding that the values of all the variables are fixed to their optimal values. 
We will apply Lemma~\ref{lem:blackbox} to the matrices $\H_{1,\X}^f,\H_{1}^f$ using the matrices
\begin{equation}\label{eq:matricesT}
\Tb_{1,\X}=\alphab^{D}\oplus \betab^{D},\ \Tb_{1}=\I_2\oplus \mathbf{0}_{2q\times 2q},
\end{equation}
respectively (in \eqref{eq:matricesT}, $\mathbf{0}_{2q\times 2q}$ denotes the $2q\times 2q$ matrix with all zeros). We first compute the determinants of $\M_{1,\X}:=\epsilon \Tb_{1,\X}-\A_{1,\X}(\H_{1,\X})^{-1}\A_{1,\X}^{\T}$, $\M_{1}:=\epsilon \Tb_{1}-\A_{1}(\H_{1})^{-1}\A_{1}^{\T}$, which contribute the most interesting factors in Lemma~\ref{lem:blackbox}.

We begin with the simplest of these matrices, $\M_{1,\X}$. Note that $\A_{1,\X}$ has rank $2q-1$, so by Lemma~\ref{lem:blackbox} we want to compute $[\epsilon]\,\Det(\M_{1,\X})$. Using \eqref{eq:adjinc}, \eqref{eq:unconstrained}, \eqref{eq:matricesT}, it is straightforward to check that $\M_{1,\X}$ has the following form
\begin{equation}\label{eq:firstWHmoment}
\M_{1,\X}=\Big[\begin{array}{cc} \alphab^{D}(\epsilon \I_q+\I_q) & \S_{\X}\\ \S_{\X}^{\T} & \betab^{D}(\epsilon \I_q+\I_q)\end{array}\Big],\mbox{ so } \Det(\M_{1,\X})=\Big(\prod_{i\in [q]}\alpha_i\prod_{j\in[q]}\beta_{j}\Big)\, \Det\big(\epsilon \I_q+\I_q+\Jb\big),
\end{equation}
where $\Jb$ is the matrix in Lemma~\ref{lem:bpeigenspace}. Note that in Equation \eqref{eq:firstWHmoment}, to get the second equality, we did the following operations on $\M_{1,\x}$: for $i=1,\hdots,q$, we divided the $i$-th row of by $\sqrt{\alpha_i}$, the $i$-th column by $\sqrt{\alpha_i}$, the $(i+q)$-th row by $\sqrt{\beta_i}$, the $(i+q)$-th column by $\sqrt{\beta_i}$. The eigenvalues of the matrix $\epsilon\I_q+\I_q+\Jb$ are shifts of the eigenvalues of $\Jb$ and are given by
\[\epsilon, \epsilon+2,\epsilon+1\pm \lambda_1, \hdots,\epsilon+1\pm \lambda_{q-1},\]
c.f., Lemma~\ref{lem:bpeigenspace} for the definition of the $\lambda_i$ and their properties. We thus obtain
\begin{equation}\label{eq:firstXXX}
[\epsilon]\,\Det\big(\M_{1,\X}\big)=2\prod_{i\in [q]}\alpha_i\prod_{j\in[q]}\beta_{j}\prod_{i\in [q-1]}\big(1-\lambda_i^2\big).
\end{equation}

The determinant of the matrix $\M_1$ is more complicated to compute due to its more intricate block structure, which requires using Schur's complement formula to handle. As in the previous argument, we first write out its block structure and then appropriately normalize the resulting matrix. Here the normalization is slightly more intricate. The analog of \eqref{eq:firstWHmoment} is
\begin{equation}\label{eq:secondWHmoment}
\Det\big(\M_1\big)=\Det\big(\H_1'\big)\, \prod_{i\in[q]}\alpha_i\prod_{j\in[q]}\beta_j,
\end{equation}
where 
\begin{equation}\label{eq:secfinalform}
\H_1':=\frac{\Delta}{\Delta-1}\Big[\begin{array}{cc}
(\epsilon\frac{\Delta-1}{\Delta}-1)\I_{2} & \Vb \\ \Vb^{\T} & -\frac{\Delta-1}{\Delta}\Wb\end{array}\Big],\quad 
\Wb:=\frac{1}{\Delta-1}\I_{2q}-\Jb,\quad
\Vb^{\T}:=\bigg[\begin{array}{cc} \sqrt{\alphab} &\mathbf{0}_q \\ \mathbf{0}_q & \sqrt{\betab}\end{array}\bigg].
\end{equation}
($\Jb$ is the matrix in Lemma~\ref{lem:bpeigenspace}; note also that $\Vb$ has dimension $2\times 2q$ and its two rows are given by  $\sqrt{\alpha_1},\hdots,\sqrt{\alpha_q},0,\hdots,0$ and $0,\hdots,0,\sqrt{\beta_1},\hdots,\sqrt{\beta_q}$ where each row has $q$ zeros.) Equation \eqref{eq:secondWHmoment} is obtained by performing the following operations on $\M_1$: for $i,j=1,\hdots,q$, divide the $2+i$ row by $\sqrt{\alpha_i}$ and the $2+q+j$ row by $\sqrt{\beta_{j}}$; and the same operations on columns. 

In light of~\eqref{eq:secondWHmoment}, it suffices to compute $\Det(\H_1')$. To do this, we proceed by taking the Schur complement of the matrix $\Wb$. The spectrum of $\Wb$ is
\[t\pm 1,t\pm\lambda_1,\hdots,t\pm\lambda_{q-1},\]
where $t=1/(\Delta-1)$. It follows that 
\begin{equation}\label{eq:puttogetherb}
\Det(\Wb)=-\frac{\Delta(\Delta-2)}{(\Delta-1)^{2q}}\prod_{i\in[q-1]}\big(1-(\Delta-1)^2\lambda^2_i\big),
\end{equation}
where $\lambda_i$, $i\in[q-1]$ are as in Lemma~\ref{lem:bpeigenspace}. Note that $\Wb$ is invertible, since the $\lambda_i$'s are non-negative and $\max\lambda_i<\frac{1}{\Delta-1}$. By taking the Schur complement of the matrix $\Wb$ in $\H_1'$, we obtain
\begin{equation}\label{eq:simpleMformone}
\Det(\H_1')=\Big(\frac{\Delta}{\Delta-1}\Big)^{2}\Det\big(\Wb\big)\,\Det\Big(\epsilon\frac{\Delta-1}{\Delta}\I_{2}+\Zb\Big), \text{ where } \Zb=-\I_2+\frac{\Delta}{\Delta-1}\Vb\Wb^{-1} \Vb^{\T}.
\end{equation}

We are left with the evaluation of $\Det\big(\epsilon\frac{\Delta-1}{\Delta}\I_{2}+\Zb\big)$. The complication here is the nontrivial inverse of $\Wb$ appearing in the formulation of $\Zb$. The key idea to circumvent the computation of $\Wb^{-1}$ is the following equality
\[\Vb \Wb=\Big(\frac{1}{\Delta-1}\I_2-\Jb'\Big)\Vb, \text{ where } \Jb'=\Big[\begin{array}{cc}0 & 1\\ 1 & 0\end{array}\Big].\]
The equality can be checked using the relations $\sum_{j}x_{ij}=\alpha_i$ and $\sum_{i}x_{ij}=\beta_{j}$. Using that $\Vb\Vb^{\T}=\I_2$, we obtain
\begin{equation}
\Zb=-\I_2+\frac{\Delta}{\Delta-1}\Big(\frac{1}{\Delta-1}\I_{2}-\Jb'\Big)^{-1}\Vb\Vb^{\T}
=-\frac{\Delta-1}{\Delta-2}\Big[\begin{array}{cc}1 & 1\\ 1 & 1\end{array}\Big],\label{eq:simpleMformtwo}
\end{equation}
We thus obtain
\begin{equation}
[\epsilon]\Det\Big(\epsilon\frac{\Delta-1}{\Delta}\I_{2}+\Zb\Big)=-\frac{2 (\Delta-1)^2}{\Delta(\Delta-2)}.\label{eq:puttogethera}
\end{equation}
Plugging \eqref{eq:puttogetherb}  and \eqref{eq:puttogethera} in \eqref{eq:simpleMformone}, we obtain
\begin{equation*}
[\epsilon]\Det(\H_1')=\frac{2\Delta^{2}}{(\Delta-1)^{2q}}\prod_{i\in[q-1]}\big(1-(\Delta-1)^2\lambda^2_i\big).
\end{equation*}
Using this and \eqref{eq:secondWHmoment}, we obtain 
\begin{equation}\label{eq:detH1}
[\epsilon]\Det\big(\epsilon\Tb_1-\A_{1} (\H_{1})^{-1}\A_{1}^{\T}\big)=\frac{2\Delta^{2}}{(\Delta-1)^{2q}}\displaystyle\prod_{i\in[q]}\alpha_i\displaystyle\prod_{j\in[q]}\beta_j\prod_{i\in[q-1]}\big(1-(\Delta-1)^2\lambda^2_i\big).
\end{equation}

Equations \eqref{eq:unconstrained}, \eqref{eq:firstXXX},  \eqref{eq:detH1} deal with the factors in Lemma \ref{lem:blackbox} which are tied to the Hessians of the functions. While these contribute the most interesting factors, some care is needed to deal with the remaining factors. This is accomplished in the following lemma, which is given in the end of this section.
\begin{lemma}\label{lem:matrixtree}
Let $\big((\A_{1,\X})_f, (\A_{1,\X})_{fs}\big)$, $\big((\A_{1})_f, (\A_{1})_{fs}\big)$ specify arbitrary full dimensional representations of the spaces $\A_{1,\X}\X=\mathbf{0}$, $\A_1\,[\alphab,\betab,\X]^{\T}=\mathbf{0}$, respectively. Then:
\begin{gather}
\Det\big((\A_{1,\X})_{fs}\big)^2=\Det\big((\A_{1})_{fs}\big)^2=1,\label{eq:fulldimspec}\\
L\big((\A_{1,\X})_{f},\A_{1,\X}, \Tb_{1,\X})=1/2,\quad L\big((\A_{1})_{f},\A_{1}, \Tb_{1})=1/2,\label{eq:lratios}
\end{gather}
where $\Tb_{1,\X}, \Tb_1$ are given by \eqref{eq:matricesT} and the quantities in \eqref{eq:lratios} are defined  in Lemma~\ref{lem:blackbox}.
\end{lemma}

We are now ready to finish the proof of Lemma~\ref{lem:momentsratio}.
\begin{proof}[Proof of Lemma~\ref{lem:momentsratio}]
Apply Lemma~\ref{lem:blackbox} two times to unravel the determinants appearing in Lemma~\ref{lem:firasympdet}.  Each of the resulting quantities has been computed and appears in one of \eqref{eq:unconstrained}, \eqref{eq:firstXXX}, \eqref{eq:detH1} or Lemma~\ref{lem:matrixtree}. Straightforward substitutions yield
\begin{align*}
\Det(-\H^f_{1,\X})&=\Big(\prod_{(i,j)\in P_1}x_{ij}\Big)^{-1}\prod_{i\in[q]} \alpha_i\prod_{j\in[q]}\beta_j \prod_{i\in[q-1]}(1-\lambda_i)^2,\\
\Det(-\H^f_{1})&=\frac{1}{\Delta^{2(q-1)}}\Big(\prod_{(i,j)\in P_1}x_{ij}\Big)^{-1}\prod_{i\in[q]} \alpha_i\prod_{j\in[q]}\beta_j \prod_{i\in[q-1]}(1-\lambda_i^2).
\end{align*}
Thus, Lemma~\ref{lem:firasympdet} gives 
\begin{equation}\label{eq:firfirfir}
\lim_{n\rightarrow\infty}\frac{\E_{\G}[Z^{\p}_G]}{e^{n\Upsilon_1(\alphab,\betab,\X)}}=\prod^{q-1}_{i=1}\big(1-(\Delta-1)^2\lambda^2_i\big)^{-1/2}\prod^{q-1}_{i=1}\big(1-\lambda^2_i\big)^{-(\Delta-1)/2}.
\end{equation}
Note that in the last expression, only the eigenvalues of the matrix $\Jb$ (different from 1) in Lemma~\ref{lem:bpeigenspace} appear. For the asymptotics of $\E_{\G}[(Z^{\p}_G)^2]$, by Lemma~\ref{lem:secsecsec}, it suffices to consider the spin system with interaction matrix $\B\otimes \B$ (and dominant phase $\gammab=\alphab\otimes \alphab,\deltab=\betab\otimes\betab,\Y=\X\otimes \X$). The eigenvalues (different from 1) of the matrix $\Jb\otimes \Jb$ are $\lambda_i$ for $i\in[q-1]$ and  $\lambda_i\lambda_j$ for $i,j\in[q-1]$. Thus, we obtain
\begin{equation}\label{eq:secsecsec}
\lim_{n\rightarrow\infty}\frac{\E_{\G}[(Z^{\p}_G)^2]}{e^{n\Upsilon_2(\gammab,\deltab,\Y)}}=C\cdot \prod^{q-1}_{i=1}\big(1-(\Delta-1)^2\lambda^2_i\big)^{-1/2}\prod^{q-1}_{i=1}\big(1-\lambda^2_i\big)^{-(\Delta-1)/2},
\end{equation}
where $C$ is the constant in the statement of the lemma. Combining \eqref{eq:firfirfir} and \eqref{eq:secsecsec} with Lemma~\ref{lem:redrzero} yields the result.
\end{proof}

Finally, we give the proof of Lemma~\ref{lem:matrixtree}.
\begin{proof}[Proof of Lemma~\ref{lem:matrixtree}]
We first prove \eqref{eq:fulldimspec}. Since $\A_{1,\x}$ is the incidence matrix of the bipartite graph $G_{\X}$, it is a totally unimodular matrix. By the way full dimensional representations are chosen, the matrix $(\A_{1,\X})_{fs}$ is invertible and hence its determinant squared equals 1. For $(\A_{1})_{fs}$, observe that $(\A_{1})_{fs}$  has the  block decomposition
\[(\A_{1})_{fs}=\Big[\begin{array}{cc} (\A_{1,\alphab\betab})_{fs} & \mathbf{0}\\ -\I & (\A_{1,\X})_{fs}\end{array}\Big], \mbox{ so that } \Det\big((\A_{1})_{fs}\big)=\Det\big((\A_{1,\alphab\betab})_{fs}\big)\, \Det\big((\A_{1,\X})_{fs}\big).\]
Since  $\A_{1,\alphab\betab}$, $\A_{1,\X}$ are totally unimodular, any invertible submatrix of them has determinant $\pm 1$. This concludes the proof of \eqref{eq:fulldimspec}.

We next turn to \eqref{eq:lratios}. We begin with $L\big((\A_{1,\X})_f,\A_{1,\X},\Tb_{1,\X}\big)$. The argument is closely related to the proof of Kirchoff's Matrix-Tree Theorem, but is written in a way that it easily extends to the more complicated $L\big((\A_{1})_f,\A_{1},\Tb_1\big)$.

Denote by $\mu_1,\hdots,\mu_{2q-1}$ the non-zero eigenvalues of $\A_{1,\X}\A_{1,\X}^{\T}$; there are exactly $2q-1$ of those since $G_{\x}$ is a  connected bipartite graph. Moreover, $\v_0^{\T}=\frac{1}{\sqrt{2q}}[-\ones_{q}\ \ \ones_{q}]$ is the unit eigenvector of $\A_{1,\X}\A_{1,\X}^{\T}$ with eigenvalue 0. We claim that
\begin{equation}\label{eq:qazwsx2}
[\epsilon]\Det\big(\epsilon \Tb_{1,\X}-\A_{1,\X}\A_{1,\X}^{\T}\big)=-\frac{\prod_{i\in[2q-1]}\,\mu_i}{q},\quad \Det\big((\A_{1,\X})_f\,(\A_{1,\X})_f^{\T}\big)=\frac{\prod_{i\in [2q-1]}\,\mu_i}{2q},
\end{equation}
which yields that  $L\big((\A_{1,\X})_f,\A_{1,\X},\Tb_{1,\X}\big)=1/2$, as wanted. The first equality is a direct application of Lemma~\ref{lem:perturbeffect}, after observing that $\v_0^{\T}\Tb_{1,\X}\v_0=1/q$. The second can be proved as follows. The matrix $(\A_{1,\X})_f\,(\A_{1,\X})_f^{\T}$ is a principal minor of $\A_{1,\X}\A_{1,\X}^{\T}$, the specific principal minor is clearly determined  by which row of $\A_1$ we chose to delete to obtain $(\A_{1})_f$.  Since $\A_{1,\X}\A_{1,\X}^{\T}$ has exactly one zero eigenvalue, we have
\begin{equation}\label{eq:summinors}
\prod_{i\in [2q-1]}\,\mu_i=\sum_{W\in \binom{[2q]}{2q-1}}\Det\big((\A_{1,\X})_W\,(\A_{1,\X})_W^{\T}\big),
\end{equation}
where $(\A_{1,\X})_W$ is the submatrix of $\A_{1,\X}$ induced by the rows indexed with $W$. It is easily checked that for any $W,W'\in \binom{[2q]}{2q-1}$, there exists a unitary matrix $\Pb$ such that $(\A_{1,\X})_W=\Pb(\A_{1,\X})_{W'}$, so that all summands in \eqref{eq:summinors} are equal. Indeed, since $\A_{1,\x}$ corresponds to the incidence matrix of a bipartite graph, the sum of the first $q$ rows (as vectors) equals the sum of the last $q$ rows. It follows that any row  of $\A_{1,\X}$ can be expressed as a $\{1,-1\}$  linear combination of the remaining rows, which easily yields the existence of $\Pb$ with the desired properties. Hence, for any $(\A_{1,\X})_f$ as in the statement of the Lemma, the second equality in \eqref{eq:qazwsx2} holds as well.

We finally give a proof sketch for $L\big((\A_{1})_f,\A_{1},\Tb_1\big)=1/2$. The matrix $\A_1\A_1^{\T}$ has zero as an eigenvalue by multiplicity one. Denote by $\sigma_1,\hdots,\sigma_{2q+1}$ the non-zero eigenvalues of $\A_1\A_1^{\T}$. By looking at the  space $\z \A_1=\mathbf{0}$, it is easy to see that $\v_1=\frac{1}{\sqrt{2(q+1)}}[-1,1,-\ones_{q},\, \ones_q]^{\T}$ is a unit length eigenvector for the eigenvalue 0.
Moreover, the analog of \eqref{eq:summinors} is
\begin{equation}\label{eq:summinors2}
\prod_{i\in [2q+1]}\sigma_i=\sum_{W\in \binom{[2q+2]}{2q+1}}\Det\big((\A_1)_W\,(\A_1)_W^{\T}\big).
\end{equation}
%
Hence, the equality $L\big((\A_{1})_f,\A_{1},\Tb_1\big)=1/2$ is obtained by the following analog of \eqref{eq:qazwsx2}
\begin{equation*}
[\epsilon]\Det\big(\epsilon \Tb_{1}-\A_1\A_1^{\T}\big)=-\frac{\prod_{i\in [2q+1]}\sigma_i}{q+1},\qquad \Det\big((\A_{1})_f\,(\A_{1})_f^{\T}\big)=\frac{\prod_{i\in [2q+2]}\sigma_i}{2(q+1)}.
\end{equation*}
\end{proof}

\section{Uniqueness of semi-translation invariant measures (Antiferromagnetic Potts)}\label{sec:semi-uniqueness}
In this section, we prove Lemma~\ref{lem:semi-uniqueness}. As noted earlier, the proof extends the respective argument in \cite{BW} for colorings in the antiferromagnetic Potts model setting. The technical details, due to the presence of the extra parameter $B$, are relatively more intricate.

\begin{proof}[Proof of Lemma~\ref{lem:semi-uniqueness}]
W.l.o.g. we may assume that the scaling factors in \eqref{eq:treePotts} are equal to 1. We may also assume that $R_1\geq \hdots \geq R_q$. Then the equations
easily imply  $C_1\leq \hdots\leq C_q$. Define
\begin{equation*}
\alpha=\frac{R_1}{R_q},\quad \beta=\frac{R_1+\hdots+R_{q-1}}{(q-1)R_q}, \quad S=R_1+\hdots+R_{q-1}.
\end{equation*}
We clearly have $\alpha\geq \beta\geq 1$, and we may assume for
the sake of contradiction that $\beta>1$. Note that
\begin{gather*}
\alpha^{1/d}=\left(\frac{R_1}{R_q}\right)^{1/d}=1+\frac{(1-B)(C_q-C_1)}{C_1+\hdots+C_{q-1}+BC_q}\\
C_q=(R_1+\hdots+R_{q-1}+BR_q)^d=\big[(q-1)\beta+B\big]^d R^d_q\\
C_1=(BR_1+R_2+\hdots+R_q)^d=\big[(q-1)\beta+1-(1-B)\alpha\big]^dR^d_q
\end{gather*}
Moreover, by Holder's inequality or otherwise, we have
\begin{align*}
C_1+\hdots+C_{q-1}+BC_q&=\sum^{q-1}_{i=1}\big[S+R_q-(1-B)R_i\big]^d+B(S+BR_q)^d\\
&\geq (q-1)\Big[\frac{q-2+B}{q-1} S+(q-1)R_q\Big]^d+B(S+BR_q)^d\\
&=(q-1)\big[(q-2+B)\beta +1\big]^dR^d_q+B\big[(q-1)\beta+B\big]^dR^d_q.
\end{align*}
Thus, we obtain that every solution must satisfy
\begin{align*}
\alpha^{1/d}\leq 1+\frac{(1-B)\left\{\big[(q-1)\beta+B\big]^d-\big[1-(1-B)\alpha+(q-1)\beta\big]^d\right\}}{(q-1)\big[(q-2+B)\beta +1\big]^d+B\big[(q-1)\beta+B\big]^d}\Longleftrightarrow \\
0\leq
1-\alpha^{1/d}+\frac{(1-B)\bigg[1-\Big(1-\frac{(1-B)(\alpha-1)}{(q-1)\beta+B}\Big)^d\bigg]}{(q-1)\Big[1-\frac{(1-B)(\beta-1)}{(q-1)\beta+B}\Big]^d+B}=:f(\alpha,\beta,B).
\end{align*}
To obtain a contradiction, our goal is to prove that for $q$ and
$B$ as in the statement of the lemma, when $(q-1)\beta>\alpha\geq
\beta>1$, it holds that $f(\alpha,\beta,B)<0$.

It is easy to see that $f$ is decreasing in $B$. This immediately
yields the lemma for $q\geq \Delta$: it holds that
$f(\alpha,\beta,B)\leq f(\alpha,\beta,0)<0$, since the last
inequality was proved by \cite{BW}. For $q\leq d$ and $B\geq \frac{d+1-q}{d+1}:=B_c$, this yields
\begin{equation*}
f(\alpha,\beta,B)\leq f\left(\alpha,\beta,B_c\right)=:g(\alpha,\beta).
\end{equation*}
We first prove that $g(\alpha,\beta)\leq g(\beta,\beta)$. For $q=2$
there is nothing to prove. Hence we may assume that $d\geq q\geq
3$. Clearly it suffices to prove that $g$ is decreasing in
$\alpha$. This requires a fair bit of work, so we state it as a
Lemma to prove later.
\begin{lemma}\label{lem:gdecreasing}
For $d\geq q\geq 3$ and $B_c=\frac{d+1-q}{d+1}$, the function
$g(\alpha,\beta)$ is decreasing in $\alpha$ for $\alpha\geq
\beta>1$.
\end{lemma}

We finish the proof by showing that for $\beta\geq 1$, it holds that
$g(\beta,\beta)\leq 0$ with equality iff $\beta=1$. After massaging the inequality, this reduces to
\[1\leq \left[1-\displaystyle\frac{(1-B_c)(\beta-1)}{(q-1)\beta+B_c}\right]^d\left[(q-1)\left(\beta^{1/d}-1\right)+1-B_c\right]+B_c\beta^{1/d}=:h(\beta)\]
Note that the inequality holds at equality for $\beta=1$, so it suffices to prove
$h'(\beta)>0$ for $\beta>1$, which is the assertion of the next lemma.
\begin{lemma}\label{lem:hincreasing}
For $d\geq q\geq 2$ and $B_c=\frac{d+1-q}{d+1}$, the function
$h(\beta)$ is increasing for $\beta\geq1$.
\end{lemma}
Modulo the proofs of Lemmas~\ref{lem:gdecreasing} and~\ref{lem:hincreasing}, which are given below, the proof is
complete.
\end{proof}

\begin{proof}[Proof of Lemma \ref{lem:gdecreasing}]
We compute
\begin{equation*}
\frac{\partial g}{\partial
\alpha}=-\frac{1}{d}\alpha^{-(d-1)/d}+\frac{(1-B_c)^2}{(q-1)\beta+B_c}\cdot
\frac{d\left[1-\frac{(1-B_c)(\alpha-1)}{(q-1)\beta+B_c}\right]^{d-1}}{(q-1)\left[1-\frac{(1-B_c)(\beta-1)}{(q-1)\beta+B_c}\right]^d+B_c}
\end{equation*}
Let $F(x)=x\left[1-\frac{(1-B_c)(x-1)}{(q-1)\beta+B_c}\right]^{d}$
for $x\in[\beta,(q-1)\beta]$. Straightforward manipulations show that $\frac{\partial g}{\partial \alpha}<0$ is equivalent to
\begin{equation}
d^2(1-B_c)^2F(\alpha)^{(d-1)/d}\leq
\big[(q-1)\beta+B_c\big]
\bigg[(q-1)\Big(1-\frac{(1-B_c)(\beta-1)}{(q-1)\beta+B_c}\Big)^d+B_c\bigg].\label{eq:mainobstacle}
\end{equation}
We prove that $F(x)$ is decreasing in $[\beta,(q-1)\beta]$. It is simple to check that
\begin{equation*}
F'(x)=
\left[1-\frac{(1-B_c)(x-1)}{(q-1)\beta+B_c}\right]^{d-1}
\frac{(q-1)\beta+1-(d+1)(1-B_c)x}{(q-1)\beta+B_c}.
\end{equation*}
For $x\in[\beta,(q-1)\beta]$, we have $(d+1)(1-B_c)x=qx=(q-1)x+x>
(q-1)\beta+1$, where in the last inequality we used that
$\beta>1$. It follows that $F(x)$ is indeed decreasing and thus $F(\alpha)\leq
F(\beta)$.

To prove \eqref{eq:mainobstacle}, it thus suffices to argue
that for $\beta>1$ it holds
\begin{equation}
d^2(1-B_c)^2F(\beta)^{(d-1)/d}\leq
\big[(q-1)\beta+B_c\big]
\bigg[(q-1)\Big(1-\frac{(1-B_c)(\beta-1)}{(q-1)\beta+B_c}\Big)^d+B_c\bigg].\label{eq:mainobstacle2}
\end{equation}
Note that $q-1+B_c=d(1-B_c)$ so that the inequality is tight for
$\beta=1$. By the weighted AM-GM inequality on $A^d$ and $1$ with
weights $(q-1)$ and $B_c$ respectively, we obtain
\begin{equation*}
(q-1)A^d+B_c\geq
(q-1+B_c)A^{d(q-1)/(q-1+B_c)}=d(1-B_c)A^{(q-1)(d+1)/q}.
\end{equation*}
We use this for
$A=\Big(1-\frac{(1-B_c)(\beta-1)}{(q-1)\beta+B_c}\Big)^d$ so that,
after simplifications, it suffices to show that
\begin{equation*}
d(1-B_c)\beta^{(d-1)/d}\leq \big[(q-1)\beta+B_c\big]
\left[1-\frac{(1-B_c)(\beta-1)}{(q-1)\beta+B_c}\right]^{-(d+1-2q)/q}.
\end{equation*}
This can further be massaged into
\begin{equation*}
G(\beta):=\beta^{(d-1)/d}\big[(q-1)\beta+B_c\big]^{-(d+1-q)/q}\big[(q-2+B_c)\beta+1\big]^{(d+1-2q)/q}\leq
\frac{1}{d(1-B_c)}.
\end{equation*}
Once again, note that the inequality holds at equality for
$\beta=1$, so it suffices to prove that $G'(\beta)<0$ for $\beta>1$. This has
nothing special, apart from tedious, but otherwise
straightforward, calculations. We include the details briefly. Differentiating $\ln G(\beta)$, we obtain
\begin{equation*}
\frac{G'(\beta)}{G(\beta)}=\frac{(d-1)}{d\beta}-\frac{(d+1-q)(q-1)}{q\big[(q-1)\beta+B_c\big]}+\frac{(d+1-2q)(q-2+B_c)}{q\big[(q-2+B_c)\beta+1\big]}.
\end{equation*}
By clearing denominators, it suffices to check that the following second order polynomial $p(\beta)$ is negative whenever $\beta>1$:
\begin{align*}
p(\beta)&:=(d-1)q[(q-1)\beta+B_c\big]\big[(q-2+B_c)\beta+1\big]-(d+1-q)(q-1)d\beta\big[(q-1)\beta+B_c\big]\\
&\hskip6cm+(d+1-2q)(q-2+B_c)\beta\big[(q-1)\beta+B_c\big].
\end{align*}
Using again that $q-1+B_c=d(1-B_c)$ it is easy to verify that $p(1)=0$. The factorization of $p(\beta)$ (using the value of $B_c$) is given by
\begin{equation*}
p(\beta)=-\frac{q(\beta-1)\big[\beta \big(d(q-1)^2-(q-1)\big)+ d(d-q) + q-1\big]}{d+1},
\end{equation*}
which is obviously negative for $\beta>1$, whenever $d\geq q\geq 2$.
\end{proof}

\begin{proof}[Proof of Lemma \ref{lem:hincreasing}] We compute
\begin{align*}
h'(\beta)&=\frac{1}{d}\beta^{-(d-1)/d}\bigg[(q-1)\Big(1-\frac{(1-B_c)(\beta-1)}{(q-1)\beta+B_c}\Big)^d+B_c\bigg]\\
&\hskip 1cm
-d\left[1-\frac{(1-B_c)(\beta-1)}{(q-1)\beta+B_c}\right]^{d-1}\frac{(1-B_c)(q-1+B_c)[(q-1)\beta^{1/d}-(q-2+B_c)]}{[(q-1)\beta+B_c]^2}.
\end{align*}
Thus, to prove $h'(\beta)>0$ it suffices to check (using
$q-1+B_c=d(1-B_c)$ and the function $F$ defined in Lemma~\ref{lem:gdecreasing})
\begin{equation*}
\begin{aligned}
d^3(1-B_c)^2&F(\beta)^{(d-1)/d}\leq\\
&\leq\frac{\big[(q-1)\beta+B_c\big]^2}{(q-1)\beta^{1/d}-(q-2+B_c)}
\bigg[(q-1)\Big(1-\frac{(1-B_c)(\beta-1)}{(q-1)\beta+B_c}\Big)^d+B_c\bigg],
\end{aligned}
\end{equation*}
This is similar to \eqref{eq:mainobstacle2} and in fact follows
from \eqref{eq:mainobstacle2}, once we prove that
\begin{equation*}
\frac{(q-1)\beta^{1/d}-(q-2+B_c)}{(q-1)\beta+B_c}\leq\frac{1}{d}.
\end{equation*}
To see the last inequality, observe that $\beta+d-1\geq
d\beta^{1/d}$ as a consequence of the weighted AM-GM inequality
(or otherwise). Hence,
\begin{equation*}
\frac{(q-1)\beta^{1/d}-(q-2+B_c)}{(q-1)\beta+B_c}\leq\frac{(q-1)\beta+(d-1)(q-1)-d(q-2+B_c)}{d\big[(q-1)\beta+B_c\big]}=\frac{1}{d},
\end{equation*}
completing the proof.
\end{proof}

\end{document}